\documentclass[12pt]{article}
\usepackage{amsmath}    
\usepackage{amssymb}    
\usepackage{amsthm}     
\usepackage{graphicx}   
\usepackage{caption}    
\usepackage{xcolor}
\usepackage{booktabs}
\usepackage{multirow}
\usepackage{url}
\usepackage{comment}
\usepackage{enumitem}
\usepackage[authoryear,round]{natbib}
\usepackage{multibib}
\usepackage{placeins}
\usepackage{hyperref}
\usepackage{bbm}
\usepackage{subcaption}
\usepackage{tikz}

\usepackage{algorithm}
\usepackage{algpseudocode}
\usepackage{float}

\newcommand{\E}{\mathbb{E}}

\newcommand{\adj}{\operatorname{adj}}

\newcites{main}{Main References}
\newcites{appendix}{Appendix References}

\makeatletter
\theoremstyle{plain}
\newtheorem{assumption}{\protect\assumptionname}
\theoremstyle{plain}

\theoremstyle{remark}

\newtheorem{remark}{\protect\remarkname}
\theoremstyle{plain}
\newtheorem{thm}{\protect\theoremname}
\theoremstyle{plain}
\newtheorem{cor}{\protect\corollaryname}
\theoremstyle{plain}

\theoremstyle{remark}
\newtheorem{example}{\protect\examplename}

\providecommand{\assumptionname}{Assumption}
\providecommand{\corollaryname}{Corollary}
\providecommand{\lemmaname}{Lemma}

\providecommand{\propositionname}{Proposition}
\providecommand{\remarkname}{Remark}
\providecommand{\theoremname}{Theorem}
\providecommand{\examplename}{Example}

\newcommand{\ii}{\imath}
\newcommand{\R}{\mathbb R}
\newcommand{\tOne}[1]{\widetilde{#1}}
\newcommand{\tTwo}[1]{\widetilde{\widetilde{#1}}}

\DeclareMathOperator{\rank}{rank}

\begin{document}

\def\spacingset#1{\renewcommand{\baselinestretch}{#1}\small\normalsize}
\spacingset{1}

\title{\bf Correlated Random Coefficient Distributions in Linear Panel Models\thanks{We thank Hide Ichimura, Hiro Kasahara, Laura Liu, Chris Muris, and Kevin Song for discussions and suggestions. This paper was presented at the California Econometrics Conference at UC Davis, the Conference in Honor of Bo Honor\'e at Princeton, and seminars at Ohio State, UCSD, San Diego, the University of Arizona, and UBC. Botosaru gratefully acknowledges financial support from the Canada Research Chairs Program.}}

\author{
Irene Botosaru\thanks{McMaster University, Department of Economics. Email: {\tt botosari@mcmaster.ca}}
\and
James L. Powell\thanks{University of Arizona, Department of Economics. Email: {\tt jlpowell@arizona.edu}}
}

\date{May 16, 2026}

\maketitle

\bigskip
\begin{abstract}
\noindent

We consider a static linear panel model with both correlated and uncorrelated random coefficients, where the former can depend arbitrarily on observable regressors while the latter are independent of them. We provide sufficient conditions for identification of the distributions of the random coefficients without imposing restrictions on the time-series structure of the error terms in short panels. Our framework applies to regular and irregular designs. The distribution of the correlated coefficients follows via a deconvolution argument. In irregular designs, identification relies on a stayer-based argument exploiting near-singular realizations of the regressor matrix.

We develop a two-step minimum distance sieve estimator, with tuning parameters selected by cross-validation. In an application to calorie-expenditure elasticities using data from the randomized evaluation of a conditional cash transfer program, we interpret the estimated distributions by program status as distributions of regime-specific structural calorie-expenditure elasticities. The estimated densities themselves reveal substantial heterogeneity in household-specific elasticities, with nontrivial mass concentrated near zero and a non-negligible share of negative realizations. This heterogeneity implies that responses to income or expenditure changes are not uniformly positive and vary widely across households. Taken together, these features support a framework in which households adjust along both quantity and quality margins, rather than conforming to a homogeneous Engel-curve response.
\end{abstract}

\vfill

\newpage
\spacingset{1.9}

\section{Introduction}
\label{sec:intro}

Heterogeneity in individual behavior is pervasive in microeconometric applications,
and panel data are a natural way to accommodate it. A leading framework is the
correlated random coefficient (CRC) model (\citealt{Chamberlain82}), in which slope
coefficients are heterogeneous across units and may covary with the regressors. In
this setting, a natural target is the average partial effect (APE), given by the
population mean of the random coefficient vector. \cite{Chamberlain1992} studied
identification of the APE in the \emph{regular} design, where the number of
effective time periods $T$ exceeds the number of correlated coefficients $p$.
\cite{GrahamPowell2012} showed that the APE remains identified in the
\emph{irregular} design, where $T=p$, by exploiting the distinct identifying
content of \emph{movers} and \emph{stayers}.  \cite{ArellanoBonhomme2012} focus on the regular design and show various moments of the correlated random coefficients, and even their density, can be identified under ARMA-type restrictions on the error process. In this paper, we derive sufficient conditions under which the density of the correlated random coefficients is identified in both regular and irregular designs, without restricting the time-series properties of the error terms.

We consider models with the following outcome equation. For
$i=1,\dots,n$,
\begin{align}
\label{eq:outcome_stack_t}
Y_i &= X_i\beta_i + W_iD_i, \qquad
Y_i \in \mathbb R^T,\quad X_i \in \mathbb R^{T\times p},\quad W_i \in \mathbb R^{T\times d},\quad T<\infty,
\end{align}
where the correlated random coefficients
$\beta_i\in \mathbb R^p$ may depend arbitrarily on the observed regressors $(X_i,W_i)$, while the uncorrelated components $D_i\in\mathbb R^d$ are independent of $(\beta_i,X_i,W_i)$.\footnote{As explained in the next section, $D_i\in\mathbb R^d$ collects all remaining unobserved 
components, which includes both time-invariant random coefficients and time-varying idiosyncratic errors. The assumption can be weakened to conditional independence, e.g., \(\beta_i\perp D_i\mid (X_i,W_i)\) and \(D_i\perp X_i\mid W_i\). Then the results would be conditional on \(W_i\). We do not pursue this extension.} The coefficient $\beta_i$ represents the unit-level structural partial effect of $X_i$. Our parameter of interest is $f_\beta$, the cross-sectional density of $\beta_i$. Ours is a short-$T$ setting.

Identification proceeds in two steps. In the first step, we identify the characteristic function of $D_i$ by using a transformation, or annihilator matrix, that eliminates the $\beta_i$ component. The transformation depends on the relation between $T$ and $p$. In the regular design ($T > p$), there are multiple annihilation matrices; projection onto the orthogonal complement of the column space of $X_i$ achieves annihilation for all 
observations. In the irregular design ($T = p$), the annihilation matrix is the adjugate matrix of $X_i$; the $\beta_i$ component is then eliminated for \emph{stayers}, i.e. units with $\det(X_i)=0$. The adjugate--determinant algebra used in the irregular design shares a common origin with \cite{GrahamPowell2012}, but it is
embedded here in an argument based on characteristic functions.

In the second step, we apply another transformation to the outcome equation that isolates $\beta_i$.
We then recover $f_\beta$ by deconvolution, exploiting the independence restriction on $D_i$. Notably, we impose no restrictions on the time-series properties of the error terms.

We propose a sieve minimum-distance estimator for the density of the correlated
coefficients and study its finite-sample properties via Monte Carlo simulations. The estimators we propose specialize to the scalar irregular design $(T,p)=(1,1)$ and to the low-dimensional regular design $(T,p)=(2,1)$, with $W_i$ being the identity matrix; these are the configurations that arise in our application. 
Our estimator is a two-step procedure. In the first step, which is
design-specific, we construct a nonparametric estimator of the 
unconditional characteristic function of the correlated coefficients. This involves estimating a conditional 
characteristic-function ratio at each observation and then averaging 
over the sample to integrate out the covariates. In the irregular design, 
we partition the sample into stayers, used to estimate the characteristic 
function of $D_i$, and movers, used to estimate a trimmed ratio of 
characteristic functions, which is then averaged across the movers. In the
regular design, no mover--stayer partition is needed: the characteristic 
function of $D_i$ is estimated by smoothing over normalized directions on 
the unit sphere, and the ratio is averaged over the full sample. In the second step,
which is common across designs, we approximate the density by a finite-dimensional
expansion in orthonormal Hermite functions and minimize a weighted distance between
the Fourier transform of this expansion and the first-step characteristic function
estimate, subject to a unit-mass normalization. The Hermite basis is a natural
choice: Hermite functions are eigenfunctions of the Fourier transform, so the
criterion is quadratic in the sieve coefficients and admits a closed-form
constrained minimizer. We provide a heuristic rate decomposition that
guides the choice of tuning parameters, a cross-validation procedure for selection of the bandwidth and sieve dimension, and Monte Carlo evidence on finite-sample performance.\footnote{Large-sample theory for the two-stage estimator, including data-driven tuning and bootstrap validity, is left for future work.}

We revisit the application in \citet{GrahamPowell2012}, which uses panel data
from rural Nicaragua collected as part of the evaluation of the conditional cash transfer program Red de Protección Social (RPS). In this application, the correlated random coefficient is the household-specific calorie--expenditure elasticity. We estimate its full cross-sectional density. The estimated densities reveal substantial heterogeneity: there is
considerable dispersion, nontrivial mass near zero, and a non-negligible share of negative elasticities. These findings help reconcile mixed estimates in the
calorie-demand literature, where average elasticities range from values close to zero \citep{BehrmanDeolalikar1987,BouisHaddad1992,Ravallion1990} to values around 0.3--0.5 \citep{SubramanianDeaton1996}. They are also consistent with the view that households differ in how additional resources are allocated across quantity, quality, and nonfood margins
\citep{Deaton1997,JensenMiller2008,Skoufias2011,DeatonDreze2009}.

The experimental design of RPS gives the distributional estimates a causal interpretation. Since assignment was randomized at the village level and take-up was high, differences between the estimated elasticity distributions for recipient and nonrecipient households can be interpreted as causal
contrasts between regime-specific densities, subject to the maintained CRC assumptions. The three survey
waves also allow us to compare scalar irregular estimates separately for 2000--2001 and 2001--2002. Under the baseline model, which imposes a
time-invariant household-specific elasticity, these estimates should recover the same density. The leftward shift we find among RPS recipients is therefore suggestive evidence that elasticities may vary with cumulative program exposure, while the regular stacked estimator should be viewed as a pooled
benchmark under the common-\(\beta_i\) restriction.

Our paper contributes to the literature on CRC models in short panels. Existing work primarily focuses on low-dimensional functionals of the heterogeneity distribution. For example, \citet{GrahamPowell2012} study average partial effects in irregular designs, \citet{GrahamHahnPoirierPowell2018} study quantile effects under comonotonicity restrictions in both regular and irregular designs, while \citet{SasakiUra2026} study inference for average partial effects in irregular designs with slow movers. In regular designs, \citet{Verdier2020} considers identification of average
treatment effects for stayers, recovering conditional means via linear
extrapolation; \cite{Laage2024} studies average effects in the presence of
time-varying endogeneity through control variables; while \citet{MurisWacker2026}
analyze interaction effects, that is, how the correlated coefficient varies
with observable regressors. We instead identify and estimate the marginal distribution of $\beta_i$ in both regular and irregular designs. This allows us to recover higher-order moments and shape features of the heterogeneity distribution, including dispersion, skewness, and multimodality, which are not accessible through mean- or quantile-based approaches.

The closest related result is \citet{ArellanoBonhomme2012}, who also establish identification of the distribution of correlated random coefficients in regular designs. Their approach exploits restrictions on the time-series properties of the error terms. In contrast, we allow for unrestricted serial dependence in the error terms and instead exploit restrictions on the dependence between the errors and the observed regressors. This alternative source of variation delivers identification in both regular and irregular designs.

Building on \citet{ArellanoBonhomme2012}, \citet{BotosaruLiu2025} uses empirical Bayes methods to recover posterior means of the correlated random coefficients. In contrast, we recover the marginal distribution of $\beta_i$ without  imposing parametric restrictions, using a deconvolution-based estimator rather than empirical Bayes methods.

In a related but distinct framework, \citet{HoderleinWhite2012} study local
average responses for stayers in nonseparable panel models under a time-invariant
structural function. \citet{Chernozhukovetal2015} extend
this to quantile effects for stayers. We maintain linearity of the outcome equation and identify the distribution
of $\beta_i$ for the entire population.

Our estimation strategy builds on tools developed for random coefficient density estimation in cross-sectional settings. \citet{HoderleinKMammen} propose a Radon-transform-based estimator, while \citet{Breunig2021} develops a sieve minimum-distance estimator based on Hermite functions.  We also use a sieve-based weighted minimum-distance criterion, where the criterion function involves a ratio of characteristic functions evaluated at data-dependent frequencies.

More broadly, our approach is related to the deconvolution literature based on characteristic functions. In particular, \cite{KatoSasakiUra2021} reformulate Kotlarski's identity as a system of complex-valued moment restrictions and conduct inference via test inversion. This strategy yields inference that is robust to potential failure of the completeness condition. Our panel setting also implies moment restrictions derived from a characteristic-function factorization. Unlike their repeated-measurement model, which targets an unconditional characteristic function, our setting involves a conditional characteristic function indexed by the regressors. Moreover, the moments depend on a generated first-step estimator evaluated at  data-dependent frequencies. These features, discussed in detail in Remark~\ref{rem:inference}, lead us to base estimation on an explicit regularized deconvolution step rather than moment inversion.

Finally, our results connect to applications with continuously measured treatments. Recent work studies average effects in such settings using difference-in-differences (DID), both when stayers are available \citep{deChaisemartinetal2024b} and when all units change treatments between periods \citep{deChaisemartinetal2024a}; see also \citet{CallawayGBSantAnna} for a framework for continuous DID under parallel trends. Distributional effects have also been analyzed under unconfoundedness, e.g., \citet{GalvaoWang2015}, \citet{CallawayHuang2020}.

The remainder of the paper is organized as follows. Section~\ref{sec:model}
introduces the model and identifying assumptions. Section~\ref{sec:identification}
presents the identification results for both regular and irregular designs, with
worked examples. Section~\ref{sec:estimation} describes the sieve minimum-distance
estimator. Section~\ref{sec:simulations} reports
Monte Carlo evidence on the finite-sample performance of the estimator.
Section~\ref{sec:application} applies the methods to estimate the distribution of
calorie demand elasticities using the Nicaraguan household panel of Graham and Powell
(2012).

\subsubsection*{Notation} 

For a matrix \(M\), \(\adj(M)\) denotes the adjugate matrix and \(\det(M)\) its determinant. For a vector $v$, $\|v\|$ denotes the Euclidean norm, for a matrix $M$, \(\|M\|\) denotes the Frobenius norm. The joint characteristic function of an arbitrary \(q\times 1\) random vector \(Z=(Z_1,\dots,Z_q)'\) is
\[
\varphi_Z(s)\equiv \E\!\left[e^{\ii s'Z}\right],\qquad s\in\mathbb R^q, \qquad \ii=\sqrt{-1}.
\]
Note that we define characteristic functions on column vectors.

Whenever we condition on an event of probability zero, the conditional object is understood as a regular conditional distribution or conditional expectation whenever it exists. In particular, conditioning on the probability zero event \(\det(X_i)=0\) in the irregular design is interpreted through the limit as \(\det(X_i)\to 0\), under appropriate continuity assumptions stated in the paper.

For each $(s,u)\in\mathbb{S}^{d-1}\times\mathbb{R}$, let $\sigma_{s,u}$ denote the
$(d-1)$-dimensional Hausdorff measure on the hyperplane
$\{\delta\in\mathbb{R}^d:s'\delta=u\}$. Let $\mathcal{R}$ denote
the Radon transform:
\[
  \mathcal{R}g(s,u)
\equiv
\int_{\{\delta:\,s'\delta=u\}} g(\delta)\,d\sigma_{s,u}(\delta),
\]
for any integrable $g:\mathbb{R}^d\to\mathbb{R}$. 

\section{Model}
\label{sec:model}

The outcome equation \eqref{eq:outcome_stack_t} is our point of departure. We think of it as having been generated by, for example,
\begin{align}
\label{eq:outcome_levels}
    y_{it}
    &= \alpha_{i}+x^\prime_{it}\beta_{i}+w^\prime_{it}\delta_{i}+u_{it},\quad t=0,\dots,T,
\end{align}
where $x_{it}\in\mathbb{R}^p$ and $w_{it}\in\mathbb{R}^q$ are observed
covariates; $\beta_i\in\mathbb{R}^p$ and $\delta_i\in\mathbb{R}^q$ are vectors
of unit-specific random coefficients; and $u_{it}$ is an idiosyncratic error.
The random coefficients $(\alpha_i,\beta_i')$ may be arbitrarily correlated with
the covariates $\{x_{it},w_{it}\}_{t=0}^{T}$. The random coefficients
$\delta_i$ are independent of the covariates. Equation \eqref{eq:outcome_stack_t} is obtained by first-differencing \eqref{eq:outcome_levels} across adjacent periods to eliminate $\alpha_i$, stacking over the effective time periods $t=1,\dots,T$, and letting $Y_{it}\equiv y_{it}-y_{i,t-1}$, $X_{it}\equiv x_{it}-x_{i,t-1}$,
$W_{it}\equiv w_{it}-w_{i,t-1}$, $U_{it}\equiv u_{it}-u_{i,t-1}$, and
$$W_i \equiv
\bigl(\,\Delta{w}_i \;\; I_T\bigr), \; \Delta{w}_i \equiv (W_{i1},\ldots,W_{iT})' \in \mathbb{R}^{T\times q},\; D_i \equiv (\delta_i',\, U_{i1},\ldots,U_{iT})'\in\mathbb R^d, \; d = q + T.$$

\begin{assumption}
\label{assm:A1}
    (i) \(\beta_i\) is absolutely continuous with respect to Lebesgue measure, with density \(f_\beta(b)\), \(b\in\mathbb{R}^p\).  
    (ii) \(D_i\) is independent of \((\beta_i,X_i,W_i)\) drawn from distribution \(F_D(s)\), \(s\in\mathbb{R}^d\).
\end{assumption}

Assumption \ref{assm:A1}(i) ensures that the correlated random coefficients \(\beta_i\) admit a density function. This can be relaxed if interest is in the distribution function or the moments of $\beta_i$ instead.\footnote{Our identification strategy recovers the characteristic function of \(\beta_i\), which can be inverted via Fourier inversion to obtain either the density function or via L\'evy inversion to obtain the cumulative distribution function; see \cite{Lukacs1970}.} Assumption \ref{assm:A1}(ii) allows arbitrary correlation between $\beta_i$ and $(X_i,W_i)$, while restricting the joint distribution of $(\beta_i,D_i,X_i,W_i)$. In particular, the assumption has the following implications. First, $D_i$ is statistically independent of $(X_i,W_i)$.\footnote{Heteroskedasticity \(\operatorname{Var}(D_i\mid W_i)\) can be allowed by weakening the statistical independence assumption. Identification results would then be stated conditional on $W_i$, which we do not pursue.}  Second, $\beta_i$ and $D_i$ are conditionally independent given $(X_i,W_i)$. These implications allow us to, first, identify the characteristic function of $D_i$, and then that of $\beta_i$ via a deconvolution argument.

Assumption \ref{assm:A1}(ii) has observable implications when $q\neq0$. Let $x\in\mathbb{R}^{T\times p}$ and
$w\in\mathbb{R}^{T\times d}$ denote matrix realizations of $X_i,\;W_i$, respectively. For finite second moments, for almost every \((x,w)\) in the support of \((X_i,W_i)\), it follows that
\[
\operatorname{Var}(Y_i\mid X_i=x,W_i=w)
=
x\,\operatorname{Var}(\beta_i\mid x,w)\,x'
+
w\,\operatorname{Var}(D_i)\,w',
\]
where we used $\operatorname{Var}(D_i\mid X_i,W_i)=\operatorname{Var}(D_i),$ and $\operatorname{Cov}(\beta_i,D_i\mid X_i,W_i)=0.$ Therefore, the difference
\[
\operatorname{Var}(Y_i\mid x,w)
-
w\,\operatorname{Var}(D_i)\,w'\in\mathbb R^{T\times T}
\]
must be positive semidefinite for almost every \((x,w)\). A negative eigenvalue of the associated matrix provides evidence against Assumption \ref{assm:A1}(ii).

\begin{remark}[Interpretation of the correlated random coefficients]
The outcome equation \eqref{eq:outcome_stack_t} admits the structural potential-outcome interpretation
\[
  Y_i(x,w)=x\beta_i+wD_i,
\]
for counterfactual regressor histories $x\in\mathbb R^{T\times p},\; w\in\mathbb R^{T\times d}$, in the relevant support. Under this interpretation, \(\beta_i\) is a unit-level partial effect with respect to the \(x\)-coordinates. In particular, for two counterfactual histories \(x_1\) and
\(x_2\), holding \(w\) fixed,
\[
  Y_i(x_2,w)-Y_i(x_1,w)
  =
  (x_2-x_1)\beta_i .
\]
In the scalar case, this reduces to
\[
  \frac{\partial Y_i(x,w)}{\partial x}=\beta_i .
\]
Thus, in program-evaluation language, \(\beta_i\) is the unit-level causal
slope of the dose-response function with respect to the continuous regressor
\(x\).

The word ``correlated'' means that realized regressor histories may be
informative about the unit-level structural slope. Formally, although
\(\beta_i\) is fixed for unit \(i\) and does not vary with the counterfactual
values of \(x\) or \(w\), its conditional distribution may vary with the
observed history:
\[
  f_{\beta\mid X,W}(\cdot\mid x,w)\neq f_\beta(\cdot).
\]
Equivalently, the model allows selection on heterogeneous slopes: units
observed at different realized histories may have different distributions of
\(\beta_i\). This selection changes the composition of latent types across
histories, but \(\beta_i\) does not change when \(x\) is counterfactually changed
for the same unit. See \citet{GrahamPowell2012}. 
\end{remark}

\section{Identification}
\label{sec:identification}

The parameter of interest is $f_\beta$, the density of the correlated random coefficients. We sketch the logic behind our identification strategy below; the main result is in Section \ref{sec:densities}.

Identification of $f_\beta$ proceeds in two steps. In the first step, we apply a transformation to \eqref{eq:outcome_stack_t} that allows us to remove the $\beta_i$ component, which then yields the characteristic function of $D_i$. In the second
step, we apply a different transformation to \eqref{eq:outcome_stack_t} that isolates $\beta_i$. Given the identification of the characteristic function of $D_i$, this then allows us to recover
$f_\beta$ by deconvolution. The choice of transformations depends on the relation
between $T$ and the number of correlated random
coefficients $p$. Concrete examples are given in Section \eqref{sec:transformations}. 

Let $\tau_1(X_i)\in\mathbb R^{T\times T}$ denote a measurable
transformation. In the regular design $T>p$, it satisfies
\begin{equation}
\tau_1(X_i)X_i=0,
\label{tau_1}
\end{equation}
so that, premultiplying \eqref{eq:outcome_stack_t} by $\tau_1(X_i)$, yields:
\begin{equation}
\tau_1(X_i)Y_i=\tau_1(X_i)W_iD_i,
\qquad\text{equivalently}\qquad
\tOne Y_i=\tOne W_iD_i.
\label{tau1_W}
\end{equation}
We use \eqref{tau1_W} to identify the characteristic function of $D_i$ in the regular design.

In the irregular design $T=p$, we take
$\tau_1(X_i)=\adj(X_i)$. Using the identity
\begin{equation}
\label{eq:adjugate_rel}
\tau_1(X_i)X_i=\det(X_i)I_p,
\end{equation}
we premultiply \eqref{eq:outcome_stack_t} by $\tau_1(X_i)$ to obtain:
\begin{equation}
\label{transformed_adj}
\tOne Y_i=\det(X_i)\beta_i+\tOne W_iD_i.
\end{equation}
The $\beta_i$ component vanishes on the stayer set $\{\det(X_i)=0\}$. We use \eqref{transformed_adj} to recover the characteristic function of $D_i$ in the irregular design, for observations with $\det(X_i)=0$.

Let $\tau_2(X_i)\in\mathbb R^{p \times T}$ be any measurable mapping satisfying
\begin{equation}
\tau_2(X_i)X_i=I_p .
\label{tau_2}
\end{equation}
Premultiplying \eqref{eq:outcome_stack_t} by $\tau_2(X_i)$ gives
\begin{equation}
\tau_2(X_i)Y_i
=
\beta_i+\tau_2(X_i)W_iD_i,
\qquad\text{equivalently}\qquad
\tTwo Y_i=\beta_i+\tTwo W_iD_i.
\label{tau2_YW}
\end{equation}
Given the characteristic function of $D_i$ from the first step, we use \eqref{tau2_YW} to identify the density of
$\beta_i$ via a deconvolution argument.\footnote{%
The role of the transformations is not algebraic recovery of the random coefficients. Generally, neither \eqref{tau1_W} nor \eqref{tau2_YW} identifies $D_i$ or $\beta_i$. For example, 
$\widetilde W_i=\tau_1(X_i)W_i$ contains at most
$\operatorname{rank}(\tau_1(X_i))$ linearly independent equations for $d$ unknown components. Since
$\operatorname{rank}(\widetilde W_i)\le \operatorname{rank}(\tau_1(X_i))<d$
in the designs of interest, the system is underdetermined for each $i$.
Similarly,
$\widetilde{\widetilde Y}_i=\beta_i+\widetilde{\widetilde W}_iD_i$
provides $p$ equations in the $p+d$ latent components
$(\beta_i,D_i)$ and is likewise underdetermined at the unit level.
}

\subsection{Identification of the density function \texorpdfstring{$f_\beta$}{f\_beta}}
\label{sec:densities}

Throughout the analysis, we maintain the following assumptions, where $\tOne Y_i\in\mathbb R^T,\;\tOne W_i\in\mathbb R^{T\times d}$ are as defined in \eqref{tau1_W} and $\tTwo Y_i\in\mathbb R^p,\;\tTwo W_i\in\mathbb R^{p\times d}$ are as defined in \eqref{tau2_YW}. Note that $d\geq 1$ in the irregular design, while $d\geq 2$ in the regular design.\footnote{Consider the regular design with $d=1$. Then $T+q=1$, so it must be that $q=0$ and $T=1$, which implies that $p=0$.}

\begin{assumption}\label{assm:A1_prime}
The characteristic function $\varphi_D$ of $D_i$ is nowhere zero on $\mathbb{R}^{d}$.
\end{assumption}

This is a standard assumption in deconvolution, ensuring that division in the Fourier domain is well defined. It rules out, e.g., uniform and truncated normal distributions.

\begin{assumption}
\label{assm:A1_double}
Let $\tau_1(X_i)$ be as defined in \eqref{tau_1} or \eqref{eq:adjugate_rel} depending on the design,
and let $\mu_i = \mu(X_i,W_i) \in \mathbb{R}^T$ be a known measurable
function. Define
\[
  \lambda_i \equiv \widetilde{W}_i'\mu_i \in \mathbb{R}^{d},
  \qquad
  S_i \equiv \frac{\lambda_i}{\|\lambda_i\|}
  \in \mathbb{S}^{d-1}.
\]

\noindent(i) \emph{Regular design ($T>p$).} $\Pr(\lambda_i=0)=0$ and $\operatorname{supp}(S_i) = \mathbb{S}^{d-1}$.

\medskip
\noindent(ii) \emph{Irregular design ($T=p$).} When $d\geq1$,
$\Pr(\lambda_i=0\mid\det(X_i)=0)=0$; additionally, $\operatorname{supp}(S_i\mid\det(X_i)=0) = \mathbb{S}^{d-1}$ when $d\geq 2$.
\end{assumption}


Assumption~\ref{assm:A1_double} is the support condition that makes the
first-step characteristic-function argument informative in all relevant directions. Whenever the annihilated equation satisfies $\widetilde Y_i=\widetilde W_iD_i,$ we have, for any measurable projection vector \(\mu_i\),
\[
  \mu_i'\widetilde Y_i
  =
  \mu_i'\widetilde W_iD_i
  =
  \lambda_i'D_i .
\]
Since \(\lambda_i\) is measurable with respect to \((X_i,W_i)\) and
\(D_i\) is independent of \((X_i,W_i)\), it follows that, for every
\(v\in\mathbb R\),
\[
  \mathbb E\!\left[
    \exp\{ \iota v\mu_i'\widetilde Y_i\}
    \,\middle|\,
    \lambda_i=\lambda
  \right]
  =
  \varphi_D(v\lambda),
\]
with the conditioning interpreted conditionally on the stayer event in the
irregular design. Thus the first step identifies \(\varphi_D\) on the following set
\[
  \mathcal C_1
  :=
  \{v\lambda:\ v\in\mathbb R,\ 
  \lambda\in\operatorname{supp}(\lambda_i)\}.
\]
Equivalently, it identifies \(\varphi_D\) along the rays generated by the
normalized directions \(S_i\). Assumption~\ref{assm:A1_double} requires these
directions to be rich enough that
\[
  \overline{\operatorname{supp}(S_i)}
  =
  \mathbb S^{d-1}
\]
in the regular design, and
\[
  \overline{\operatorname{supp}(S_i\mid \det(X_i)=0)}
  =
  \mathbb S^{d-1}
\]
in the irregular design. Since characteristic functions are continuous, knowledge
of \(\varphi_D\) on a dense collection of rays identifies \(\varphi_D\) on all of
\(\mathbb R^d\). This is the sense in which Assumption~\ref{assm:A1_double}
ensures compatibility between the frequencies at which \(\varphi_D\) is recovered
in the first step and the frequencies at which it is needed in the deconvolution
step.

In the irregular design, \(\tau_1(X_i)=\adj(X_i)\). On the stayer set
\(\{\det(X_i)=0\}\), and under the generic rank condition
\[
  \rank(X_i)=p-1,
\]
the adjugate matrix has rank one. Hence $\widetilde W_i$ has row space of dimension at most one. Each stayer therefore contributes at most
one direction, up to scale. If a fixed row of \(\widetilde W_i\) is nonzero on the
relevant support, one may use that row. For instance, if the first row is nonzero
on the stayer set, one may take \(\mu_i=e_1\), giving
\[
  \lambda_i
  =
  \widetilde W_i'e_1
  =
  W_i'\adj(X_i)'e_1 .
\]
A sufficient support condition is then
\[
  \overline{\operatorname{supp}\left(
  \frac{\lambda_i}
       {\|\lambda_i\|}
  \,\middle|\,
  \det(X_i)=0
  \right)}
  =
  \mathbb S^{d-1}.
\]
If no fixed row is nonzero on the relevant support, one can instead define a
measurable row-selection rule. Thus, in the irregular design, directional
variation must come from the conditional distribution of \((X_i,W_i)\) on, or
near, the singular set.

In the regular design, \(\tau_1(X_i)X_i=0\), and the annihilated equation $\widetilde Y_i=\widetilde W_iD_i$
holds for every observation. If \(\rank(X_i)=p\), the standard annihilator has
rank \(T-p\), so
\[
  \rank(\widetilde W_i)\leq T-p.
\]
When the row space of \(\widetilde W_i\) has dimension greater than one, a single
observation generates a family of possible directions, and the projection vector
\(\mu_i\) selects one direction from that family. When the row space is
one-dimensional, each observation contributes only one direction up to scale, and
the required richness must come from cross-sectional variation in \((X_i,W_i)\).

A fixed projection vector need not exploit all available directional variation. For
example, if \(\mu_i\equiv e_1\), then $  \lambda_i=\widetilde W_i'e_1$ uses only one fixed linear combination of the rows of \(\widetilde W_i\). The resulting directions may lie in a strict subset of \(\mathbb S^{d-1}\), even if
the collection of row spaces of \(\widetilde W_i\) is itself rich. Allowing
\(\mu_i=\mu(X_i,W_i)\) to vary with the data enlarges the attainable set of
directions by selecting different linear combinations of the rows of
\(\widetilde W_i\). Assumption~\ref{assm:A1_double} holds when the normalized
selected directions are dense in the unit sphere.

\begin{remark}[Support condition for \(D_i\)]
Assumption~\ref{assm:A1_double} is sufficient for Theorem~\ref{main_result}.
It identifies \(\varphi_D\) along every ray generated by the support of \(S_i\);
continuity of characteristic functions then extends identification to all of
\(\mathbb R^d\). This is enough for the subsequent deconvolution step that
identifies \(f_\beta\).

If \(D_i\) admits a density \(f_D\), then \(f_D\) can also be identified through a
Radon-transform representation. That route requires a stronger condition, e.g., the
directions \(S_i\) must have a density bounded away from zero with respect to surface measure, see Corollary \ref{cor_fD}.
This condition is used only for identifying \(f_D\), not for the main
identification result for \(f_\beta\).

The full-support condition can be weakened if \(\varphi_D\) is real analytic on
\(\mathbb R^d\). In that case, knowledge of \(\varphi_D\) on any nonempty open
subset of \(\mathbb R^d\) determines \(\varphi_D\) globally. Hence it is enough
that the cone $\mathcal C_1$ contain a nonempty open subset of \(\mathbb R^d\). 
\end{remark}

\begin{remark}[Discrete regressors]
When \((X_i,W_i)\) have discrete support, identification depends on whether the
first-step directions cover the second-step frequencies required for
deconvolution. If the support is finite and \(d\geq 2\), then, for any
measurable choice of \(\mu_i\), the set of normalized directions \(S_i\) is finite
and therefore cannot be dense in \(\mathbb S^{d-1}\). Hence the first step
identifies \(\varphi_D\) only on a finite union of rays and does not identify the
full distribution of \(D_i\) without additional restrictions.

This does not mechanically rule out identification of \(f_\beta\). The
deconvolution step requires \(\varphi_D\) only on the second-step frequency set
\[
  \mathcal C_2
  =
  \{W_i'\tau_2(X_i)'u:
    u\in\mathbb R^p,\ (X_i,W_i)\in\operatorname{supp}(X_i,W_i)\}.
\]
Thus finite-support designs may still identify \(f_\beta\) if $\mathcal C_2\subseteq \overline{\mathcal C_1}$. When this compatibility fails, the denominator required in the deconvolution
step is not identified at some relevant frequencies, and \(f_\beta\) is not
point identified without additional variation or structural restrictions.
Section~\ref{ID_with_discrete_regressors} provides a formal discussion.
\end{remark}

\begin{assumption}\label{Fourier_inversion}
Let $\tau_2(X_i)$ be the measurable transformation defined in \eqref{tau_2}. 
For almost every realization $(x,w)$ of $(X_i,W_i)$, the ratio
\[
s\mapsto
\frac{\varphi_{\tTwo Y\mid X,W}(s\mid x,w)}
     {\varphi_{\tTwo W D\mid X,W}(s\mid x,w)}
\]
is absolutely integrable on $\mathbb{R}^p$.
\end{assumption}

Assumption~\ref{Fourier_inversion} is an integrability requirement ensuring that the conditional density \(f_{\beta\mid X,W}(\cdot\mid x,w)\) exists and that the characteristic function of $\beta_i$ is integrable so that $f_\beta$ can be obtained via Fourier inversion. 

\begin{assumption}\label{support}
\begin{enumerate}[label=(\roman*)]
\item \emph{Irregular design ($T=p$).}
The random variable $\det(X_i)$ admits a Lebesgue density that is
continuous and strictly positive in a neighborhood of $0$. Moreover, for every $\lambda$ in a dense subset of $\mathbb{R}^d\setminus\{0\}$, and every bounded continuous function $h:\mathbb{R}\to\mathbb{R}$, the mapping
    \[
        d_x \;\longmapsto\;
        \mathbb{E}\!\bigl[
          h\bigl(\mu_i'\widetilde{Y}_i\bigr)
          \,\big|\,
          \widetilde{W}_i'\mu_i = \lambda,\;
          \det(X_i)=d_x
        \bigr]
    \]
is well defined and continuous on a neighborhood of~$0$.

\item \emph{Regular design ($T>p$).}
$\Pr(\rank(X_i)=p)=1$.
\end{enumerate}
\end{assumption}

Assumption~\ref{support}(i) implies $\Pr(\det(X_i)=0)=0$, so stayers occur with probability zero.
Identification in the irregular case therefore exploits the fact that
$\det(X_i)$ can be arbitrarily close to zero with positive probability. The second part of (i) ensures that the joint conditional
distribution used in the proof of Theorem~\ref{main_result} (specifically,
the distribution of $\mu_i'\widetilde{Y}_i$ given both
$\widetilde{W}_i'\mu_i=\lambda$ and $\det(X_i)=d_x$) has a well-defined
limit as $d_x\to 0$.\footnote{The conditioning notation
$\mathbb{E}[\cdot\mid\widetilde{W}_i'\mu_i=\lambda,\,\det(X_i)=0]$ in
Theorem~\ref{main_result} is understood as the limit
$\lim_{d_x\to 0}
\mathbb{E}[\cdot\mid\widetilde{W}_i'\mu_i=\lambda,\,\det(X_i)=d_x]$.} This rules out pathological configurations in which
the conditional relationship between $\widetilde{Y}_i$ and
$\widetilde{W}_i'\mu_i$ changes discontinuously as the regressor matrix
approaches singularity. A sufficient primitive condition is that the transformation
\[
  (X_i,W_i,\beta_i)
  \mapsto
  \bigl(\mu_i'\beta_i,\widetilde W_i'\mu_i,\det(X_i)\bigr)
\]
induce a joint density that is continuous in a neighborhood of
\(\{\det(X_i)=0\}\), and that the marginal density of
\((\widetilde W_i'\mu_i,\det(X_i))\) be continuous and bounded away from zero
on the relevant conditioning set.

Assumption~\ref{support}(ii) is the standard full-rank condition ensuring that $X_i$ has a well-defined $p$-dimensional column space for all observations in the regular design. This ensures that both $\tau_1(X_i)$ and $\tau_2(X_i)$ exist for all units.

\begin{thm}\label{main_result}
Consider the outcome equation \eqref{eq:outcome_stack_t} with $T\ge p$, and suppose
Assumptions \ref{assm:A1}, \ref{assm:A1_prime}, \ref{assm:A1_double},
\ref{Fourier_inversion}, and \ref{support} hold.
Then the density of $\beta_i$ is point identified and given by
\begin{equation}
\label{fbeta_main}
f_\beta(b)
=
\Big(\tfrac{1}{2\pi}\Big)^p
\mathbb E\!\left[
\int_{\mathbb R^p}
e^{-\ii u^\prime b}\,
\frac{\varphi_{\tTwo Y\mid X,W}(u\mid X_i,W_i)}
     {\varphi_{\tTwo W D\mid X,W}(u\mid X_i,W_i)}
\,du
\right],
\qquad b\in\mathbb R^p,
\end{equation}
where
\begin{align}
\varphi_{\tTwo Y\mid X,W}(u\mid X_i,W_i)
&=
\varphi_{Y\mid X,W}\!\left(\tau_2(X_i)^\prime u\mid X_i,W_i\right),
\label{eq:phiY-tilde}
\\
\varphi_{\tTwo W D\mid X,W}(u\mid X_i,W_i)
&=
\varphi_D\!\left(W_i^\prime \tau_2(X_i)^\prime u\right).
\label{eq:phiWD-tilde}
\end{align}

Furthermore, $\varphi_D$ is point identified and given by:
\begin{equation}
\varphi_D(v\lambda)
=
\begin{cases}
\E\!\left[
\exp\!\big(\ii v\,\mu_i^\prime \tOne Y_i\big)
\,\middle|\,
\tOne W_i^\prime \mu_i=\lambda
\right],
& \text{regular design}, \\[0.5em]
\E\!\left[
\exp\!\big(\ii v\,\mu_i^\prime \tOne Y_i\big)
\,\middle|\,
\tOne W_i^\prime \mu_i=\lambda,\;
\det(X_i)=0
\right],
& \text{irregular design},
\end{cases}
\end{equation}
for any $v\in\mathbb{R}$ and for either any $\lambda\in\mathbb R^d$ in the support of
$\tOne W_i^\prime\mu_i$ in the regular design, or any $\lambda\in\mathbb R^d$ in the conditional support of
$\tOne W_i^\prime\mu_i$ given $\det(X_i)=0$ in the irregular design.\footnote{In this case, where $\Pr(\det(X_i)=0)=0$, the
conditioning notation is understood as shorthand for the corresponding
limit as $\det(X_i)\to 0$, as described in the notation section and
justified by Assumption~\ref{support}(i).}
\end{thm}

\begin{proof} The proof can be found in Section \ref{appendix_proof_thm1}.
\end{proof}

\begin{remark}[Discrete stayer--mover design]\label{rem:discrete_stayers}
Assumption~\ref{support}(i) corresponds to a continuous irregular design in which
$\Pr(\det(X_i)=0)=0$. An alternative specification allows a mass point at the
singular set. Suppose $T=p$ and
\[
\Pr(\det(X_i)=0)>0, 
\qquad 
\Pr(\det(X_i)\neq 0)>0,
\]
and assume that $\det(X_i)$ is absolutely continuous on
$\{\det(X_i)\neq 0\}$. Suppose further that
\[
\rank(X_i)=p-1
\quad \text{on } \{\det(X_i)=0\}
\]
with positive probability, so that $\adj(X_i)$ is nonzero on the relevant
stayer set. Then the stayer set $\{\det(X_i)=0\}$ has positive probability,
and the first-step equation
\[
\tOne Y_i=\tOne W_iD_i
\]
holds on this event. Hence the characterization of $\varphi_D$
follows from conditioning on
\[
\{\mu_i'\tOne W_i=\lambda',\ \det(X_i)=0\},
\]
without the limiting argument used in the continuous design.

This modification shows that exact stayers can replace the limiting argument in the
irregular design.
\end{remark}

Define
\[
Z_i^* := \frac{\mu_i' \tOne Y_i}{\|\lambda_i\|},
\qquad
\lambda_i := \tOne W_i'\mu_i \in \mathbb R^d,
\qquad
S_i := \frac{\lambda_i}{\|\lambda_i\|}\in\mathbb S^{d-1},
\]
on the event $\{\lambda_i\neq 0\}$. On the event where the annihilated equation satisfies $\tOne Y_i=\tOne W_iD_i$, we have
\[
Z_i^*
=
\frac{\mu_i'\tOne W_iD_i}{\|\lambda_i\|}
=
\frac{\lambda_i'D_i}{\|\lambda_i\|}
=
S_i'D_i .
\]
In the regular design this equality holds for every observation, whereas in the
irregular design it holds on the stayer set $\{\det(X_i)=0\}$.

Since $S_i$ is a measurable function of $(X_i,W_i)$ and
$D_i$ is independent of $(X_i,W_i)$, it follows that $D_i$ and $S_i$
are independent. Hence, conditional on $S_i=s$, the distribution of $Z_i^*$
coincides with the distribution of the scalar projection $s'D_i$. If $D_i$
admits a density $f_D$, then the density of this scalar projection is
\[
\mathcal R f_D(s,u)
:=
\int_{\{d\in\mathbb R^d:s'd=u\}} f_D(d)\,d\sigma(d),
\]
where $d\sigma$ denotes the induced Lebesgue measure on the hyperplane
$\{d\in\mathbb R^d:s'd=u\}$. Thus, for observed directions $s$,
\[
f_{Z^*\mid S}(u\mid s)=\mathcal R f_D(s,u),
\]
with the corresponding conditional-on-stayers version in the irregular design.
Under the stronger directional support conditions stated in
Corollary~\ref{cor_fD}, the Radon transform is known for enough directions to
identify $f_D$; see, for example, \citet{HoderleinKMammen}.

\begin{cor}
\label{cor_fD}
Let the assumptions of Theorem~\ref{main_result} hold and suppose that $D_i$
admits a density $f_D$. Suppose further that $S_i$ has a density with respect to surface measure on
$\mathbb S^{d-1}$ that is bounded away from zero uniformly on
$\mathbb S^{d-1}$, unconditionally in the regular design and conditionally on
the stayer event in the irregular design. In the irregular continuous design,
conditioning on the stayer event is understood in the limiting sense described
in Assumption~\ref{support}(i). Then, for surface-a.e. $s\in\mathbb S^{d-1}$
and all $u\in\mathbb R$ at which the relevant conditional densities are well
defined,
\begin{equation}
\mathcal{R}f_D(s,u)
=
\begin{cases}
f_{Z^*\mid S}(u\mid s),
& \text{regular design}, \\[0.5em]
f_{Z^*\mid S,\,\det(X_i)=0}(u\mid s),
& \text{irregular design}.
\end{cases}
\label{eq:radon_fD_cases}
\end{equation}
Consequently, $f_D$ is identified from \eqref{eq:radon_fD_cases} by the
Fourier slice theorem and Fourier inversion.
\end{cor}

\begin{proof}
The proof can be found in Section \ref{appendix_proof_cor}.
\end{proof}

\subsection{Examples}
\label{sec:transformations}

The examples in this section give concrete choices of the first- and
second-step transformations in irregular and regular designs. The examples also show that, in low-dimensional settings, a fixed $\mu_i$ (rather than a data-dependent one) is often convenient.

\subsubsection{Irregular design}

\begin{example}[Scalar irregular design, $T=p=1,\;q\geq 0$]
\label{irregular_T_1}

Consider equation \eqref{eq:outcome_stack_t} with
\[
  Y_i\in\mathbb{R},
  \qquad
  X_i\in\mathbb{R},
  \qquad
  W_i\in\mathbb{R}^{1\times(q+1)},
  \qquad
  D_i\in\mathbb{R}^{q+1}.
\]

In this case, the first-step transformation is trivial:
\[
  \tau_1(X_i)=1,
  \qquad
  \widetilde{Y}_i=Y_i,
  \qquad
  \widetilde{W}_i=W_i,
\]
and stayers correspond to units with $X_i=0$. Since $Y_i=W_iD_i$ on the stayer set $\{X_i=0\}$, it follows that for any $s\in\mathbb{R}$,
\begin{equation}
\label{eq:phiD_scalar_irregular}
\varphi_D(sW_i')
=
\mathbb{E}\!\left[
  e^{\ii sY_i}
  \,\middle|\,
  X_i=0,\;W_i
\right] \; \text{ a.s.}
\end{equation}
Thus each stayer realization $W_i=w$ identifies $\varphi_D$ along the
line $\{sw':s\in\mathbb{R}\}\subset\mathbb{R}^{q+1}$.

The second-step transformation is
\[
  \tau_2(X_i)=\frac{1}{X_i},
  \qquad
  \widetilde{\widetilde{Y}}_i=\frac{Y_i}{X_i},
  \qquad
  \widetilde{\widetilde{W}}_i=\frac{W_i}{X_i},
  \qquad
  \text{on }\{X_i\neq 0\}.
\]
Since $ \widetilde{\widetilde{Y}}_i
  =
  \beta_i+\tTwo W_iD_i$ and by conditional independence of $\beta_i$ and $D_i$,
\[
  \varphi_{\widetilde{\widetilde{Y}}\mid X,W}(u\mid X_i,W_i)
  =
  \varphi_{\beta\mid X,W}(u\mid X_i,W_i)\,
  \varphi_D\!\left(u\,\tTwo W_i'\right),
  \qquad u\in\mathbb{R}.
\]

Under the assumptions of Theorem~\ref{main_result},
\begin{equation}
\label{eq:fbeta_scalar_irregular_corrected}
f_\beta(b)
=
\frac{1}{2\pi}\,
\mathbb{E}\!\left[
  \int_{\mathbb{R}}
  e^{-\mathrm{i}ub}
  \frac{\varphi_{\widetilde{\widetilde{Y}}\mid X,W}(u\mid X_i,W_i)}
       {\varphi_D\!\left(u\,\tTwo W_i'\right)}
  \,du
\right],
\qquad b\in\mathbb{R},
\end{equation}
where the denominator is given by \eqref{eq:phiD_scalar_irregular}.

This example makes the geometric link between the two steps transparent. The
first step identifies $\varphi_D$ along the lines generated by stayer values
of $W_i$, while the second step requires only the values of $\varphi_D$ at
the mover-specific frequency vectors
\[
  \xi_i(u):=u\,\tTwo{W}_i'\in\mathbb{R}^{q+1}.
\]
Therefore, identification of $f_\beta$ requires that, for almost every mover
realization $(X_i,W_i)$ and every frequency $u$ used in the inversion, the
vector $\xi_i(u)$ lie on a line attained in the first step. To illustrate, suppose $q=1$, so $D_i=(D_{i1},D_{i2})'$. Each stayer
realization $W_i=(w_1,w_2)$ identifies $\varphi_D$ along the line
\[
  \{s(w_1,w_2)':s\in\mathbb{R}\}.
\]
If stayer values of $W_i$ lie only on the two coordinate directions
$(1,0)$ and $(0,1)$, then one identifies $\varphi_D(t,0)$ and
$\varphi_D(0,t)$, which determine the marginal distributions of $D_{i1}$
and $D_{i2}$ but not their joint distribution. A sufficient condition for
identification of the joint distribution is that the normalized stayer
directions $\frac{W_i'}{\|W_i\|}$ be sufficiently rich to span all directions in $\mathbb{R}^2$; for example,
it suffices that they have full support on $\mathbb{S}^1$. Under such a
condition, the union of these lines spans $\mathbb{R}^2$.

In the special case $q=0$, the vector $W_i$ reduces to the scalar $1$,
and $D_i$ is one-dimensional. Then
\[
  \varphi_D(s)
  =
  \mathbb{E}\!\left[
    e^{\ii sY_i}
    \,\middle|\,
    X_i=0
  \right],
\]
so the first step identifies the disturbance characteristic function on the
real line and no directional-richness condition is needed.

In terms of Assumption~\ref{assm:A1_double}, when $(T,p,q)=(1,1,0)$, the 
disturbance $D_i$ is scalar, and Assumption~\ref{assm:A1_double}(ii) reduces to
\[
  \Pr(\lambda_i\neq 0\mid X_i=0)=1.
\]
In this one-dimensional case, any nonzero value of $\lambda_i$ already
generates the line
\[
  \{v\lambda_i:v\in\mathbb{R}\}=\mathbb{R},
\]
so there is no separate directional-richness requirement. Thus, unlike the
multivariate cases discussed above, first-step identification does not rely
on cross-sectional variation in the direction of $W_i$, but only on the
nondegeneracy of the scalar loading on the stayer set.
\end{example}

The logic of the scalar irregular design extends to higher-dimensional
irregular designs, as illustrated by the next example.

\begin{example}[Irregular design $T=p=2,\;q\geq 0$]
\label{irregular_T_2}

Consider equation \eqref{eq:outcome_stack_t} with
\[
  Y_i=
  \begin{pmatrix}Y_{i1}\\Y_{i2}\end{pmatrix},
  \qquad
  X_i=
  \begin{pmatrix}X_{i11}&X_{i12}\\X_{i21}&X_{i22}\end{pmatrix},
  \qquad
  W_i\in\mathbb{R}^{2\times(q+2)},
  \qquad
  D_i\in\mathbb{R}^{q+2}.
\]

In the irregular design, the first-step transformation is
\[
  \tau_1(X_i)=\operatorname{adj}(X_i)
  =
  \begin{pmatrix}
    X_{i22} & -X_{i12}\\
    -X_{i21} & X_{i11}
  \end{pmatrix}.
\]
Thus
\[
  \widetilde{Y}_i:=\tau_1(X_i)Y_i,
  \qquad
  \widetilde{W}_i:=\tau_1(X_i)W_i.
\]
Premultiplying the outcome equation by $\operatorname{adj}(X_i)$ yields
\[
  \widetilde{Y}_i
  =
  \det(X_i)\beta_i+\widetilde{W}_iD_i.
\]
Stayers correspond to singular realizations of $X_i$, that is,
$\det(X_i)=0$. On the stayer set,
\[
  \widetilde{Y}_i=\widetilde{W}_iD_i.
\]

Since $\operatorname{adj}(X_i)$ has rank at most one on
$\{\det(X_i)=0\}$, the transformed system contains at most one independent
scalar restriction. Fix a measurable projection vector
$\mu_i\in\mathbb{R}^2$ such that
\[
  \mu_i'\operatorname{adj}(X_i)\neq 0
\]
with positive probability conditional on $\{\det(X_i)=0\}$. For
concreteness, take $\mu_i=e_1$ and suppose that
$e_1'\operatorname{adj}(X_i)\neq 0$ conditional on $\det(X_i)=0$. Define
\[
  Y_i^*:=\mu_i'\widetilde{Y}_i,
  \qquad
  \lambda_i:=\widetilde{W}_i'\mu_i\in\mathbb{R}^{q+2}.
\]
Then, on the stayer set,
\[
  Y_i^*=\lambda_i'D_i.
\]
Using the independence of $D_i$ from $(X_i,W_i)$, and hence from
$\lambda_i$, it follows that, for any $v\in\mathbb{R}$ and for a.e.
$\lambda$ in the conditional support of $\lambda_i$ given $\det(X_i)=0$,
\begin{equation}
\label{phiD_irregular_T2p2}
\varphi_D(v\lambda)
=
\mathbb{E}\!\left[
  e^{\ii vY_i^*}
  \,\middle|\,
  \lambda_i=\lambda,\;
  \det(X_i)=0
\right].
\end{equation}
Thus the first step identifies $\varphi_D$ along lines generated by the
stayer projection directions $\lambda_i$.

With the choice $\mu_i=e_1$, the vector
$\lambda_i=\widetilde{W}_i'e_1$ is the transpose of the first row of
$\widetilde{W}_i$. In the benchmark case $q=0$ and $W_i=I_2$, this reduces to
\[
  \lambda_i
  =
  \operatorname{adj}(X_i)'e_1
  =
  (X_{i22},-X_{i12})',
\]
which is a $90$-degree rotation, up to sign, of the second column
$(X_{i12},X_{i22})'$ of $X_i$. Since rotation is a bijection on
$\mathbb{R}^2\setminus\{0\}$, requiring $\lambda_i$ to have full directional
support conditional on $\det(X_i)=0$ is equivalent in this benchmark case to
requiring that the second column of $X_i$ have full directional support on the
same event. Intuitively, among units with $\det(X_i)=0$, the two columns of
$X_i$ are collinear, so each stayer regressor matrix is characterized by a
single direction in $\mathbb{R}^2$. The support condition requires this
direction to vary sufficiently across individuals so that all directions are
represented. A sufficient primitive condition in the benchmark case is that
the conditional distribution of $(X_{i12},X_{i22})'$ given $\det(X_i)=0$
have support $\mathbb{R}^2\setminus\{0\}$, or equivalently that its normalized
direction have support $\mathbb{S}^1$.

For movers with $\det(X_i)\neq 0$, the second-step transformation is
\[
  \tau_2(X_i)=X_i^{-1},
  \qquad
  \widetilde{\widetilde{Y}}_i
  :=
  \tau_2(X_i)Y_i
  =
  X_i^{-1}Y_i,
  \qquad
  \widetilde{\widetilde{W}}_i
  :=
  \tau_2(X_i)W_i
  =
  X_i^{-1}W_i.
\]
Thus
\[
  \widetilde{\widetilde{Y}}_i
  =
  \beta_i+\widetilde{\widetilde{W}}_iD_i.
\]
By the maintained independence of $D_i$ from $(\beta_i,X_i,W_i)$,
\begin{equation}
\label{eq:cf_factor}
\varphi_{\widetilde{\widetilde{Y}}\mid X,W}(u\mid X_i,W_i)
=
\varphi_{\beta\mid X,W}(u\mid X_i,W_i)\,
\varphi_D(\widetilde{\widetilde{W}}_i'u),
\qquad u\in\mathbb{R}^2.
\end{equation}
Thus the second step requires $\varphi_D$ only at the mover-specific
frequency vectors
\[
  \xi_i(u)
  :=
  \widetilde{\widetilde{W}}_i'u
  =
  W_i'X_i^{-1\prime}u
  \in\mathbb{R}^{q+2}.
\]
Under the assumptions of Theorem~\ref{main_result},
\begin{equation}
\label{fbeta_T2_irregular}
f_\beta(b)
=
\frac{1}{(2\pi)^2}
\mathbb{E}\!\left[
  \int_{\mathbb{R}^2}
  e^{-\ii u'b}
  \frac{\varphi_{\widetilde{\widetilde{Y}}\mid X,W}(u\mid X_i,W_i)}
       {\varphi_D(\widetilde{\widetilde{W}}_i'u)}
  \,du
\right],
\qquad b\in\mathbb{R}^2.
\end{equation}

The compatibility requirement between the first and second stages is now
transparent. The first step identifies $\varphi_D$ only along lines of the
form
\[
  \{v\lambda:v\in\mathbb{R}\},
\]
where $\lambda$ is generated by stayer projections. The second step requires
the values of the same characteristic function at the mover-specific points
\[
  \xi_i(u)=W_i'X_i^{-1\prime}u,
  \qquad u\in\mathbb{R}^2.
\]
Therefore, identification of $f_\beta$ requires that, for almost every mover
realization and every frequency $u$ entering the inversion, the vector
$\xi_i(u)$ lie on a line attainable from the family of lines recovered in the
first step. Equivalently, the collection of stayer projection directions must
be rich enough that the second-step frequency set lies in the region where
$\varphi_D$ is identified.

A strong sufficient condition is that, conditional on $\det(X_i)=0$, the
random vector
\[
  \lambda_i
  =
  W_i'\operatorname{adj}(X_i)'\mu_i
\]
has support $\mathbb{R}^{q+2}\setminus\{0\}$. Equivalently, the normalized
stayer projection directions $\frac{\lambda_i}{\|\lambda_i\|}$ have full support on $\mathbb{S}^{q+1}$. Under this condition, the union of
the first-step frequency lines spans $\mathbb{R}^{q+2}$.

When $q=0$, the same geometric argument applies with
$D_i=(D_{i1},D_{i2})'\in\mathbb{R}^2$. The first step identifies
$\varphi_D$ along one-dimensional lines through the origin, while the second
step requires it at the mover-specific frequencies
$\xi_i(u)=X_i^{-1\prime}u$.
\end{example}

\subsubsection{Regular design.}

When $T>p$, the regressor matrix $X_i$ has full column rank with probability
one under Assumption~\ref{support}(ii). In this case the annihilator
$\tau_1(X_i)$ and the left inverse $\tau_2(X_i)$ can be constructed for every
observation.

A convenient choice for $\tau_1(X_i)$ is the orthogonal projector onto the
null space of $X_i'$,
\begin{equation}
\label{ortho_proj}
\tau_1(X_i)
=
I_T - X_i(X_i'X_i)^{-1}X_i',
\end{equation}
which satisfies $\tau_1(X_i)X_i=0$ identically.
Thus the transformed equation $\widetilde{Y}_i=\widetilde{W}_iD_i$ holds for
all units, so the first step can exploit the full sample.
Any measurable $\tau_2(X_i)$ satisfying $\tau_2(X_i)X_i=I_p$ suffices for
the second-step decomposition in \eqref{tau2_YW}, since $X_i$ has full column
rank and a left inverse exists for every observation.

In contrast to the irregular design, where first-step identification relies on
a restricted subset of observations, in the regular design all observations
contribute to the first step. Accordingly, the relevant support requirements
are unconditional rather than conditional on a stayer-type event.

\begin{example}[Regular design $T=2,\;p=1,\;q=0$]
\label{ex_regular_T2p1}

Although the orthogonal projector \eqref{ortho_proj} is a natural choice, the
regular design admits many valid annihilators.
To illustrate, consider \eqref{eq:outcome_stack_t} with
\[
  Y_i=
  \begin{pmatrix}Y_{i1}\\Y_{i2}\end{pmatrix},
  \qquad
  X_i=
  \begin{pmatrix}X_{i1}\\X_{i2}\end{pmatrix},
  \qquad
  D_i=
  \begin{pmatrix}D_{i1}\\D_{i2}\end{pmatrix}
  \in\mathbb{R}^2,
  \qquad
  W_i=I_2.
\]

An admissible annihilator is
\[
  \tau_1(X_i)
  =
  \begin{pmatrix}X_{i2}&-X_{i1}\\0&0\end{pmatrix},
  \qquad
  \tau_1(X_i)X_i=0.
\]
The transformed outcome is therefore
\[
  \widetilde{Y}_i
  =
  \tau_1(X_i)Y_i
  =
  \begin{pmatrix}X_{i2}D_{i1}-X_{i1}D_{i2}\\0\end{pmatrix}.
\]
Let $\mu_i=e_1$ and define
\[
  Y_i^* = e_1'\widetilde{Y}_i = X_{i2}D_{i1}-X_{i1}D_{i2},
  \qquad
  \lambda_i = \widetilde{W}_i'e_1 =
  \begin{pmatrix}X_{i2}\\-X_{i1}\end{pmatrix}.
\]
To interpret the support condition in Assumption~\ref{assm:A1_double}(i),
note that $\lambda_i=(X_{i2},-X_{i1})'$ is a $90$-degree rotation of
$X_i=(X_{i1},X_{i2})'$. Hence requiring $\lambda_i$ to have full support on
$\mathbb{R}^2\setminus\{0\}$ is equivalent to requiring that $X_i$ have full
support on $\mathbb{R}^2\setminus\{0\}$. This holds, for example, if $X_i$ has
a continuous distribution with full support on $\mathbb{R}^2$. Intuitively,
each observation contributes one first-step projection direction
$\lambda_i$, which is orthogonal to $X_i$. As the direction of $X_i$ varies
across the population, the normalized directions
$\lambda_i/\|\lambda_i\|$ rotate through the full circle and cover all of
$\mathbb{S}^1$.

Since $Y_i^*=\lambda_i'D_i$, the first step identifies the disturbance characteristic function only
along lines generated by $\lambda_i$:
\begin{equation}
\label{phiD_projection_example}
\varphi_D(v\lambda)
=
\mathbb{E}\!\left[
  \exp(\ii vY_i^*)
  \,\middle|\,
  \lambda_i=\lambda
\right],
\qquad v\in\mathbb{R},
\end{equation}
for any $\lambda$ in the support of $\lambda_i$.

To identify $f_\beta$, use the left inverse
\[
  \tau_2(X_i)=\frac{X_i'}{X_i'X_i},
\]
which is well defined under Assumption~\ref{support}(ii), since
$X_i'X_i=X_{i1}^2+X_{i2}^2>0$ a.s.
Then
\[
  \widetilde{\widetilde{Y}}_i
  =
  \tau_2(X_i)Y_i
  =
  \beta_i+\frac{X_i'D_i}{X_i'X_i}.
\]
Accordingly, for any $s\in\mathbb{R}$,
\[
  \varphi_{\widetilde{\widetilde{Y}}\mid X}(s\mid X_i)
  =
  \varphi_{\beta\mid X}(s\mid X_i)\,
  \varphi_D\!\left(s\,\frac{X_i}{X_i'X_i}\right).
\]
Since $q=0$, there is no additional $W$-term, and the denominator in
\eqref{fbeta_main} is simply the disturbance characteristic function evaluated
at the mover-specific frequency vector
\begin{equation}
\label{eq:xi_regular}
  \xi_i(s) := s\,\frac{X_i}{X_i'X_i} \in \mathbb{R}^2.
\end{equation}
Under the assumptions of Theorem~\ref{main_result},
\begin{equation}
\label{eq:regular_example_density}
f_\beta(b)
=
\frac{1}{2\pi}
\mathbb{E}\!\left[
  \int_{\mathbb{R}}
  e^{-\ii sb}
  \frac{\varphi_{\widetilde{\widetilde{Y}}\mid X}(s\mid X_i)}
       {\varphi_D\!\left(s\,\frac{X_i}{X_i'X_i}\right)}
  \,ds
\right],
\qquad b\in\mathbb{R}.
\end{equation}

To see how the two steps are linked, note that each required frequency vector
$\xi_i(s)$ lies on the line generated by $X_i$; see \eqref{eq:xi_regular}.
Writing $\xi_i(s)$ in polar form,
\[
  \xi_i(s) = r_i(s)\,d_i(s),
  \qquad
  r_i(s) = \frac{|s|}{\|X_i\|},
  \qquad
  d_i(s) = \operatorname{sgn}(s)\frac{X_i}{\|X_i\|},
\]
shows that as $s$ ranges over $\mathbb{R}$, $d_i(s)$ takes both values
$\pm X_i/\|X_i\|$, so the second step requires values of $\varphi_D$ along
the full line generated by $X_i$ (both directions).
The first step identifies $\varphi_D$ along lines generated by $\lambda_i =
(X_{i2},-X_{i1})'$, which are orthogonal to $X_i$ (since
$\lambda_i'X_i = X_{i1}X_{i2} - X_{i1}X_{i2} = 0$).
The regular design therefore faces the same compatibility issue as the
irregular design: identification of $f_\beta$ requires the collection of
first-step lines to be rich enough, as $X_i$ varies in the sample, to cover
the lines required in the second step.
A strong sufficient condition is that $X_i$ have full support on
$\mathbb{R}^2\setminus\{0\}$, which holds in particular if $X_i$ has a
continuous distribution with full support on $\mathbb{R}^2$.
This guarantees that the normalized direction $X_i/\|X_i\|$ has full support
on $\mathbb{S}^1$, and that the orthogonal directions
$\lambda_i/\|\lambda_i\|$ also rotate through the full circle, so
every $\xi_i(s)$ is attainable.

This example also illustrates the contrast with the irregular design
$T=2=p,\;q=0$.
In the irregular design, the first-step characterization of $\varphi_D$ uses
only singular realizations of the regressor matrix.
If those realizations generate only a small set of directions, then $\varphi_D$
is identified only on a sparse collection of lines and $F_D$ need not be point
identified.
In the regular design, by contrast, each observation contributes a first-step
line orthogonal to $X_i$, while the second-step denominator requires only the
values of $\varphi_D$ at the specific frequency vectors $\xi_i(s)$, which lie
on the line generated by $X_i$.
As $X_i$ rotates across the sample, these two families of lines rotate with it,
and whether they cover the relevant region of frequency space depends on the
support of $X_i$.
\end{example}

In higher-dimensional regular designs, the same logic applies, but the
support condition involves the joint variation of $(X_i$, $W_i)$, and the
measurable choice of~$\mu_i$. When $q > 0$, variation in~$X_i$ alone
generally does not suffice. Example~\ref{ex_regular_mu_variation}
illustrates this for $(T,p,q)=(2,1,1)$.

\begin{example}[Regular design requiring variation in $\mu_i$]
\label{ex_regular_mu_variation}

Consider a regular design with $(T,p,q)=(2,1,1)$, so that
\[
  X_i=
  \begin{pmatrix}
    X_{i1}\\X_{i2}
  \end{pmatrix},
  \qquad
  W_i=
  \begin{pmatrix}
    W_{i1} & 1 & 0\\
    W_{i2} & 0 & 1
  \end{pmatrix},
  \qquad
  D_i\in\mathbb{R}^3.
\]
Let $\tau_1(X_i)$ be the orthogonal projector
\[
  \tau_1(X_i)
  =
  I_2-X_i(X_i'X_i)^{-1}X_i'.
\]

For each $i$, $\tOne W_i=\tau_1(X_i)W_i$ has rank one, so the row space of
$\tOne W_i$ is one-dimensional, but its direction varies with $X_i$.
As $(X_i,W_i)$ vary across the population, these row spaces generate a
rich collection of directions in $\mathbb{R}^3$.

Now fix $\mu_i\equiv e_1$. Then
\[
  \lambda_i=\widetilde W_i' e_1 = \frac{1}{X_i'X_i}
  \begin{pmatrix}
    W_{i1}X_{i2}^2
    - W_{i2}X_{i1}X_{i2}
    \\[0.6em]
    X_{i2}^2
    \\[0.6em]
    -X_{i1}X_{i2}
  \end{pmatrix}.
\]
Notice that the second coordinate satisfies
\[
  \lambda_{i2}=\frac{X_{i2}^2}{X_i'X_i}\geq 0,
\]
so $\lambda_i$ is restricted to the closed half-space $\{\lambda\in\mathbb{R}^3:\lambda_2\geq 0\}$. Hence
\[
  \operatorname{supp}(S_i)\subsetneq\mathbb{S}^2,
\]
and the support condition in Assumption~\ref{assm:A1_double} fails even though the row spaces of $\widetilde W_i$ vary
with $(X_i,W_i)$.

Now instead choose $\mu_i=(X_{i2},-X_{i1})^\prime$. Then
\[
  \lambda_i=\widetilde W_i'\mu_i
  =
  \begin{pmatrix}
    X_{i2}W_{i1}-X_{i1}W_{i2}\\
    X_{i2}\\
    -X_{i1}
  \end{pmatrix}.
\]
Therefore, if $X_i$ has full support on $\mathbb{R}^2\setminus\{0\}$ and,
conditional on $X_i=x$, the pair $(W_{i1},W_{i2})$ has full support on
$\mathbb{R}^2$, then $\lambda_i$ has full support on
$\mathbb{R}^3\setminus\{0\}$. In this case
$\operatorname{supp}(S_i)=\mathbb{S}^2$, so
Assumption~\ref{assm:A1_double} holds.

This example shows that, in the regular design, a fixed choice of $\mu_i$
may fail to exploit the available directional variation, whereas a measurable
choice of $\mu_i=\mu(X_i,W_i)$ can recover full support.
\end{example}

\begin{remark}
\label{rem:orthogonality}
The geometric relationship between the first-step directions and the second-step frequencies is one of orthogonal complementarity. Identification succeeds not by aligning these directions, but by ensuring the first step covers enough of the sphere to pin down $\varphi_D$.
 
First, in the regular design $(T,p)=(2,1)$, the choice of
annihilator has no effect on the first-step directions. The condition
$\tau_1(X_i)X_i=0$ forces the rows of $\tau_1(X_i)$ to lie in the
orthogonal complement of $X_i$ in $\mathbb{R}^2$.  Since the orthogonal
complement of a nonzero vector in $\mathbb{R}^2$ is one-dimensional,
spanned by $(X_{i2},-X_{i1})'$, every annihilator
$\tau_1(X_i)$, including the orthogonal projector~\eqref{ortho_proj} and
the adjugate-based choice in Example~\ref{ex_regular_T2p1}, generates the
same first-step direction $(X_{i2},-X_{i1})'$, up to a scalar factor that
cancels upon normalization.  In particular, the normalized direction
$S_i = \lambda_i/\|\lambda_i\|$, the characteristic function $\varphi_D$, and consequently the density $f_\beta$
are all invariant to the choice of~$\tau_1(X_i)$.  
This
invariance is specific to $T-p=1$; when $T-p>1$, the null space of $X_i'$
is multidimensional, and different annihilators may recover different
first-step directions.

Second, identification
may fail if the projection vector $\mu_i$ is fixed in a way that
does not fully exploit the row space of
$\tOne W_i=\tau_1(X_i)W_i$.  A fixed choice such as $\mu_i= e_1$
selects only one linear combination of the rows of $\tOne W_i$, so the
resulting vectors $\lambda_i=\tOne W_i'\mu_i$ may lie in a strict subset
of $\mathbb{R}^d$, even when the row spaces of $\tOne W_i$ are themselves
rich.  Example~\ref{ex_regular_mu_variation} illustrates this mechanism in
a regular design with $(T,p,q)=(2,1,1)$.

Third, the orthogonality between first-step and second-step directions
observed in Example~\ref{ex_regular_T2p1} is not a low-dimensional
artifact but a structural feature of the standard annihilator--left-inverse
pair when $q = 0$.  When both $\tau_1(X_i)$ and $\tau_2(X_i)$ are chosen
as in Example~\ref{ex_regular_T2p1}, the first-step directions
$\lambda_i = \tau_1(X_i)'\mu_i$ lie in $\textrm{null}(X_i')$, and the second-step
frequency vectors $\xi_i(u) = X_i(X_i'X_i)^{-1}u$ lie in
$\mathrm{col}(X_i)$; since these subspaces are orthogonal complements of
one another in $\mathbb{R}^T$, every first-step direction is orthogonal
to every second-step direction, regardless of the dimension $T - p$ of the
null space.  Enlarging the null space creates additional independent
first-step directions, but they all remain orthogonal to
$\mathrm{col}(X_i)$.

Two mechanisms can break this orthogonality.  First, when $q > 0$, the
matrix $W_i \neq I_T$ maps $\textrm{null}(X_i')$ and $\mathrm{col}(X_i)$ into
$\mathbb{R}^d$ in a way that need not preserve their orthogonality: the
first-step directions
$\lambda_i = W_i'\tau_1(X_i)'\mu_i$ and the second-step frequency vectors
$\xi_i(u) = W_i'\tau_2(X_i)'u$ can have nonzero inner product in
$\mathbb{R}^d$ even though their pre-images are orthogonal in
$\mathbb{R}^T$.  Example~\ref{ex_regular_mu_variation} illustrates this
with $(T,p,q) = (2,1,1)$.  Second, even with $q = 0$, one could use a
non-standard left inverse $\tau_2(X_i)$ whose transpose has image not
contained in $\mathrm{col}(X_i)$; any $\tau_2(X_i)$ satisfying
$\tau_2(X_i)X_i = I_p$ is admissible, and non-Moore--Penrose choices can
introduce a $\textrm{null}(X_i')$ component into the second-step directions.

Thus, the orthogonality visible in
Example~\ref{ex_regular_T2p1} is a consequence of combining the standard
annihilator--left-inverse pair with $q = 0$.  The support loss caused by fixing $\mu_i$ is a
separate issue that can arise more generally, as
Example~\ref{ex_regular_mu_variation} shows.
\end{remark}

\section{Estimation of \texorpdfstring{$f_\beta$}{f\_beta}}
\label{sec:estimation}

The parameter of interest is $f_\beta$, with characteristic function
$\varphi_\beta(u)$, $u\in\R$. The results in
Section~\ref{sec:densities} imply that there exists an identified
function $m_0:\R\to\mathbb{C}$ such that
\begin{equation}
\label{moment_restriction}
m_0(u)=\varphi_\beta(u), \qquad u\in\R.
\end{equation}
Given \eqref{moment_restriction}, define the weighted population
criterion
\begin{equation}
Q(\varphi)
=
\int_{\R}
    \big|\varphi(u)-m_0(u)\big|^2
\,d\nu(u),
\label{eq:pop_criterion_common}
\end{equation}
where $d\nu(u)=\nu_0(u)\,du$ is a finite weighting measure with
$\nu_0$ having full support on $\R$, so that $Q(\varphi)=0$ if and
only if $\varphi(u)=\varphi_\beta(u)$ $\nu$-almost everywhere.%
\footnote{To see this from \eqref{fbeta_main}: the ratio of the two
conditional characteristic functions appearing in that expression
identifies the conditional characteristic function
$\varphi_{\beta|X,W}$, and $m_0(\cdot)$ is obtained by integrating
this ratio with respect to the distribution of $(X,W)$.}

\begin{remark}[Choice of weighting measure]\label{rem:parseval}
When $\nu$ is Lebesgue measure, Parseval's theorem gives
\begin{equation}
Q(\varphi)
=
\int_{\R}\big|\varphi(u)-\varphi_\beta(u)\big|^2\,du
= 2\pi
\int_{\R}\big|f(b)-f_\beta(b)\big|^2\,db,
\label{eq:integrated_squared_error}
\end{equation}
where $f(\cdot)$ is the density corresponding to the characteristic function $\varphi(\cdot)$. Consequently, the frequency-domain criterion \eqref{eq:pop_criterion_common} targets the integrated squared error of the density \eqref{eq:integrated_squared_error}.
In practice, however, a non-uniform weighting measure $\nu$ may be preferred. For example, the first-stage estimator of ${m}_0(u)$ can become increasingly noisy at high frequencies, especially in the irregular design, where the denominator $\varphi_D(u/X_i)$ decays as $|u|$ grows.  A non-uniform weighting measure $\nu_0$ that places less mass on high frequencies (such as a normal or Student-$t$ density) regularizes the criterion by downweighting the region where estimation noise dominates, at the cost of no longer targeting $L^2$ density error exactly.  The choice of $\nu_0$ thus trades off statistical stability against fidelity to the $L^2$ loss: heavier tails in $\nu_0$ preserve more high-frequency information (useful for recovering sharp features such as multimodality), while lighter tails provide more effective regularization.
\end{remark}

Estimation proceeds in two stages. The first stage, which is
design-specific, constructs an estimator $\widehat{m}_N(u)$ of
$m_0(u)$. The second stage, which is common across designs, recovers
$f_\beta$ by minimizing the sample analog of \eqref{eq:pop_criterion_common}
over a finite-dimensional sieve. We describe the common second stage
first, then specialize the first stage to the scalar irregular design
$(T,p,q)=(1,1,0)$ and the regular design $(T,p,q)=(2,1,0)$. These are the configurations that arise in our application.

The identification results in Section~\ref{sec:identification} apply to
both regular and irregular designs in arbitrary dimensions. The estimators developed in this section specialize to the scalar irregular design and a low-dimensional regular design. The first-step estimation of $\varphi_D$ and
second-step sieve minimum distance extend to higher-dimensional
designs by replacing the scalar kernel regressions with their
multivariate counterparts and the directional smoothing on $\mathbb{S}^1$
with smoothing on $\mathbb{S}^{d-1}$, but the additional tuning
complexity and the curse of dimensionality in the first stage make a
general-purpose implementation less practical. We therefore focus on the
cases with empirical relevance.

\subsection{Sieve minimum distance estimator}
\label{subsec:common_md}

Let $\{q_s:s=0,\ldots,S-1\}\subset L^1(\R)$ be a vector of basis
functions, $q^S(b)=(q_0(b),\ldots,q_{S-1}(b))'$, and for
$\pi\in\R^S$ set $f_S(b;\pi)=q^S(b)'\pi$. The characteristic
function of $f_S(\cdot\,;\pi)$ is then
\begin{equation}
\varphi_S(u;\pi)
=
\int_{\R} e^{\ii ub} f_S(b;\pi)\,db
=
z^S(u)'\pi,
\qquad
z^S(u)\equiv (\mathcal{F} q^S)(u).
\label{eq:sieve_cf_common}
\end{equation}
Enforcing the unit-mass restriction $\varphi_S(0;\pi)=1$ defines the
parameter space
\begin{equation}
\Pi_S
=
\{\pi\in\R^S:\ z^S(0)'\pi=1\}.
\label{eq:pi_constraint_common}
\end{equation}
The sieve minimum distance estimator minimizes the sample analog
of \eqref{eq:pop_criterion_common} over $\Pi_S$:
\begin{equation}
\widehat{\pi}
\;\in\;
\arg\min_{\pi \in \Pi_S}
\widehat{Q}_N(\pi),
\qquad
\widehat{Q}_N(\pi)
=
\int
  \bigl|\varphi_S(u;\pi) - \widehat{m}_N(u)\bigr|^2
\,d\nu(u),
\label{eq:sample_criterion_common}
\end{equation}
and the estimator of the density function is given by:
\begin{equation}
\widehat{f}_\beta(b) = q^S(b)'\widehat{\pi}.
\label{eq:fhat_def_common}
\end{equation}

\paragraph{Hermite sieve.}
For implementation we take $\{q_s\}$ to be the orthonormal Hermite
functions on $L^2(\R)$:
\[
q_s(v)=c_s^{-1/2} e^{-v^2/2}H_s(v),
\qquad
c_s=2^s s!\sqrt{\pi},
\]
where $H_s$ denotes the physicists' Hermite polynomial of degree $s$.
These functions are eigenfunctions of the Fourier transform $(\mathcal{F}q_s)(u)=\sqrt{2\pi}\,\ii^s q_s(u)$,
so that
\begin{equation}
\varphi_S(u;\pi)
=
\sqrt{2\pi}
\sum_{s=0}^{S-1}\ii^s\pi_s q_s(u).
\label{eq:hermite_cf_common}
\end{equation}
Since $f_\beta$ is real-valued and each $q_s$ is real, we restrict $\pi\in\Pi_S$. Under this restriction, the criterion
$\widehat{Q}_N(\pi)$ is quadratic in $\pi$ and admits a closed-form
constrained minimizer given by \eqref{eq:final_estimator_common} below.

\paragraph{Closed-form solution.}
Approximating the integral in \eqref{eq:sample_criterion_common} by a
quadrature rule $\int g(u)\,d\nu(u)\approx\sum_{\ell=1}^L w_\ell
g(u_\ell)$, define
\begin{equation}
z_\ell = z^S(u_\ell),
\qquad
\widehat{\Omega}
=
\sum_{\ell=1}^L
w_\ell\,\mathrm{Re}\!\bigl(z_\ell \bar{z}_\ell'\bigr),
\qquad
\widehat{V}
=
\mathrm{Re}\!\left(
\sum_{\ell=1}^L
w_\ell\,\overline{\widehat{m}_N(u_\ell)}\, z_\ell
\right).
\label{eq:OmegaV_def}
\end{equation}
Expanding the squared modulus in \eqref{eq:sample_criterion_common}
gives
\[
\widehat{Q}_N(\pi)
=
\pi'\widehat{\Omega}\,\pi
-
2\,\widehat{V}'\pi
+
\sum_{\ell=1}^L w_\ell\,\bigl|\widehat{m}_N(u_\ell)\bigr|^2,
\]
with unconstrained first-order condition $\widehat{\Omega}\pi=\widehat{V}$.
Setting $A=z^S(0)\in\R^S$, and assuming $\widehat{\Omega}$ is
invertible,%
\footnote{A sufficient condition is that the quadrature weights
$w_\ell$ are positive and the $L\times S$ matrix with rows
$z^S(u_\ell)'$ has full column rank $S$. By construction
$\widehat{\Omega}$ is real symmetric and positive semidefinite; under
the stated rank condition it is positive definite, and hence
invertible.}
the constrained minimizer subject to $A'\pi=1$ is
\begin{equation}
\widehat{\pi}
=
\widehat{\Omega}^{-1}
\!\left(
  \widehat{V}
  +
  \frac{1 - A'\widehat{\Omega}^{-1}\widehat{V}}
       {A'\widehat{\Omega}^{-1}A}
  \,A
\right),
\label{eq:final_estimator_common}
\end{equation}
which satisfies $A'\widehat{\pi}=1$ by construction. The sieve
estimator of $f_\beta$ is \eqref{eq:fhat_def_common} with
$\widehat{\pi}$ given by \eqref{eq:final_estimator_common}.

\begin{remark}
    The identification result delivers the characteristic function of $\beta_i$ as a ratio of observable characteristic functions, suggesting a direct plug-in estimator. However, such an approach would compound two sources of ill-posedness: division by $\varphi_D$, whose magnitude decays at high frequencies, and subsequent Fourier inversion to recover the density $f_\beta$ from the estimated characteristic function. The sieve minimum distance estimator avoids this compounding by replacing the inverse Fourier transform with a projection problem. Specifically, because Hermite functions are eigenfunctions of the Fourier transform, the mapping from sieve coefficients $\pi$ to the candidate characteristic function $\varphi_S(u,\pi)$ is linear and bounded. Minimizing the integrated squared distance between $\varphi_S(u,\pi)$ and the estimated target $\widehat m_N(u)$ is therefore a well-posed quadratic program admitting a closed-form solution. This construction isolates the ill-posedness to a single step, i.e. the estimation of the deconvolution ratio.
\end{remark}

The remainder of this section constructs the design-specific
first-stage estimator $\widehat{m}_N(u)$.

\subsection{Scalar irregular design: \texorpdfstring{$(T,p,q)=(1,1,0)$}{(T,p,q)=(1,1,0)}}
\label{subsec:estimation_irregular_scalar}

In the scalar irregular design,
\[
Y_i=X_i\beta_i+D_i,
\]
where $D_i$ is scalar. Let $\tau_x>0$ be a mover--stayer threshold and
define
\[
\mathcal{S}_N(\tau_x)=\{i:\ |X_i|<\tau_x\},
\qquad
\mathcal{M}_N(\tau_x)=\{i:\ |X_i|\ge \tau_x\}.
\]
Observations in $\mathcal{S}_N(\tau_x)$ are stayers, used to
estimate the disturbance characteristic function $\varphi_D$, while
observations in $\mathcal{M}_N(\tau_x)$ are movers, used in the
deconvolution step.

For movers, the second-step transformation is
\[
\tTwo Y_i=\frac{Y_i}{X_i}
=
\beta_i+\frac{D_i}{X_i},
\qquad i\in\mathcal{M}_N(\tau_x).
\]

We estimate $\varphi_{\tTwo Y\mid X}(u\mid X_i)=\E[e^{\ii u\tTwo Y_i}\mid X_i]$
by kernel regression on the mover sample:
\begin{equation}
\widehat\varphi_{\tTwo Y\mid X}(u\mid X_i)
=
\sum_{k\in\mathcal{M}_N(\tau_x)}
\omega_{k,i}^{(x)}
\exp\!\left\{\ii u\frac{Y_k}{X_k}\right\},
\label{eq:numerator_scalar_irregular}
\end{equation}
where
\[
\omega_{k,i}^{(x)}
=
\frac{K_{h_x}(X_k-X_i)}
     {\sum_{s\in\mathcal{M}_N(\tau_x)}K_{h_x}(X_s-X_i)}.
\]
We estimate $\varphi_D$ by local smoothing at $X_i=0$, where
$Y_i=D_i$:%
\footnote{The sample average in \eqref{eq:mhat_scalar_irregular} is
taken over movers only, since $\tTwo Y_i=Y_i/X_i$ is defined only for
$i\in\mathcal{M}_N(\tau_x)$; the exclusion of stayers is
asymptotically negligible as $\tau_x\to 0$.}
\begin{equation}
\widehat\varphi_D(v)
=
\sum_{j\in\mathcal{S}_N(\tau_x)}
\omega_j^{(0)}
e^{\ii v Y_j},
\qquad
\omega_j^{(0)}
=
\frac{K_{h_0}(X_j)}
     {\sum_{s\in\mathcal{S}_N(\tau_x)}K_{h_0}(X_s)}.
\label{eq:denominator_scalar_irregular}
\end{equation}
For each mover $i$, the trimmed conditional ratio is
\begin{equation}
\widehat R_N(u\mid X_i)
=
\frac{
\widehat\varphi_{\tTwo Y\mid X}(u\mid X_i)
}{
\widehat\varphi_D(u/X_i)
}
\mathbf{1}\!\left\{
\big|\widehat\varphi_D(u/X_i)\big|>\tau_{\rm den}
\right\},
\label{eq:Rhat_scalar_irregular}
\end{equation}
where $\tau_{\rm den}>0$ is a denominator trimming threshold. Averaging
over movers gives the first-stage estimator
\begin{equation}
\widehat m_N(u)
=
\frac{1}{N_m}
\sum_{i\in\mathcal{M}_N(\tau_x)}
\widehat R_N(u\mid X_i),
\qquad
N_m=|\mathcal{M}_N(\tau_x)|.
\label{eq:mhat_scalar_irregular}
\end{equation}
The tuning parameters $h_x$, $h_0$, $\tau_x$, and
$\tau_{\rm den}$ are discussed in Section~\ref{subsec:implementation}.

\subsection{Regular design: \texorpdfstring{$(T,p,q)=(2,1,0)$}{(T,p,q)=(2,1,0)}}
\label{subsec:estimation_regular_T2p1}

In the regular design,
\[
Y_i=X_i\beta_i+D_i,
\qquad
X_i=
\begin{pmatrix}
X_{i1}\\ X_{i2}
\end{pmatrix},
\qquad
D_i=
\begin{pmatrix}
D_{i1}\\ D_{i2}
\end{pmatrix}.
\]
We use the same transformations as in Example~\ref{ex_regular_T2p1}. Define
\[
\tau_1(X_i)
=
\begin{pmatrix}
X_{i2} & -X_{i1}\\
0 & 0
\end{pmatrix},
\qquad
\tau_2(X_i)=\frac{X_i'}{X_i'X_i},\qquad X_i'X_i>0.
\]

Setting
$\tOne Y_i=\tau_1(X_i)Y_i$ and $Y_i^*:=e_1'\tOne Y_i$, the
condition $\tau_1(X_i)X_i=0$ gives
\[
Y_i^*=X_{i2}Y_{i1}-X_{i1}Y_{i2}=\lambda_i'D_i,
\qquad
\lambda_i=
\begin{pmatrix}
X_{i2}\\ -X_{i1}
\end{pmatrix}.
\]
Each observation therefore identifies $\varphi_D$ along the line generated by $\lambda_i$: for any $v\in\R$ and any $\lambda$ in the support of $\lambda_i$, 
\[
\varphi_D(v\lambda)
=
\E\!\left[
e^{\ii vY_i^*}
\,\middle|\,
\lambda_i=\lambda
\right].
\]

The second-step transformation isolates $\beta_i$:
\[
\tTwo Y_i
=
\tau_2(X_i)Y_i
=
\frac{X_i'Y_i}{X_i'X_i}
=
\beta_i+\frac{X_i'D_i}{X_i'X_i}.
\]

We estimate $\varphi_{\tTwo Y\mid X}(u\mid X_i)=\E[e^{\ii u\tTwo Y_i}\mid X_i]$
by kernel regression:
\begin{equation}
\widehat\varphi_{\tTwo Y\mid X}(u\mid X_i)
=
\sum_{k=1}^N
\omega_{k,i}^{(X)}
\exp\!\left\{
\ii u\frac{X_k'Y_k}{X_k'X_k}
\right\},
\label{eq:numerator_regular}
\end{equation}
where
\[
\omega_{k,i}^{(X)}
=
\frac{K_{h_X}(X_k-X_i)}
     {\sum_{s=1}^N K_{h_X}(X_s-X_i)}.
\]

To estimate $\varphi_D$, normalize the first-step direction:
\[
S_i=\frac{\lambda_i}{\|\lambda_i\|}\in\mathbb{S}^1,
\qquad
\|\lambda_i\|=\sqrt{X_{i1}^2+X_{i2}^2}.
\]
Since $Y_i^*/\|\lambda_i\|=S_i'D_i$ depends only on the unit
direction $S_i$, for any $\xi\in\R^2\setminus\{0\}$ with polar
decomposition $\xi=\|\xi\|\,s(\xi)$, $s(\xi)=\xi/\|\xi\|$,
\[
\E\!\left[
e^{\ii \|\xi\|\,Y_i^*/\|\lambda_i\|}
\,\middle|\,
S_i=s(\xi)
\right]
=
\varphi_D(\xi).
\]
We therefore estimate
$\varphi_D$ by smoothing over directions on $\mathbb{S}^1$:
\[
\widehat\varphi_D(\xi)
=
\begin{cases}
\displaystyle
\sum_{j=1}^N
\omega_j^{(S)}(s(\xi))
\exp\!\left\{\ii \|\xi\| \dfrac{Y_j^*}{\|\lambda_j\|}\right\},
& \xi\neq 0,\\[10pt]
1, & \xi=0,
\end{cases}
\]
where $\omega_j^{(S)}(s)=K(\|S_j-s\|/h_S)/\sum_{k=1}^N
K(\|S_k-s\|/h_S)$, and $\widehat\varphi_D(0)=1$ by the
property $\varphi_D(0)=1$.

For each observation $i$, we form the trimmed conditional ratio:
\begin{equation}
\widehat R_N(u\mid X_i)
=
\frac{
\widehat\varphi_{\tTwo Y\mid X}(u\mid X_i)
}{
\widehat\varphi_D\!\left(u\,\frac{X_i}{X_i'X_i}\right)
}
\mathbf{1}\!\left\{
\left|
\widehat\varphi_D\!\left(u\,\frac{X_i}{X_i'X_i}\right)
\right|>\tau_{\rm den}
\right\}.
\label{eq:Rhat_regular}
\end{equation}
Since $X_i'X_i>0$ for all $i$, the average is taken over the full
sample:%
\footnote{In the irregular design the average is restricted to
movers $\mathcal{M}_N(\tau_x)$ because $\tTwo Y_i=Y_i/X_i$ requires
$|X_i|\ge\tau_x$. Here $\tTwo Y_i=X_i'Y_i/(X_i'X_i)$ is well-defined
for all $i$ under the maintained assumption $X_i'X_i>0$, so no
such restriction is needed.}
\begin{equation}
\widehat m_N(u)
=
\frac{1}{N}
\sum_{i=1}^N \widehat R_N(u\mid X_i).
\label{eq:mhat_regular}
\end{equation}

\subsection{Implementation}
\label{subsec:implementation}

Algorithms~\ref{alg:irregular_estimator_short}
and~\ref{alg:regular_estimator_short} in Section \ref{appendix:algorithms} summarize the first-stage
implementation for the scalar irregular and regular designs,
respectively. In both cases the output is a vector of estimated
characteristic function values $\{\widehat{m}_N(u_\ell)\}_{\ell=1}^L$
on a prescribed frequency grid $\{u_\ell\}_{\ell=1}^L \subset
[-U_N,U_N]$. Algorithm~\ref{alg:smd_second_stage} in Section \ref{appendix:algorithms} describes the common
second stage, which takes these values as input and returns
$\widehat{f}_\beta$.

The estimator depends on a number of tuning parameters.
Three are common to all algorithms: the denominator trimming threshold
$\tau_{\mathrm{den}}$, the frequency truncation $U_N$, and the sieve
dimension $S$. The remaining parameters are design-specific.
In the scalar irregular design, the first stage uses the bandwidth
$h_0$ for estimating $\varphi_D$ from stayers, the bandwidth
$h_x$ for estimation $\varphi_{\tTwo Y | X}$ from movers, and the mover--stayer threshold $\tau_x$.
The role of $\tau_x$ is distinct from that of the bandwidths: it
defines the stayer sample and thereby localizes the denominator
estimator in regressor space. Crucially, $\widehat\varphi_D$ is
constructed once from the stayer subsample and subsequently
evaluated at the mover-specific arguments $u/X_i$.
In the regular design, the first stage uses the bandwidth $h_X$ for
kernel regression of the second-step transformed outcome on $X_i$, and
the bandwidth $h_S$ for smoothing over normalized directions
$S_i\in\mathbb{S}^1$; no mover--stayer partition is needed.

The integrals entering both the first-stage frequency evaluations and
the second-stage criterion are approximated on the grid
$\{u_\ell\}_{\ell=1}^L$ with positive weights $\{w_\ell\}_{\ell=1}^L$,
which determine the quadrature objects $\widehat{\Omega}$ and
$\widehat{V}$ in \eqref{eq:OmegaV_def}. The choice of all tuning
parameters is discussed next.

\subsection{Heuristic rate discussion and tuning parameter selection in the scalar irregular case}
\label{sec:heuristic_rates_irregular}

This subsection provides heuristic guidance for selecting $(\tau_x, h_0, U_N, \tau_{\mathrm{den}})$. Conditional on these choices, the remaining tuning parameters $(h_x, S)$ are selected by cross-validation. The tuning parameters play different roles in the estimator. The first group of tuning parameters primarily determines how the stayer conditioning event is approximated and how stable the deconvolution step is, while the second group primarily determines the degree of smoothing and the complexity of the sieve approximation.

Let the numerator and denominator estimation errors be defined as
\[
\rho_N
:=
\sup_{\substack{|u|\le U_N \\ |x|\ge \tau_x}}
\left|
\widehat{\varphi}_{\widetilde{\widetilde{Y}}\mid X}(u\mid x)
-
\varphi_{\widetilde{\widetilde{Y}}\mid X}(u\mid x)
\right|,
\qquad
b_N
:=
\sup_{|v|\le U_N/\tau_x}
\left|
\widehat{\varphi}_D(v)
-
\varphi_D(v)
\right|,
\]
and define the stability factor
\[
\Delta_N
:=
\inf_{\substack{|u|\le U_N \\ |x|\ge \tau_x}}
|\varphi_D(u/x)|.
\]

To separate truncation from estimation error, let \(f_{\beta,U_N}\) denote the population density obtained by restricting Fourier inversion to the frequency region \(|u|\le U_N\), and let \(f_{\beta,S}\) denote its sieve approximation. A heuristic decomposition then yields
\begin{equation}
\label{eq:total_error}
\|\widehat f_\beta - f_\beta\|_{L^2}
\;\lesssim\;
\underbrace{\|f_{\beta,S}-f_{\beta,U_N}\|_{L^2}}_{\text{sieve approximation}}
+
\underbrace{\Delta_N^{-1}\rho_N + \Delta_N^{-2}b_N}_{\text{propagated estimation error}}
+
\underbrace{\|f_{\beta,U_N}-f_\beta\|_{L^2}}_{\text{truncation bias}}.
\end{equation}
The final term captures the bias induced by restricting inversion to a bounded frequency region and is decreasing in \(U_N\). The middle terms reflect the ratio structure of the estimator: numerator errors enter linearly through multiplication by \(1/\varphi_D\), while denominator errors enter through the derivative of \(1/\varphi_D\), leading to amplification factors of order \(\Delta_N^{-1}\) and \(\Delta_N^{-2}\), respectively.

The decomposition in \eqref{eq:total_error} highlights how the tuning
parameters enter the different components of the estimation error. The
threshold $\tau_x$ and bandwidth $h_0$ determine the accuracy of the
denominator estimator through $b_N$, while the trimming parameter
$\tau_{\mathrm{den}}$ and frequency cutoff $U_N$ jointly determine the
stability factor $\Delta_N$. The cutoff $U_N$ also governs the
truncation bias induced by restricting
inversion to a bounded frequency region. The numerator bandwidth $h_x$
controls the estimation error $\rho_N$, and the sieve dimension $S$
determines the sieve approximation error.
The discussion below explains how these parameters are chosen.

In the scalar irregular design, $\varphi_D(v)=\E[e^{\ii v Y_i}\mid X_i=0]$.
If $P(X_i=0)>0$, this quantity could be estimated by an average over $\{i:X_i=0\}$, so that no bandwidth would be required. In the continuous case, however, $P(X_i=0)=0$, and conditioning on
$X_i=0$ must be approximated using observations with $|X_i|$ close to
zero. We therefore estimate $\varphi_D(v)$ using the stayer sample
$\{|X_i|<\tau_x\}$ together with kernel weights $h_0$. To ensure that the
resulting estimator mimics the infeasible conditional average at
$X_i=0$, we impose
\begin{equation}
\label{taux_h0}
\tau_x = o(h_0),
\end{equation}
so that $X_i/h_0 \to 0$ uniformly over $|X_i|<\tau_x$. Under this regime, the kernel weights are approximately constant over
the stayer window, and the estimator behaves like an averaging
estimator over $\{|X_i|<\tau_x\}$. Consequently, the relevant localization scale is
$\tau_x$, while $h_0$ serves only to ensure regularity of the estimator. 

The numerator bandwidth $h_x$ is selected by cross-validation. To ensure that numerator smoothing is not affected by the thresholding
operation, our heuristics presumes that $\tau_x = o(h_x)$, in which case, smoothing
in the numerator step is governed by $h_x$. 

The parameters $U_N$ and $\tau_{\mathrm{den}}$ must be chosen jointly,
as both govern the stability of the deconvolution ratio. Since
$|u|\le U_N$ and $|x|\ge \tau_x$ imply $|u/x|\le U_N/\tau_x$, the
denominator is evaluated over the frequency region
\[
|v| \le \frac{U_N}{\tau_x},
\]
so that
\[
\Delta_N \asymp \inf_{|v|\le U_N/\tau_x} |\varphi_D(v)|.
\]
Increasing $U_N$ expands this region and reduces truncation bias, but
also decreases $\Delta_N$ as $\varphi_D(v)$ decays in $|v|$, thereby
amplifying denominator noise. Stability therefore requires that
$\Delta_N$ remains at least of the same order as
$\tau_{\mathrm{den}}$. This implies that $U_N$ cannot grow faster than
allowed by the decay of $\varphi_D$. In particular, if
$|\varphi_D(v)| \asymp (1+|v|)^{-\alpha}$ (ordinary smooth), then
\[
\Delta_N \asymp (U_N/\tau_x)^{-\alpha},
\quad\text{so that}\quad
U_N \lesssim \tau_x\,\tau_{\mathrm{den}}^{-1/\alpha},
\]
whereas if $|\varphi_D(v)| \asymp \exp(-c|v|^\gamma)$ (supersmooth),
\[
\Delta_N \asymp \exp\!\left(-c(U_N/\tau_x)^\gamma\right),
\quad\text{so that}\quad
U_N \lesssim \tau_x \bigl[\log(1/\tau_{\mathrm{den}})\bigr]^{1/\gamma}.
\]
Thus $U_N$ and $\tau_{\mathrm{den}}$ are linked through the requirement
that the smallest denominator value over the working frequency region
remains of the same order as the trimming threshold.

These considerations lead to the following rules of thumb. Let $s_0$ index the smoothness of $x \mapsto \E[e^{\ii vY_i}\mid X_i=x]$ at $x=0$. 
Under 
$$\tau_x \asymp N^{-\kappa},$$ the averaging logic with $\tau_x=o(h_0)$ suggests the choice $$h_0 \asymp (N\tau_x)^{-1/(2s_0+1)}\asymp N^{-(1-\kappa)/(2s_0+1)},\qquad \kappa>\frac{1}{2(s_0+1)}.$$ 

The frequency trimming tuning parameter $U_N$ is pinned down by the
stability constraint and the tail behavior of $\varphi_D$. Under 
\[
\tau_{\mathrm{den}} \asymp N^{-\eta},\quad \eta >0,
\]
in the ordinary smooth case with, e.g., $|\varphi_D(v)| \asymp (1+|v|)^{-\alpha}$
for some $\alpha>0$, the stability condition yields polynomial growth for $U_N$:
\[
U_N \asymp N^{-\kappa + \eta/\alpha},\qquad 
\frac{\eta}{\alpha} > \kappa.
\]

In the supersmooth case with, e.g., $|\varphi_D(v)| \asymp \exp(-c|v|^\gamma)$ for some $c,\gamma>0$, to obtain a diverging $U_N$, the trimming
threshold must decay faster than polynomially. For example, if
\[
\tau_{\mathrm{den}} \asymp \exp(-N^{r}), \quad r>0,
\]
then
\[
U_N \asymp \tau_x N^{r/\gamma}
\asymp N^{-\kappa + r/\gamma}, \qquad 
\frac{r}{\gamma} > \kappa.
\]

Under these choices, each term in \eqref{eq:total_error} is controlled:
the denominator error $b_N$ is governed by $\tau_x$, the stability factor
is maintained at the order of $\tau_{\mathrm{den}}$, and the truncation
bias is reduced by allowing $U_N$ to grow at the maximal rate compatible
with stability.

\subsection{Cross-Validation}
\label{sec:cv_irregular}

A practical way to select the remaining tuning parameters is via
$K$-fold cross-validation. Conditional on $(\tau_x, h_0, U_N, \tau_{\mathrm{den}})$ choices, cross-validation is used to
select the numerator bandwidth $h_x$ and the sieve dimension $S$.

We develop the cross-validation procedure below for the irregular design.
In the regular design, the directional bandwidth $h_S$, which governs the
estimation of $\varphi_D$ by smoothing over $S_i\in\mathbb{S}^{d-1}$,
plays the role of $h_0$ and is fixed; the numerator bandwidth
$h_X$ plays the role of $h_x$ and is selected jointly with $S$ by
cross-validation.

In the irregular design we fix $\tau_x$ using a rule motivated by
\begin{equation}
\label{taux_ref}
\tau_x = c_\tau \cdot \tau_x^{\mathrm{ref}}, \qquad
\tau_x^{\mathrm{ref}}
=
N^{-\kappa}\min\!\left(\mathrm{SD}(X_i),\;
\frac{\mathrm{IQR}(X_i)}{1.34}\right), \qquad \kappa \in (0,1/2),
\end{equation}
where $\kappa$ governs the size of the stayer window. In the simulations and application, we use $\kappa = 1/3$.

In the regime $\tau_x = o(h_0)$, the bandwidth $h_0$ is set 
\begin{equation}
\label{h0_ref}
h_0 = c_0 \cdot N^{-(1-\kappa)/(2s_0+1)},\qquad c_0 > 0.
\end{equation}
Setting $s_0=1$, yields $h_0 \asymp N^{-2/9}$ when $\kappa=1/3$.

To construct the candidate grid for $h_x$, we compute a reference bandwidth
\begin{equation}
\label{hx_ref}
h_x^{\mathrm{ref}}
=
0.9 \cdot N_m^{-1/5}\,\min\!\left(\mathrm{SD}(X_i^{\mathcal M}),\;
\frac{\mathrm{IQR}(X_i^{\mathcal M})}{1.34}\right),
\end{equation}
where $X_i^{\mathcal M}$ denotes the mover subsample and
$N_m = |\mathcal M_N(\tau_x)|$. The candidate grid is
$\Theta_{h_x} = \{c \cdot h_x^{\mathrm{ref}} : c \in \mathcal C\}$ for a
finite grid of positive multipliers $\mathcal C$.

Finally, we set $\tau_{\mathrm{den}}$ at a small value (e.g., $10^{-4}$),
which determines the minimal magnitude of the denominator that is deemed
numerically stable. Conditional on this choice, the cutoff $U_N$ is
selected using the stability considerations in
Section~\ref{sec:heuristic_rates_irregular}. In particular, $U_N$ is
increased conservatively until the empirical denominator
$|\widehat{\varphi}_D(v)|$ becomes small relative to
$\tau_{\mathrm{den}}$, so that the working frequency region remains
restricted to values for which the ratio is well behaved.

For each sample size $N$, fix a frequency grid
$\{u_j\}_{j=1}^J \subset [-U_N,U_N]$
with associated nonnegative weights $\{\varpi_j\}_{j=1}^J$, where the
cutoff $U_N$ is determined as above and is held fixed across candidate
values of $\theta=(h_x,S)$ in the cross-validation search.

\subsubsection{Stabilized cross-validation}

To describe the cross-validation procedure, let $\mathcal I_1, \dots, \mathcal I_K$ be a partition of the sample indices
$\{1,\dots,N\}$, and let
$\mathcal I_{-k} = \{1,\dots,N\}\setminus \mathcal I_k$ denote the training
sample associated with fold~$k$.

For each fold $k$ and candidate $\theta$, define the training and validation
stayer and mover sets
\[
\mathcal S_{-k}(\tau_x)
=
\{i\in \mathcal I_{-k}: |X_i|<\tau_x\},
\qquad
\mathcal M_{-k}(\tau_x)
=
\{i\in \mathcal I_{-k}: |X_i|\ge \tau_x\},
\]
\[
\mathcal S_k(\tau_x)
=
\{i\in \mathcal I_k: |X_i|<\tau_x\},
\qquad
\mathcal M_k(\tau_x)
=
\{i\in \mathcal I_k: |X_i|\ge \tau_x\},
\]
with cardinalities
$N_{m,-k}=|\mathcal M_{-k}(\tau_x)|$ and $N_{m,k}=|\mathcal M_k(\tau_x)|$.
Since $\tau_x$ is fixed, this partition does not vary across candidates
$\theta$.

On the training sample, estimate
\begin{equation}
\widehat\varphi_D^{(-k)}(v)
=
\sum_{\ell\in \mathcal S_{-k}(\tau_x)}
\omega_\ell^{(0,-k)} e^{\ii vY_\ell},
\qquad
\omega_\ell^{(0,-k)}
=
\frac{K_{h_0}(X_\ell)}
     {\sum_{s\in \mathcal S_{-k}(\tau_x)}K_{h_0}(X_s)},
\label{eq:phid_train_cv}
\end{equation}
and on the validation sample, estimate
\begin{equation}
\widehat\varphi_D^{(k)}(v)
=
\sum_{\ell\in \mathcal S_k(\tau_x)}
\omega_\ell^{(0,k)} e^{\ii vY_\ell},
\qquad
\omega_\ell^{(0,k)}
=
\frac{K_{h_0}(X_\ell)}
     {\sum_{s\in \mathcal S_k(\tau_x)}K_{h_0}(X_s)}.
\label{eq:phid_valid_cv}
\end{equation}
Since $h_0$ and $\tau_x$ are held fixed, $\widehat\varphi_D^{(-k)}$ and
$\widehat\varphi_D^{(k)}$ are computed once per fold and do not vary across
candidate values of $\theta=(h_x,S)$.

For each candidate bandwidth $h_x$ and each mover
$i\in \mathcal M_{-k}(\tau_x)$, estimate
\begin{equation}
\widehat\varphi_{\tTwo Y\mid X}^{(-k)}(u\mid X_i)
=
\sum_{r\in \mathcal M_{-k}(\tau_x)}
\omega_{r,i}^{(x,-k)}
\exp\!\left\{\ii u\frac{Y_r}{X_r}\right\}, \quad \omega_{r,i}^{(x,-k)}
=
\frac{K_{h_x}(X_r-X_i)}
     {\sum_{s\in \mathcal M_{-k}(\tau_x)}K_{h_x}(X_s-X_i)}.
\label{eq:phiy_train_cv}
\end{equation}
Form the trimmed training ratio
\begin{equation}
\widehat R_N^{(-k)}(u\mid X_i)
=
\frac{
\widehat\varphi_{\tTwo Y\mid X}^{(-k)}(u\mid X_i)
}{
\widehat\varphi_D^{(-k)}(u/X_i)
}
\mathbf 1\!\left\{
\bigl|\widehat\varphi_D^{(-k)}(u/X_i)\bigr|>\tau_{\mathrm{den}}
\right\},
\label{eq:rhat_train_cv}
\end{equation}
and average over the training movers:
\begin{equation}
\widehat m_N^{(-k)}(u)
=
\frac{1}{N_{m,-k}}
\sum_{i\in \mathcal M_{-k}(\tau_x)}
\widehat R_N^{(-k)}(u\mid X_i).
\label{eq:mhat_train_cv}
\end{equation}
Apply the second-stage sieve minimum distance estimator with sieve
dimension $S$ to $\{\widehat m_N^{(-k)}(u_j)\}_{j=1}^J$ to obtain the
training-fold sieve coefficient vector $\widehat\pi^{(-k)}$, and define the
implied training-fold characteristic function estimator
\begin{equation}
\widehat\varphi_\beta^{(-k)}(u)
=
\varphi_S(u;\widehat\pi^{(-k)}).
\label{eq:phibeta_train_cv}
\end{equation}

Standard cross-validation performs poorly in deconvolution problems
because the validation target inherits the instability of the estimator.
When the candidate bandwidth $h_x$ is small, the validation ratio
$\widehat{R}_N^{(k)}(u \mid X_i)$ is dominated by high-frequency noise
amplified through division by $\widehat{\varphi}_D^{(k)}(u/X_i)$. If the
training estimator uses the same small $h_x$, both sides of the
cross-validation comparison are noisy at high frequencies, and the
criterion effectively compares noise to noise.\footnote{\label{fn:cv_lit}This issue is
well documented in the bandwidth selection literature; see
\citet{ScottTerrell1987} for the distinction between biased and unbiased
cross-validation, and \citet{HallMarronPark1992} for a discussion of how
presmoothing affects the bias--variance tradeoff. \citet{Kent2024}
develops a related stabilized cross-validation approach for
smoothness-penalized deconvolution, showing that standard CV leads to
severe undersmoothing in ill-posed problems.}

We adopt a stabilized cross-validation approach in which the validation
target is constructed using a fixed pilot bandwidth $g_{\mathrm{pilot}}$
that does not vary with the candidate $h_x$. This produces an
oversmoothed, low-variance validation target, so that the CV criterion
compares candidate estimators against a stable reference rather than
against an equally noisy holdout estimate. The pilot bandwidth is set to
\begin{equation}
\label{eq:g_pilot}
g_{\mathrm{pilot}} = c_{\mathrm{pilot}} \cdot h^{\mathrm{ref}},
\qquad
h^{\mathrm{ref}} = 0.9 \cdot \min\!\left(\mathrm{SD}(X_i^{\mathcal M}),\;
\frac{\mathrm{IQR}(X_i^{\mathcal M})}{1.34}\right) N_m^{-1/r},
\end{equation}
where $c_{\mathrm{pilot}} \in [1.5, 3]$ and $r > 5$. Since $r > 5$ implies
$N^{-1/r} > N^{-1/5}$, the reference bandwidth $h^{\mathrm{ref}}$ decays more
slowly than the standard Silverman rate, so that $g_{\mathrm{pilot}}$ is
larger than a typical bandwidth choice. This produces an \emph{oversmoothed}
validation target, which provides a stable reference against which candidate
estimators are compared.

\paragraph{Heuristic justification.}
Let $\theta=(h_x,S)$ denote the smoothing and sieve parameters. The infeasible
finite-grid risk is
\[
R_k(\theta)
=
\sum_{j=1}^J \varpi_j
\left|
\widehat\varphi_{\beta,S}^{(-k)}(u_j;\theta)
-
\varphi_\beta(u_j)
\right|^2 .
\]
Cross-validation replaces the unknown target $\varphi_\beta$ by a validation
estimate constructed on the hold-out fold. In the present deconvolution problem,
the natural validation object is itself a ratio estimator,
\[
\widehat m_N^{(k)}(u;h)
=
\frac{1}{N_{m,k}}
\sum_{i\in\mathcal M_k}
\frac{
\widehat\varphi_{\tilde Y\mid X}^{(k)}(u\mid X_i;h)
}{
\widehat\varphi_D^{(k)}(u/X_i)
}
\mathbf{1}\bigl\{
|\widehat\varphi_D^{(k)}(u/X_i)|>\tau_{\rm den}
\bigr\}.
\]
For a generic bandwidth $h$, write heuristically
\[
\widehat m_N^{(k)}(u;h)
\approx
\varphi_\beta(u)+b_h(u)+\xi_h^{(k)}(u),
\]
where $b_h(u)$ is smoothing bias, $\xi_h^{(k)}(u)$ is stochastic error, and
the approximation reflects a first-order expansion of the ratio estimator. The
difficulty is that $\xi_h^{(k)}(u)$ is amplified by division by
$\widehat\varphi_D^{(k)}(u/X_i)$, especially at frequencies for which the
denominator is small. Thus a validation target constructed with a small bandwidth
can be dominated by high-frequency noise.

We therefore construct the validation target using a fixed oversmoothed pilot
bandwidth,
\[
\widehat m_N^{(k)}(u; g_{\mathrm{pilot}})
\approx
\varphi_\beta(u)+b_g(u)+\xi_g^{(k)}(u).
\]
The purpose of $g_{\mathrm{pilot}}$ is to reduce the numerator noise before that noise is amplified by the
inverse characteristic function. Since the denominator $\widehat{\varphi}_D^{(k)}$
does not depend on $g_{\mathrm{pilot}}$, a larger pilot bandwidth reduces the
variance of the numerator and hence the variance of the ratio.

Substitution into the cross-validation criterion gives
\begin{align*}
\mathrm{CV}_k(\theta)
&=
R_k(\theta)
+
\sum_{j=1}^J \varpi_j |b_g(u_j)+\xi_g^{(k)}(u_j)|^2
\\ &-
2\textrm{Re}\sum_{j=1}^J \varpi_j
\Bigl(
\widehat\varphi_{\beta,S}^{(-k)}(u_j;\theta)
-
\varphi_\beta(u_j)
\Bigr)
\overline{(b_g(u_j)+\xi_g^{(k)}(u_j))}.
\end{align*}
The second term on the right hand side is common across candidate values of $\theta$ and therefore does
not affect the ranking of tuning parameters. The last term is the relevant
distortion. Sample splitting mitigates its stochastic component in the sense that there is no systematic bias, since
$\widehat\varphi_{\beta,S}^{(-k)}$ and $\xi_g^{(k)}$ are computed on disjoint
folds. The remaining pilot-bias component is the price paid for stabilization:
the pilot bandwidth is chosen to trade a small, smooth bias against a substantial
reduction in the variance of the validation target.

Thus the stabilized criterion is intended to approximate the infeasible risk up
to a candidate-invariant term and a reduced cross term. This does not establish
optimality of the selected tuning parameters, but it explains why stabilizing the
validation target is preferable to comparing the estimator with another noisy
inverse estimate. A related approach with formal consistency guarantees is
developed by \citet{Kent2024}; that method estimates risk for a smaller
hypothetical sample size and requires knowledge of asymptotic rates to rescale
the selected parameter.

\medskip

For each validation mover $i\in \mathcal M_k(\tau_x)$, estimate the
numerator using the pilot bandwidth:
\begin{equation}
\widehat\varphi_{\tTwo Y\mid X}^{(k)}(u\mid X_i; g_{\mathrm{pilot}})
=
\sum_{r\in \mathcal M_k(\tau_x)}
\omega_{r,i}^{(\mathrm{pilot},k)}
\exp\!\left\{\ii u\frac{Y_r}{X_r}\right\},
\quad
\omega_{r,i}^{(\mathrm{pilot},k)}
=
\frac{K_{g_{\mathrm{pilot}}}(X_r-X_i)}
     {\sum_{s\in \mathcal M_k(\tau_x)}K_{g_{\mathrm{pilot}}}(X_s-X_i)}.
\label{eq:phiy_valid_cv_pilot}
\end{equation}
Form the validation ratio
\begin{equation}
\widehat R_N^{(k)}(u\mid X_i; g_{\mathrm{pilot}})
=
\frac{
\widehat\varphi_{\tTwo Y\mid X}^{(k)}(u\mid X_i; g_{\mathrm{pilot}})
}{
\widehat\varphi_D^{(k)}(u/X_i)
}
\mathbf 1\!\left\{
\bigl|\widehat\varphi_D^{(k)}(u/X_i)\bigr|>\tau_{\mathrm{den}}
\right\},
\label{eq:rhat_valid_cv_pilot}
\end{equation}
and average over the validation movers:
\begin{equation}
\widehat m_N^{(k)}(u; g_{\mathrm{pilot}})
=
\frac{1}{N_{m,k}}
\sum_{i\in \mathcal M_k(\tau_x)}
\widehat R_N^{(k)}(u\mid X_i; g_{\mathrm{pilot}}).
\label{eq:mhat_valid_cv_pilot}
\end{equation}
This construction uses the pilot bandwidth $g_{\mathrm{pilot}}$ for the
numerator smoothing while the denominator $\widehat\varphi_D^{(k)}$ continues
to be estimated using the fixed bandwidth $h_0$. Since the validation target
does not depend on the candidate $h_x$, it can be computed once per fold
outside the loop over candidate bandwidths, which provides a computational
advantage.

\subsubsection{Repeated cross-validation and selection}

To reduce sensitivity to the random fold assignment, we employ repeated
cross-validation. Let $n_{\mathrm{rep}}$ denote the number of repetitions.
For each repetition $r = 1, \ldots, n_{\mathrm{rep}}$, we draw an independent
random partition of the sample into $K$ folds and compute the fold-specific
holdout discrepancy
\begin{equation}
\mathrm{CV}_k^{(r)}(\theta)
=
\sum_{j=1}^J
\varpi_j
\left|
\widehat m_N^{(k)}(u_j; g_{\mathrm{pilot}})
-
\widehat\varphi_\beta^{(-k)}(u_j)
\right|^2,
\label{eq:cvk_def}
\end{equation}
where $\theta = (h_x, S)$. The repetition-specific cross-validation score is
$\mathrm{CV}^{(r)}(\theta) = K^{-1}\sum_{k=1}^{K}\mathrm{CV}_k^{(r)}(\theta)$.

Aggregating across repetitions, define
\begin{equation}
\label{eq:cv_mean_sd}
\overline{\mathrm{CV}}(\theta)
=
\frac{1}{n_{\mathrm{rep}}} \sum_{r=1}^{n_{\mathrm{rep}}} \mathrm{CV}^{(r)}(\theta),
\qquad
\mathrm{SD}(\theta)
=
\left(
\frac{1}{n_{\mathrm{rep}}-1} \sum_{r=1}^{n_{\mathrm{rep}}}
\bigl(\mathrm{CV}^{(r)}(\theta) - \overline{\mathrm{CV}}(\theta)\bigr)^2
\right)^{1/2},
\end{equation}
and the standard error $\mathrm{SE}(\theta) = \mathrm{SD}(\theta) / \sqrt{n_{\mathrm{rep}}}$.

\subsubsection{One-standard-error rule}

In ill-posed inverse problems, favoring simpler models can improve stability.
We adopt a one-standard-error rule: let $\theta^* = \arg\min_\theta \overline{\mathrm{CV}}(\theta)$
denote the minimizer of the mean CV score, and let $\mathrm{SE}^* = \mathrm{SE}(\theta^*)$
be the standard error at the minimum. Define the set of near-optimal candidates
\begin{equation}
\label{eq:one_se_set}
\Theta_{\mathrm{1SE}}
=
\bigl\{\theta \in \Theta_{h_x} \times \Theta_S :
\overline{\mathrm{CV}}(\theta) \leq \overline{\mathrm{CV}}(\theta^*) + \mathrm{SE}^*
\bigr\}.
\end{equation}
Among candidates in $\Theta_{\mathrm{1SE}}$, we select the pair with the
largest $h_x$ (most smoothing) and, among those, the smallest $S$
(simplest sieve). This lexicographic ordering favors stability: larger bandwidths
reduce variance in the numerator estimation, and smaller sieve dimensions reduce
the complexity of the second-stage projection.

The selected tuning parameters are
\begin{equation}
\label{eq:cv_def}
\widehat\theta
=
\arg\max_{(h_x, S) \in \Theta_{\mathrm{1SE}}} h_x,
\quad
\text{with ties broken by smallest } S.
\end{equation}

Algorithm~\ref{alg:pcv_tuning} in
Section~\ref{appendix:algorithms} summarizes the procedure.

\begin{remark}[Feasibility of the cross-validation environment]
\label{rem:CVfeasibility}

Deconvolution estimators are ill-posed, and the trimmed ratio can become
unstable or degenerate when the denominator $\widehat\varphi_D^{(k)}$ is small.
We impose hard feasibility conditions that must be satisfied before the CV
search proceeds. These conditions depend only on the fixed parameters
$(\tau_x,h_0,\tau_{\mathrm{den}})$ and the data partition. This approach is
standard in ill-posed inverse problems; see \citet{CarrascoFlorens2011}.

The precise definitions of the feasibility conditions are given in
Appendix~\ref{appendix:cv_feasibility}. In particular, the CV environment
is deemed infeasible if (i) validation mover sets are empty, (ii) the average
trimming fraction exceeds $\rho_{\max}$, (iii) the denominator instability
$\Gamma_k := \max_{u,i} |\widehat\varphi_D(u/X_i)|^{-1}$ exceeds $\Gamma_{\max}$,
or (iv) for some fold $k$ and frequency $u_j$, all validation movers are trimmed.

If any of these conditions fails, the procedure stops and the user must
revise $\tau_x$ or $\tau_{\mathrm{den}}$ before proceeding. Otherwise, the
CV search is conducted over the candidate grid
$\Theta_{h_x}\times\Theta_S$.
\end{remark}

\subsection{Bootstrap Inference}
\label{sec:bootstrap}

We assess sampling uncertainty using a nonparametric pairs bootstrap.\footnote{The method here is provisional.}
For each bootstrap replication $b = 1, \ldots, B$:
\begin{enumerate}
\item Draw a bootstrap sample $(Y_i^{*b}, X_i^{*b})_{i=1}^N$ by resampling
the original observations with replacement.
\item Re-estimate the density $\widehat{f}_\beta^{*b}$ using the tuning
parameters $(\tau_x, h_0, h_x, S)$ selected by cross-validation on the
original sample. The stayer--mover partition is recomputed for each
bootstrap sample, but the tuning parameters are held fixed.
\item Post-process the bootstrap estimate: truncate negative values at
zero and renormalize to integrate to one.
\end{enumerate}

\noindent
Inference is therefore for the post-processed density estimator. In the application, we use $B = 499$ replications.

\subsubsection{Confidence intervals}

We construct pointwise confidence intervals using the basic (reverse-percentile)
bootstrap method. For each evaluation point $b$ in the density grid:
\begin{enumerate}
\item Compute centered bootstrap deviations:
$\widetilde{f}^{*b}(b) = \widehat{f}_\beta^{*b}(b) - \widehat{f}_\beta(b)$.
\item Compute the $\alpha/2$ and $(1-\alpha/2)$ quantiles of the centered
distribution across bootstrap replications:
\[
\widetilde{q}_{\alpha/2}(b), \quad \widetilde{q}_{1-\alpha/2}(b).
\]
\item Form the confidence interval by reversing the quantiles:
\[
\mathrm{CI}_{1-\alpha}(b)
=
\bigl[\widehat{f}_\beta(b) - \widetilde{q}_{1-\alpha/2}(b),\;
\widehat{f}_\beta(b) - \widetilde{q}_{\alpha/2}(b)\bigr].
\]
\end{enumerate}

\noindent
This is equivalent to the formula
$\mathrm{CI}_{1-\alpha}(b) = [2\widehat{f}_\beta(b) - \widehat{q}_{1-\alpha/2}^*(b),\;
2\widehat{f}_\beta(b) - \widehat{q}_{\alpha/2}^*(b)]$,
where $\widehat{q}_\tau^*(b)$ denotes the $\tau$-quantile of the bootstrap
estimates $\{\widehat{f}_\beta^{*1}(b), \ldots, \widehat{f}_\beta^{*B}(b)\}$.
After construction, confidence bands are truncated below at zero to respect
the non-negativity constraint on densities.

\subsubsection{Standard errors and moments}

Bootstrap standard errors are computed as
\[
\widehat{\mathrm{SE}}(b)
=
\mathrm{SD}\bigl(\widehat{f}_\beta^{*1}(b), \ldots, \widehat{f}_\beta^{*B}(b)\bigr).
\]
Since the standard deviation is translation-invariant, this quantity does not
depend on whether the bootstrap distribution is centered.

Moments of the random coefficient distribution are computed by numerical
integration over the post-processed density:
\[
\widehat{\mathbb{E}}[\beta_i]
=
\int b\, \widehat{f}_\beta(b)\, db,
\qquad
\widehat{\mathrm{Var}}(\beta_i)
=
\int b^2\, \widehat{f}_\beta(b)\, db - \bigl(\widehat{\mathbb{E}}[\beta_i]\bigr)^2.
\]
Bootstrap standard errors and confidence intervals for moments are constructed
analogously, applying the same integration to each bootstrap density
$\widehat{f}_\beta^{*b}$.

\begin{remark}[Inference]
\label{rem:inference}

The pairs bootstrap procedure in Section~\ref{sec:bootstrap} provides practical pointwise inference for the density estimator, conditional on the selected tuning parameters.

For kernel deconvolution estimators, \citet{BissantzDumbgenHolzmannMunk2007} establish conditions under which the nonparametric bootstrap is consistent for constructing uniform confidence bands when the error distribution is known. \citet{KatoSasaki2018} extend this to unknown error distributions using a multiplier bootstrap. For sieve estimators in ill-posed inverse problems, \citet{ChenChristensen2018} develop a sieve score bootstrap that delivers asymptotically exact uniform confidence bands for nonlinear functionals. Our setting differs in that estimation proceeds via a two-stage procedure with a generated nuisance function; the obstacles to extending these bootstrap results are discussed below.

An alternative approach would exploit the moment restrictions implied by the characteristic function factorization. Multiplying the factorization \eqref{eq:cf_factor} by a bounded function $g(X_i)$ and integrating yields
\[
E\left[ g(X_i) \left( e^{\imath u \tilde{Y}_i} - \phi_{\beta|X}(u \mid X_i) \cdot \phi_D(u/X_i) \right) \right] = 0,
\]
which, after substituting a first-step estimator $\widehat{\phi}_D$, is linear in a sieve representation of the conditional characteristic function. This structure is related to \citet{KatoSasakiUra2021}, who reformulate Kotlarski's identity as linear complex-valued moment restrictions and conduct inference by inverting a supremum test statistic within a Hermite sieve framework. Inference via moment inversion shares features with high-dimensional Gaussian approximation problems \citep{ChernozhukovChetverikov2017}, but existing results do not cover the present setting without additional arguments. Three features are critical.

First, the unknown characteristic function $\phi_{\beta|X}(u|x)$ depends on both the frequency $u$ and the conditioning value $x$. The moment system therefore indexes a function of two arguments, requiring a growing collection of moments over frequencies and covariates. In particular, sieve approximations involve tensor-product bases in $(u,x)$, leading to a dimension that grows with the product of the number of basis functions in each argument.

Second, the moments depend on the first-step estimator $\widehat{\phi}_D$. The moment restriction is nonorthogonal with respect to this nuisance: perturbing $\phi_D$ changes the moment at first order, with the perturbation entering multiplicatively through $\phi_{\beta|X}(u|X_i)\cdot (\widehat{\phi}_D - \phi_D)(u/X_i)$. Equivalently, the Gateaux derivative of the moment with respect to $\phi_D$ at the truth is nonzero. Although high-dimensional Gaussian approximation results have been extended to settings with estimated nuisance functions, such extensions typically rely on orthogonality or sample-splitting conditions that eliminate the first-order effect of nuisance estimation. Here, orthogonality does not hold, so the first-step contribution must be explicitly controlled.\footnote{The stayer and mover samples are disjoint by construction, but this design-induced separation is not Neyman orthogonality. One could additionally cross-fit by partitioning all observations into $K$ folds: stayers outside fold $k$ are used to estimate $\widehat{\phi}_D^{(-k)}$, and movers inside fold $k$ use $\widehat{\phi}_D^{(-k)}$ when evaluating the moment restrictions. This removes own-fold dependence between the generated nuisance and the moment observations, but does not restore orthogonality. Because $\phi_D$ enters the moment multiplicatively, first-step estimation error remains first order in the asymptotic expansion.}

Third, the effective index set is data-dependent. The first-step estimator is evaluated at random arguments $u/X_i$: when $|X_i|$ is small, the evaluation point lies in the high-frequency region where $\widehat{\phi}_D$ is less precisely estimated. Moreover, observations are retained only when $|\widehat{\phi}_D(u/X_i)|$ exceeds a threshold, so the set of active moment conditions depends on the first-step estimator. This induces a data-dependent function class with discontinuous inclusion rules, complicating Gaussian approximation because standard entropy bounds apply to fixed classes. Addressing this would require controlling the empirical process uniformly over a family of classes indexed by the first-step estimator.

These features interact: the random evaluation points $u/X_i$ determine where the generated nuisance $\widehat{\phi}_D$ is evaluated, and at each such point the lack of orthogonality transmits first-step estimation error into the moment at first order. Because the target is conditional, this interaction propagates across a growing collection of moments indexed by both $u$ and $x$. Taken together, these considerations suggest that rigorous inference, whether via bootstrap or moment inversion, would require careful extensions of existing deconvolution, moment-inversion, and high-dimensional Gaussian approximation tools.
\end{remark}

\section{Monte Carlo: Finite sample performance}
\label{sec:simulations}

\subsection{Irregular design $(T,p,q)=(1,1,0)$}
\label{sec:sim_irregular}

We consider the following data generating process.  For each $i=1,\dots,N$,
\begin{equation}\label{eq:dgp_irreg}
y_{it} = \alpha_i + x_{it}\beta_i + u_{it}, \qquad t=1,2,
\end{equation}
where the covariates, random coefficients, and error terms are generated as follows:
\begin{alignat*}{3}
x_{i1} &\sim \mathcal N(0,1), &\qquad x_{i2} &= z_i + x_{i1}, &\quad z_i&\sim\mathcal N(0,4), \\
\alpha_i &\sim \mathcal N(x_{i1}+x_{i2},\,1), \\
\beta_i &= \zeta_i\,\varepsilon_{\beta,i}, 
&\qquad \zeta_i &= 1+\delta(x_{i2}-x_{i1}), &\quad \delta&=0.1, \\
u_{i1} &= \varepsilon_{1i}, &\qquad u_{i2} &= \varepsilon_{2i}+\theta\,\varepsilon_{1i}, &\quad \theta&=0.7,
\end{alignat*}
with $\varepsilon_{\beta,i}\sim F_{\varepsilon_\beta}$, $\varepsilon_{1i}\sim F_{\varepsilon_{1}}$, and $\varepsilon_{2i}\sim F_{\varepsilon_{2}}$, where $\mathrm{Var}(\varepsilon_{1i})=\sigma^2_{\varepsilon_1}=1$ and $\mathrm{Var}(\varepsilon_{2i})=\sigma^2_{\varepsilon_2}=2$.  The innovations $\varepsilon_{1i}$ and $\varepsilon_{2i}$ are independent and mean-zero; their distributions $F_{\varepsilon_{1}}$, $F_{\varepsilon_{2}}$, and $F_{\varepsilon_\beta}$ vary across specifications as described below.
After first-differencing, the model reduces to
\[
Y_i = X_i\beta_i + D_i,
\]
where $X_i = z_i\sim\mathcal N(0,4)$ and $D_i = \varepsilon_{2i}+(\theta-1)\varepsilon_{1i}$.
The random coefficient $\beta_i$ is a scale mixture: since $\beta_i = \zeta_i\varepsilon_{\beta,i}$
with $\zeta_i = 1+\delta X_i$, $\beta_i$ and~$X_i$ are statistically dependent. The error terms $u_{it}$
are serially correlated through the parameter~$\theta$.

We consider four specifications, varying the distribution of the innovation $\varepsilon_{\beta,i}$, hence the distribution of the CRC, and the distribution of the error terms:
\begin{itemize}
\item[(a)] Symmetric CRC and supersmooth errors: $F_{\varepsilon_\beta} = \mathcal N(0,1)$, and Gaussian errors $F_{\varepsilon_t} = \mathcal N(0,\sigma^2_{\varepsilon_t}),\,t=1,2$;
\item[(b)] Skewed CRC and supersmooth errors: $F_{\varepsilon_\beta} = \mathrm{SN}(0,\omega,\lambda)$ with scale $\omega=2$ and shape $\lambda=4$, and Gaussian errors $F_{\varepsilon_t} = \mathcal N(0,\sigma^2_{\varepsilon_t}),\,t=1,2$;
\item[(c)] Bimodal CRC and supersmooth errors: $F_{\varepsilon_\beta} = \tfrac{1}{2}\mathcal N(-1,0.25)+\tfrac{1}{2}\mathcal N(1,0.15)$, and Gaussian errors $F_{\varepsilon_t} = \mathcal N(0,\sigma^2_{\varepsilon_t}),\,t=1,2$;
\item[(d)] Bimodal CRC and ordinary smooth errors: same $F_{\varepsilon_\beta}$ as~(c), but $F_{\varepsilon_1}$ and $F_{\varepsilon_2}$ are Laplace with scale parameters $b_1 = \sqrt{\sigma^2_{\varepsilon_1}/2}$ and $b_2 = \sqrt{\sigma^2_{\varepsilon_2}/2}$, so that $\mathrm{Var}(D_i)$ is identical to specifications~(a)--(c).
\end{itemize}
In specifications~(a)--(c), the disturbance $D_i = \varepsilon_{2i}+(\theta-1)\varepsilon_{1i}$ is a linear combination of Gaussians, so its characteristic function $\varphi_D$ decays exponentially (supersmooth).  In specification~(d), the Laplace innovations yield $\varphi_D(t) = O(|t|^{-4})$ (ordinary smooth), while the distribution of $\beta_i$ is unchanged.  Comparing (c) and~(d) isolates the effect of the smoothness class of $D_i$ on the estimator's ability to recover multimodal features.

In each case, the unconditional density of~$\beta_i$ is the scale mixture
\[
f_\beta(b) = \int |s|^{-1}\,f_{\varepsilon_\beta}(b/s)\,f_\zeta(s)\,ds,
\]
where $f_\zeta$ is the density of $\zeta_i\sim\mathcal N(1,4\delta^2)$.  When $\delta$ is small,
the scaling factor~$\zeta_i$ remains close to one with high probability, and the shape of
$f_{\varepsilon_\beta}$ is largely preserved in the unconditional density~$f_\beta$.

\paragraph{Implementation.}
We generate $N=2000$ observations and run $100$ Monte Carlo replications.  The
stayer threshold is set to $\tau_x = c_\tau\cdot\tau_x^{\mathrm{ref}}$ with
$\tau_x^{\mathrm{ref}}$ defined in \eqref{taux_ref}.  In specifications~(a) and~(b), $c_\tau=4$, yielding approximately $25\%$ of observations classified as stayers across replications; in the bimodal specifications~(c) and~(d), $c_\tau=5$, yielding approximately $31\%$ stayers.  The bandwidth $h_0$ for the stayer denominator estimator is set
proportional to the threshold, $h_0 = c_0\cdot\tau_x$ with $c_0 = 1$.\footnote{This deviates from the discussion in Section~\ref{sec:heuristic_rates_irregular}. The discussion there motivates the roles of \(\tau_x\) and \(h_0\), while the implementation localizes the denominator estimator on the same near-stayer region used to define the stayer sample.} The denominator trimming threshold is $\tau_{\mathrm{den}}=10^{-4}$, and the frequency grid consists of $L=101$ equally spaced points on $[-U_N,U_N]$ with $U_N=4$.

The tuning parameters $(h_x,S)$ are selected by $5$-fold cross-validation as
described in Section~\ref{sec:cv_irregular} and summarized in Algorithm~\ref{alg:pcv_tuning}.  In specifications~(a) and~(b), the bandwidth
$h_x$ is global (fixed across movers), and the candidate grid consists of the multiples
$$\Theta_{h_x}=\{0.5,\,0.75,\,1,\,1.5,\,2\}\times h_x^{\mathrm{ref}},$$ where $h_x^{\mathrm{ref}}$
is Silverman's rule on the mover subsample.  The candidate grid for $S$ is
$\{3,5,\dots,15\}$ in specification~(a) and $\{3,5,\dots,19\}$ in specification~(b).
The weighting function $\nu$ is the standard normal density.

In the bimodal specifications~(c) and~(d), the global bandwidth is replaced by a $k$-nearest-neighbor
($k$-NN) adaptive rule (see, e.g., \cite{GaoOhViswanath2017}).  For each mover~$i$, the local bandwidth $h_x(X_i)$ is set
equal to the distance from $X_i$ to its $k$-th nearest neighbor in the mover subsample.
Cross-validation selects $(k,S)$ jointly over $k\in\{5,10,15,20,30\}$ and $S\in\{7,9,\dots,19\}$.  The weighting function is the Student-$t$
density with $3$ degrees of freedom, which places more mass on the higher frequencies
needed to resolve bimodality.

In all four specifications, the feasibility thresholds are $\Gamma_{\max}=100$ and $\rho_{\max}=0.5$ (see Remark~\ref{rem:CVfeasibility}).  After
estimation, the sieve density is truncated below at zero and renormalized to integrate
to one.

\paragraph{Results.}
Figure~\ref{fig:sim_all} displays the estimation results
for all four specifications.  In each panel, the black solid line is
the true density $f_\beta$, the dashed line is the pointwise average of
$\hat f_\beta$ across replications, the solid gray line is the pointwise median, and
the shaded region is the pointwise interquartile range.

In the symmetric case (panel~(a)), the estimator performs well.  The
pointwise median tracks the true density closely, and the interquartile range is narrow
throughout the support.  The cross-validation procedure selects $S=3$ in the large
majority of replications ($79$ out of $100$), with occasional selections of $S=5$ ($13$ replications) and $S=7$ ($7$ replications).  The selected bandwidth
$h_x$ has a median of approximately $0.71$, with an interquartile range of $[0.36,\,0.95]$.

In the skewed case (panel~(b)), the estimator captures the asymmetry
and the location of the mode accurately.  The pointwise median follows the true density
closely, while the pointwise average shows a slight negative bias near the peak and some
positive mass in the far left tail where the true density is near zero.  The interquartile
range is wider than in the symmetric case, reflecting the greater difficulty of the
estimation problem.  The cross-validation selects a median sieve dimension of $S=9$
(interquartile range $[7,11]$) across replications, with the full distribution spread over $S\in\{5,\dots,19\}$.  The selected bandwidth $h_x$ has
a median of $0.93$ with an interquartile range of $[0.48,\,0.96]$.  This variability
reflects the sensitivity of the numerator characteristic-function estimator to the local
density of movers, as discussed in Section~\ref{sec:heuristic_rates_irregular}.

\begin{figure}[h!]
    \centering
    \begin{subfigure}{0.45\textwidth}
        \centering
        \includegraphics[width=\linewidth]{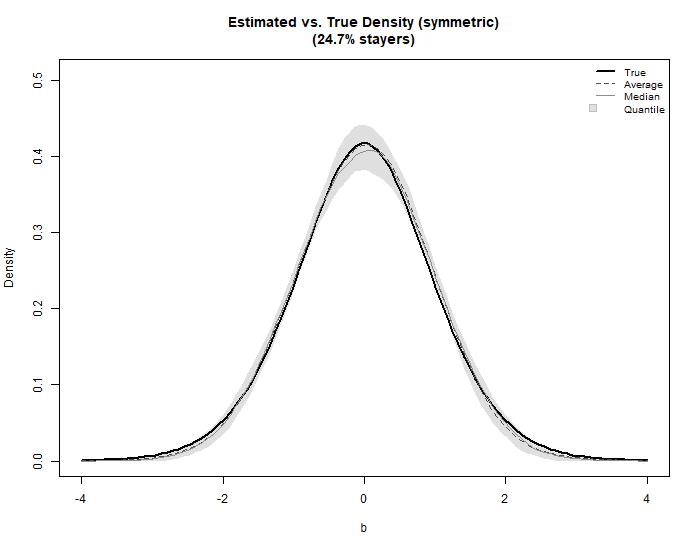}
        \caption{Symmetric, Gaussian errors.}
        \label{fig:sim_symmetric}
    \end{subfigure}
    \hfill
    \begin{subfigure}{0.45\textwidth}
        \centering
        \includegraphics[width=\linewidth]{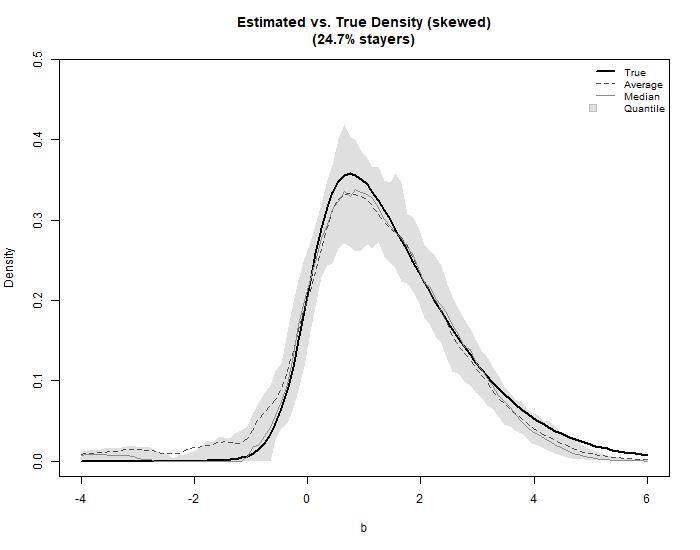}
        \caption{Skewed, Gaussian errors.}
        \label{fig:sim_skewed}
    \end{subfigure}

    \vspace{0.4cm}

    \begin{subfigure}{0.45\textwidth}
        \centering
        \includegraphics[width=\linewidth]{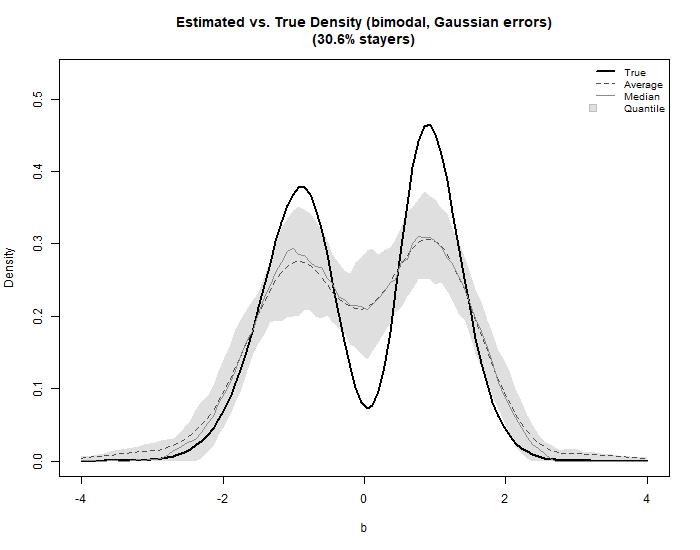}
        \caption{Bimodal, Gaussian errors.}
        \label{fig:sim_bimodal_gaussian}
    \end{subfigure}
    \hfill
    \begin{subfigure}{0.45\textwidth}
        \centering
        \includegraphics[width=\linewidth]{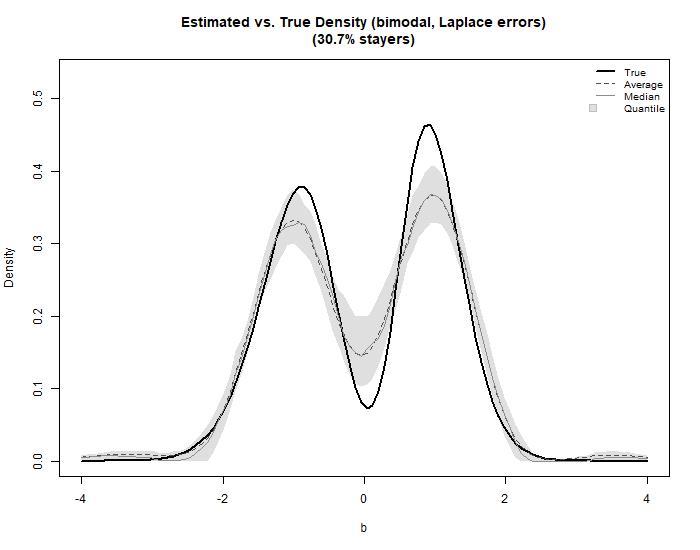}
        \caption{Bimodal, Laplace errors.}
        \label{fig:sim_bimodal_laplace}
    \end{subfigure}

    \caption{Estimation results for $f_\beta$ across the four specifications.\\
    \footnotesize Note: In each panel, the black curve represents the true density function $f_\beta$, the gray shaded area represents the interquartile range, the dashed line represents the pointwise average estimate $\hat{f}_\beta$, and the solid light gray line represents the pointwise median estimate, based on $100$ replications. Panels (a) and (b): Gaussian errors, global bandwidth, approximately $24.7\%$ stayers. Panels (c) and (d): $k$-NN adaptive bandwidth with CV-selected tuning parameters; panel~(c) Gaussian errors, approximately $30.6\%$ stayers; panel~(d) Laplace errors, approximately $30.7\%$ stayers; $\mathrm{Var}(D_i)$ matched across all specifications.}
    \label{fig:sim_all}
\end{figure}

The bimodal case is considerably more demanding.  Panel~(c) of Figure~\ref{fig:sim_all} displays the results under specification~(c) (Gaussian errors), where the $k$-NN adaptive bandwidth and the sieve dimension are selected jointly by cross-validation.  Across $100$ replications, the procedure selects $k=30$ in $88$ replications and $S=7$ in $93$ replications.  While the estimator identifies the general location and spread of the density, it does not fully resolve the trough between the two modes: the pointwise median and average estimates are flatter than the true density in the bimodal region.  The $k$-NN adaptive bandwidth and the Student-$t$ weighting help, but the attenuation of the bimodal feature persists.

The difficulty can be traced to the supersmooth character of the Gaussian error distribution.  Bimodality in~$f_\beta$ manifests as oscillations in the characteristic function $\varphi_\beta$ at moderate to high frequencies.  Recovering these oscillations requires accurate estimation of the ratio $\hat\varphi_{\tTwo Y\mid X}(u\mid X_i)/\hat\varphi_D(u/X_i)$ at frequencies where the denominator $\varphi_D$ is small.  Under Gaussian errors, $\varphi_D$ decays exponentially: for the baseline DGP, $|\varphi_D(v)| = \exp(-\mathrm{Var}(D_i)\,v^2/2)$, which at the argument $v = u/X_i = 3$ equals approximately $8\times 10^{-5}$---close to the trimming threshold $\tau_{\mathrm{den}} = 10^{-4}$.  The deconvolution ratio is therefore dominated by noise at precisely the frequencies that encode the bimodal structure.

The need for the adaptive bandwidth can also be understood from this
perspective. With a global bandwidth, the kernel regression averages over
movers whose deconvolution problems have very different conditioning: movers
with large~$|X_i|$ produce well-conditioned ratios (since $|u/X_i|$ is small,
keeping $\hat\varphi_D$ bounded away from zero), while movers with small~$|X_i|$
produce ratios dominated by noise. A single bandwidth that stabilizes the
latter necessarily oversmooths the former, attenuating the high-frequency
content of $\hat m_N(u)$.

The $k$-NN rule addresses this trade-off indirectly. Because movers with
small~$|X_i|$ concentrate near the threshold~$\tau_x$, where the mover
density is highest, they receive tight bandwidths under the $k$-NN rule;
movers with large~$|X_i|$, located in the sparser tails, receive wider
bandwidths. The tight bandwidths in the dense region preserve the
oscillatory structure of the conditional characteristic function that encodes
bimodality, which a global bandwidth would smooth away. At the same time,
the wider bandwidths in the tails provide additional smoothing for movers
that, while well-conditioned for deconvolution, are few in number and would
otherwise contribute noisy numerator estimates. Thus, the $k$-NN rule
primarily improves resolution where data are plentiful, rather than
stabilizing the noisiest deconvolution ratios directly. Stabilization of the
small-$|X_i|$ ratios is instead provided by the denominator trimming
threshold~$\tau_{\mathrm{den}}$ and the averaging over movers
in~\eqref{eq:mhat_scalar_irregular}.

Specification~(d) isolates the role of the smoothness class of $D_i$.  Since the Laplace characteristic function decays polynomially, $\varphi_D(t) = O(|t|^{-4})$ remains well above the trimming threshold at the moderate-to-high frequencies where the bimodal information resides: at the frequency argument $v = u/X_i = 3$, we have $|\varphi_D(v)| \approx 0.07$ under Laplace errors, compared to $8\times 10^{-5}$ under Gaussian errors---a ratio of approximately $865$.  All tuning parameters are as in specification~(c), and cross-validation again selects the tuning parameters jointly.  Across $100$ replications, the procedure selects $S=7$ in $91$ replications (with occasional values of $9$ or $11$) and $k=30$ in $84$ replications, with approximately $30.7\%$ of observations classified as stayers.

Panel~(d) of Figure~\ref{fig:sim_all} displays the results under specification~(d).  The bimodality is recovered substantially better: the pointwise median tracks both modes and the trough between them, and the interquartile range is markedly tighter than under specification~(c).  The comparison between panels~(c) and~(d) is particularly clean because the cross-validation procedure selects essentially the same tuning parameters in both specifications---$k=30$ and $S=7$ in the large majority of replications---yet produces strikingly different outcomes.  Since $f_\beta$ is identical in both
specifications, this confirms that the attenuation observed under specification~(c) is a consequence of the exponential decay of $\varphi_D$---a well-known limitation of deconvolution with supersmooth noise \citep[see, e.g.,][]{Fan1991}---rather than a deficiency of the sieve estimator or the adaptive bandwidth procedure.  Under Gaussian errors, the deconvolution ratio $\hat\varphi_{\tTwo Y\mid X}/\hat\varphi_D$ is dominated by noise at moderate-to-high frequencies, so the CV-selected $S=7$ reflects the highest sieve dimension that avoids fitting this noise.  Under Laplace errors, the polynomial decay of $\varphi_D$ preserves
the high-frequency content of $\widehat m_N(u)$ that encodes the bimodal
structure, so the same sieve dimension $S=7$ now captures genuine features of~$f_\beta$.


\section{Application}
\label{sec:application}

We revisit the Nicaraguan calorie-demand application in \citet{GrahamPowell2012}.
The data consist of a balanced panel of $N=1{,}358$ poor rural households observed
in 2000, 2001, and 2002 in communities covered by the conditional cash transfer
program \emph{Red de Protecci\'on Social} (RPS). Data construction is described in
\citet{GrahamPowell2012} and the references therein.

\subsection{Model and object of interest}

Let \(r \in \{0,1\}\) denote the RPS regime, where \(r=1\)
corresponds to assignment to receipt of RPS transfers and \(r=0\) corresponds to the
no-RPS regime. For each household \(i\), define the potential structural
calorie equation
\begin{equation}\label{eq:app_model}
\log \text{Cal}_{is}(r,x)
=
\alpha_i(r)+\beta_i(r)x+u_{is}(r),
\qquad
x \in \mathcal X,
\qquad
s=2000,\;2001,\;2002,
\end{equation}
where \(x\) denotes a counterfactual value of log real per-capita total
household expenditure. The coefficient \(\beta_i(r)\) is the
household-specific elasticity of calorie intake with respect
to expenditure under regime \(r\). It is a causal derivative of the
potential outcome schedule with respect to expenditure, holding the RPS
regime fixed. 

Within each RPS stratum, we suppress the regime argument, so that \(\beta_i\) should be read as \(\beta_i(r)\) for households
observed under regime \(r\). Consequently, when the estimator is applied
separately to RPS recipients and non-recipients, the corresponding
densities estimate \(f_{\beta(1)}\) and \(f_{\beta(0)}\), respectively.

The raw panel has three time periods, \(s=2000,2001,2002\). After
fixed-effect elimination, for each adjacent pair of periods and within a
fixed RPS regime \(r\),
\begin{equation}\label{eq:app_fd}
Y_{it}(r)
=
X_{it}(r)\beta_i(r)+U_{it}(r),
\qquad
t=1,2,
\end{equation}
where \(t=1\) corresponds to 2000--2001 and \(t=2\) to 2001--2002, with
\begin{align*}
Y_{it}(r) &=
\log \text{Cal}_{is}(r)-\log \text{Cal}_{i,s-1}(r),
\\
X_{it}(r)
&=
\log \text{Exp}_{is}(r)-\log \text{Exp}_{i,s-1}(r),
\\
U_{it}(r)
&=
u_{is}(r)-u_{i,s-1}(r).
\end{align*}
Within each RPS stratum, we again suppress the regime argument and write
\[
Y_{it}=X_{it}\beta_i+U_{it}.
\]
The regressor is the change in log expenditure between adjacent periods,
so that
\[
X_{it}
\approx
\frac{\text{Exp}_{is}-\text{Exp}_{i,s-1}}
     {\text{Exp}_{i,s-1}}
\]
represents the approximate percentage change in household expenditure.
Each year pair, taken separately, is a scalar irregular design with
\((T,p,q)=(1,1,0)\). Stacking the two differenced periods,
\(Y_i=(Y_{i1},Y_{i2})'\), yields the regular design with
\((T,p,q)=(2,1,0)\).

Our objects of interest are the regime-specific cross-sectional density
functions \(f_{\beta(r)}\), \(r\in\{0,1\}\). These densities summarize
heterogeneity in household-specific structural calorie-expenditure
elasticities under each RPS regime. In the full sample, the corresponding
object is the distribution of the observed-regime elasticity
\[
\beta_i^{\mathrm{obs}}
=
\beta_i(R_i),
\]
where \(R_i\) denotes the observed RPS status of household \(i\). Thus
the full-sample density is an observed-regime mixture, whereas the
subsample densities estimate \(f_{\beta(1)}\) and \(f_{\beta(0)}\)
separately.

We estimate the density of the observed-regime elasticity in the full
sample and estimate the regime-specific densities separately for RPS
recipients, \(\mathrm{RPS}=1\), and non-recipients,
\(\mathrm{RPS}=0\). For the full sample, we report both the scalar
irregular estimator, applied separately to the two adjacent year pairs,
and the regular estimator, which pools the two differenced periods. For
the subsample analysis, we apply the scalar irregular estimator
separately within each RPS stratum.

Because RPS assignment was randomized at the community level and take-up
was high, comparisons of the estimated densities across RPS strata can be
interpreted as causal contrasts between the regime-specific distributions, subject to the maintained CRC
assumptions. Formally, the experimental design identifies an
assignment-induced contrast in structural elasticity distributions; high
take-up supports interpreting this contrast as receipt-regime contrast.

Allowing \(\beta_i(r)\) to be correlated with the path of expenditure
growth \(X_i(r)=(X_{i1}(r),X_{i2}(r))'\) is natural in this application.
Transfer amounts varied with household composition, so households
receiving larger transfers experienced larger expenditure changes
\citep{MaluccioFlores2005}. The same households plausibly faced greater
unmet caloric needs and therefore had stronger incentives to allocate
additional resources toward calorie-dense staple consumption. In
addition, households closer to caloric subsistence likely had both larger
expenditure shocks in proportional terms and larger marginal calorie
responses to additional resources. These considerations motivate a
correlated random coefficient specification rather than an exogenous
random coefficient model.

Within each RPS regime, Assumption 1(ii) requires that the transitory
component \(U_{it}(r)\) is independent of
\((\beta_i(r),X_i(r))\), where
\(X_i(r)=(X_{i1}(r),X_{i2}(r))'\) denotes the vector of expenditure
changes under regime \(r\). In the realized stratum-specific equation,
where the regime argument is suppressed, this restriction becomes the
requirement that \(U_{it}\) is independent of \((\beta_i,X_i)\).

This restriction should be interpreted as a decomposition of the outcome
into systematic and idiosyncratic components. The model allows for
arbitrary dependence between the structural elasticity \(\beta_i(r)\) and
the expenditure path \(X_i(r)\), so that systematic co-movement between
expenditure dynamics and calorie demand may operate through heterogeneous
behavioral responses. The disturbance \(U_{it}(r)\) is then the residual
component that remains after accounting for this heterogeneous structure,
and is interpreted as an idiosyncratic noise term.

Under this interpretation, Assumption 1(ii) rules out residual dependence
between the unexplained component and the regressors within each RPS
regime, but does not preclude economically meaningful dependence between
expenditure and calorie demand more generally. Instead, it requires that
such dependence be captured by the distribution of \(\beta_i(r)\), rather
than by the disturbance \(U_{it}(r)\).

\subsection{Data}

Figure~\ref{fig:histograms} reports histograms of changes in log expenditure and
log calorie intake over the two differenced periods. A central feature of the data
for the scalar irregular design is that the distribution of
$X_i=\Delta\log(\text{Exp}_i)$ is centered near zero, with substantial mass in a
neighborhood of the origin. This is important because first-step identification in
the scalar irregular design relies on stayers, that is, households with
$|X_i|$ close to zero. Beyond this concentration near zero, the expenditure-change
distributions are heavy-tailed. The distributions of calorie changes are
unimodal and fat-tailed.

\begin{figure}[h!]
    \centering
    \includegraphics[width=0.9\linewidth]{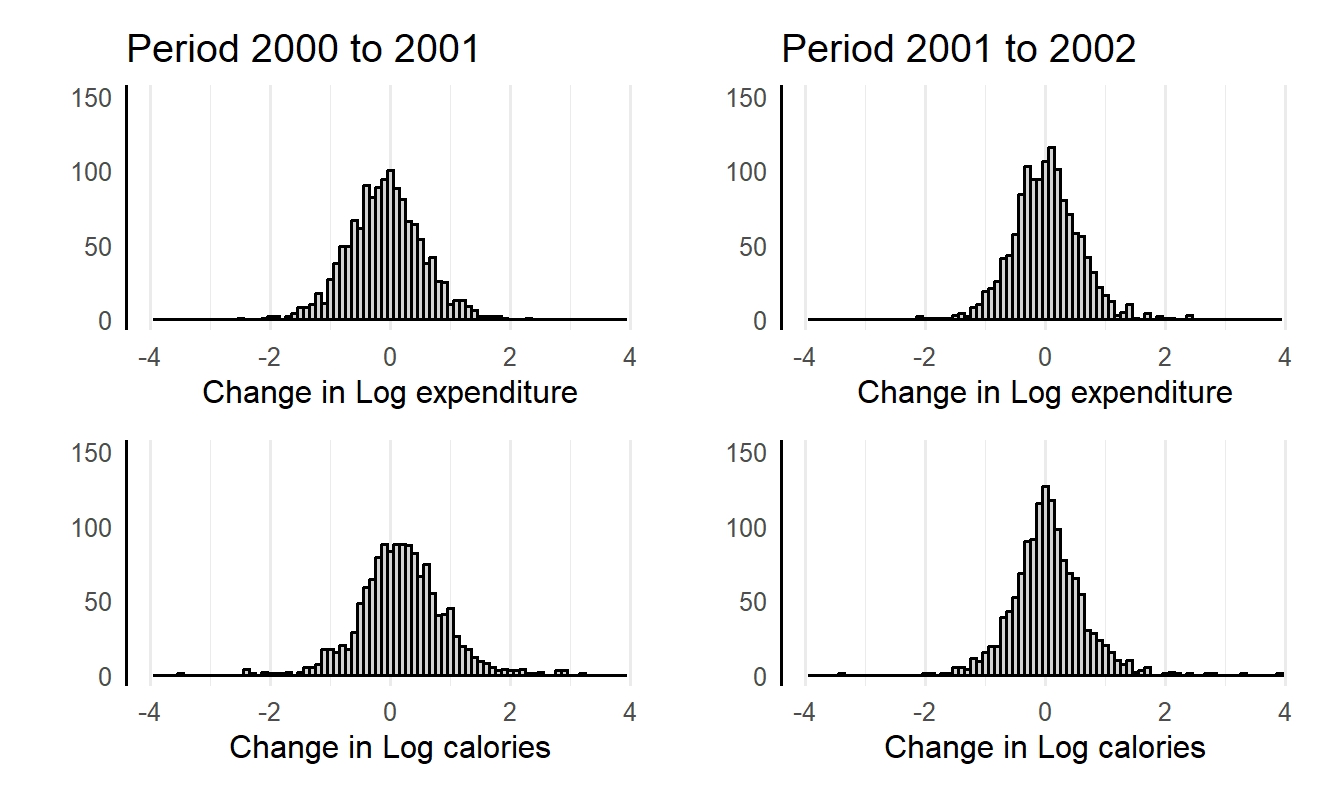}
    \caption{Histograms of changes in log expenditure (top panels) and log calorie
    intake (bottom panels), by differenced period.}
    \label{fig:histograms}
\end{figure}

\subsection{Estimation}

We implement both the scalar irregular estimator, which treats the two 
adjacent year pairs separately as $(T, p, q) = (1, 1, 0)$ designs, and 
the regular estimator, which pools them as a $(T, p, q) = (2, 1, 0)$ 
design.

\paragraph{Scalar irregular estimator.}
The scalar irregular estimator exploits the presence of stayers, or households whose expenditure changed by less than a threshold $\tau_x$ in absolute value, to identify the 
characteristic function of the transitory component $D_i$ along a ray. 
This characteristic function is then used to deconvolve the distribution 
of $\beta_i$ from the conditional characteristic function of the 
transformed outcome for movers ($|X_i| \geq \tau_x$).

In this application, the stayer condition $|X_{it}| < \tau_x$ has a 
direct economic interpretation: households are classified as stayers 
if their expenditure changed by less than approximately $100 \times \tau_x$ 
percent between adjacent survey waves. The threshold $\tau_x$ is set 
according to a data-driven rule that scales with sample size and the 
dispersion of $X_i$; the resulting stayer proportions range from 7 to 10 
percent across specifications, as reported in 
Table~\ref{tab:tuning_parameters}.

\paragraph{Regular estimator.}
When both differenced periods are stacked, the regressor matrix 
$X_i = (X_{i1}, X_{i2})' \in \mathbb{R}^2$ has $T = 2 > p = 1$, placing the 
problem in the regular design. The first-step transformation annihilates 
$\beta_i$ for every observation, yielding a scalar equation from which 
the characteristic function $\varphi_D$ is estimated by smoothing over 
normalized directions on $\mathbb{S}^1$. The second step recovers $f_\beta$ 
by deconvolution via the sieve minimum distance procedure.

\paragraph{Tuning parameter selection.}
In all cases, tuning parameters are selected by the cross-validation 
procedure of Algorithm~\ref{alg:pcv_tuning}. The selected values are 
summarized in Table~\ref{tab:tuning_parameters}. Details on the implementation, 
including the formula for $\tau_x$ and the cross-validation grid, can 
be found in Section~\ref{sec:applicationestimation}.

\paragraph{Inference.}
Sampling uncertainty is assessed using a nonparametric pairs bootstrap 
with $B = 499$ replications. In the scalar irregular design, each 
bootstrap draw resamples the differenced observations corresponding to 
the relevant year pair and recomputes the estimator. The stayer 
threshold $\tau_x$ and the first-stage bandwidth $h_0$ are held fixed 
at their original-sample values, while the second-stage tuning parameters 
$(h_x, S)$ are also fixed at their original cross-validation choices. 
In the regular design, each bootstrap draw resamples the household-level 
observations $(Y_{i1}, Y_{i2}, X_{i1}, X_{i2})$ and recomputes the 
estimator with the first-stage tuning parameter $h_S$ fixed at its 
original value and $(h_X, S)$ fixed at their original cross-validation 
selections.

Confidence intervals for the density are constructed using the basic 
(reverse-percentile) bootstrap.\footnote{We emphasize that this procedure is a practical approximation. 
Theoretical validity of the bootstrap for deconvolution density 
estimators has not been established in the literature, and extends 
\textit{a fortiori} to the two-stage sieve minimum distance estimator 
with post-processing employed here. The reported confidence bands 
should therefore be interpreted as indicative measures of sampling 
variability rather than as intervals with guaranteed asymptotic 
coverage.} For each evaluation point $b$, let 
$\hat{f}_\beta(b)$ denote the estimator of $f_\beta(b)$ and let 
$\{\hat{f}_\beta^{*(r)}(b)\}_{r=1}^B$ denote the bootstrap draws. The 
$(1 - \alpha)$ confidence interval for $f_\beta(b)$ is given by
\[
\left[\, \hat{f}_\beta(b) - q_{1-\alpha/2}(b), \; 
\hat{f}_\beta(b) - q_{\alpha/2}(b) \,\right],
\]
where $q_p(b)$ denotes the $p$-th quantile of the centered bootstrap 
distribution $\hat{f}_\beta^*(b) - \hat{f}_\beta(b)$. The reported 
confidence bands are obtained by applying this procedure pointwise 
over $b$ and therefore do not provide uniform coverage. Inference is 
conducted conditional on the selected tuning parameters and does not 
account for additional uncertainty arising from their estimation.

\subsection{Results}

\paragraph{Moments.}
Table~\ref{tab:moments_combined} reports estimates of the mean and variance
of the cross-sectional distribution of $\beta_i$, together with bootstrap
standard errors, for the full sample and by RPS status.

Across specifications, the mean elasticity is positive and statistically
significant, consistent with \citet{GrahamPowell2012}. In the full sample,
the regular estimator yields $\widehat{\mathbb{E}}[\beta_i] = 0.677$
(s.e.\ $0.035$), which lies between the irregular estimates for
2000--2001 ($0.684$) and 2001--2002 ($0.467$) but is much closer to the
former. The regular estimator is roughly twice as precise as either
irregular estimate, reflecting both pooling across periods and the
imposition of time-invariance: rather than identifying $f_\beta$
separately from each pair of differenced periods using only stayer
households for the denominator characteristic function, the regular
estimator uses all $N$ households and the full collinearity
structure of $X_i'X_i$, eliminating the leading source of variance in
the irregular design.

All specifications exhibit substantial dispersion. In the full sample,
the regular estimator gives $\widehat{\mathrm{Var}}[\beta_i] = 0.591$
(s.e.\ $0.017$), while the irregular estimates yield $0.765$ and $0.535$,
albeit with markedly larger sampling variability (s.e.\ $0.188$ and $0.224$,
respectively). The implied standard deviations range from $0.73$ to $0.88$
across designs, indicating considerable heterogeneity in marginal calorie
responses across households, including within each RPS stratum.

The irregular estimates suggest an intertemporal shift concentrated among
recipients. In the full sample, the mean declines from $0.684$ to $0.467$,
but this difference is not statistically significant at conventional levels.
Among non-recipients, the mean declines modestly from $0.726$
(s.e.\ $0.161$) to $0.606$ (s.e.\ $0.082$), a change that is also not
statistically significant. In contrast, for RPS recipients, the mean
falls sharply from $0.787$ (s.e.\ $0.086$) to $0.518$ (s.e.\ $0.099$),
a difference that is statistically significant.

Across strata, the regular estimator places the recipient mean above the
non-recipient mean ($\widehat{\mathbb{E}}[\beta_i \mid \mathrm{RPS}=1] =
0.698$ vs.\ $0.625$ for $\mathrm{RPS}=0$), with a difference of $0.073$
that is not statistically significant
(s.e.\ of the difference $\approx 0.065$ under the standard
independent-subsamples approximation, since the two pairs bootstraps draw
from disjoint household sets). The regular variance estimates are
$\widehat{\mathrm{Var}}[\beta_i \mid \mathrm{RPS}=1] = 0.635$
(s.e.\ $0.027$) and $\widehat{\mathrm{Var}}[\beta_i \mid \mathrm{RPS}=0]
= 0.586$ (s.e.\ $0.028$); the cross-stratum difference of $0.049$ is
likewise within one standard error of zero. Thus, while the point
estimates suggest that recipients have a slightly higher mean and a
slightly more dispersed distribution of marginal calorie responses than
non-recipients, the regular estimator does not detect a statistically
distinguishable shift in either moment.

\paragraph{Densities.}
Figures~\ref{fig:irregular_densities} and \ref{fig:regular_densities}
report the corresponding density estimates.

\textit{Irregular estimates.} In the full sample (panel c), both
period-specific densities are unimodal with modes near $0.5$, and exhibit
substantial overlap, consistent with the failure to reject time-invariance
at conventional significance levels. The 2001--2002 density is modestly
shifted leftward relative to the 2000--2001 density. This shift reflects
a change in location of the distribution, with the mode and central mass
shifting leftward, rather than changes confined to the tails.

Among recipients (panel b), the shift is more pronounced: the 2000--2001
density has a mode near $0.6$, while the 2001--2002 density shifts leftward
with a mode near $0.4$--$0.5$. The bootstrap bands for the two periods show
less overlap than in the full sample, consistent with the statistically
significant decline in the mean. Among non-recipients (panel a), the
densities are similar across periods, with substantially overlapping
bootstrap bands, consistent with time-invariant elasticities in this group.

\textit{Regular estimates.} The regular estimator yields a unimodal density
centered near $0.5$ with substantially narrower bootstrap
bands than the irregular estimates, reflecting the efficiency gains from
pooling and from imposing time-invariance. Comparing strata
(Figure~\ref{fig:regular_densities}), the two density estimates are
remarkably close: both peak in the interval $[0.4, 0.6]$ at
heights near $0.50$, and the $90\%$ pointwise bootstrap bands overlap
across essentially the entire support. The recipient density has slightly
more mass in the right tail and slightly less
mass at the mode, consistent with its larger mean and variance reported
in Table~\ref{tab:moments_combined}, but neither feature is large relative
to sampling uncertainty. 

The regular design provides little evidence of a large
regime-level shift in the pooled elasticity distribution. Its main implication
is instead that calorie--expenditure elasticities are heterogeneous. In light
of the period-specific irregular estimates, however, the regular estimator
should be interpreted as a common-\(\beta_i\) benchmark: if elasticities change
with cumulative RPS exposure, the regular estimator necessarily averages over,
and may mask, those exposure-specific changes.

\paragraph{Reconciling the literature.} The estimated density of the calorie--expenditure elasticity $\beta_i$ reveals substantial heterogeneity in household responses to changes in expenditure. While the average elasticity is positive and of the same order as those commonly reported in the literature, the distribution exhibits significant mass near zero as well as non-negligible probability on negative values.

This pattern helps reconcile seemingly conflicting findings in the development literature. Early studies document small calorie--expenditure elasticities, often close to zero 
\citep{BehrmanDeolalikar1987,BouisHaddad1992,Ravallion1990}, whereas \citet{SubramanianDeaton1996} report elasticities in the range of 0.3--0.5 and argue against the view that calorie responses are negligible. Our results suggest that these differences may reflect aggregation: the mean elasticity masks substantial dispersion in individual responses.

A natural interpretation of the heterogeneity in $\beta_i$ is that households adjust along both quantity and quality margins. For some 
households, particularly those facing caloric constraints, higher 
expenditure translates into increased calorie intake, corresponding to 
positive values of $\beta_i$. For others, however, expenditure increases 
are primarily allocated toward higher-quality, more diverse foods rather 
than additional calories. This compositional shift, which is well 
documented in the literature \citep{Deaton1997,JensenMiller2008,Skoufias2011}, 
weakens the relationship between income and calorie intake and can generate 
elasticities close to zero or even negative.

The presence of negative support in the estimated density indicates that, for a subset of households, 
increases in expenditure are associated with reductions in calorie intake, 
consistent with substitution away from inexpensive, calorie-dense staples 
toward foods with lower caloric density per unit of expenditure. Such 
behavior is also consistent with broader evidence on dietary transitions 
and changing consumption patterns as living standards improve 
\citep{DeatonDreze2009}.

Overall, the distributional evidence underscores the limitations of 
focusing exclusively on average elasticities. By recovering the full 
distribution of $\beta_i$, the analysis reveals that the relationship 
between expenditure and nutrition is fundamentally heterogeneous, possibly
reflecting differences in constraints, preferences, and consumption 
margins across households.

\begin{figure}[ht]
\centering
\begin{subfigure}{0.45\textwidth}
  \centering
  \includegraphics[width=\linewidth]{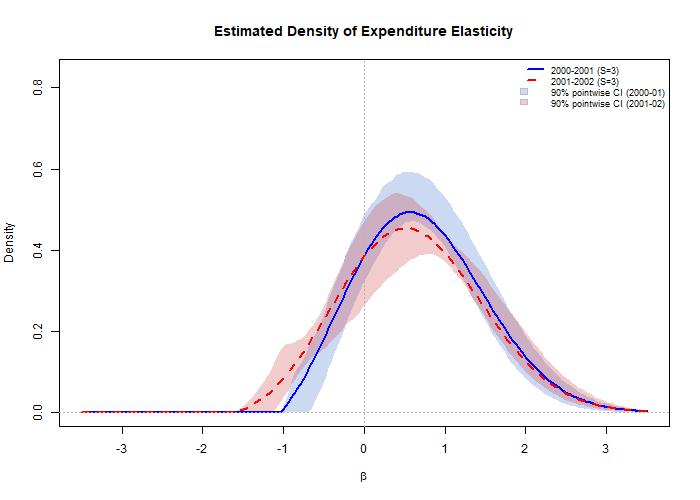}
  \caption{RPS $=0$}
\end{subfigure}
\hfill
\begin{subfigure}{0.45\textwidth}
  \centering
  \includegraphics[width=\linewidth]{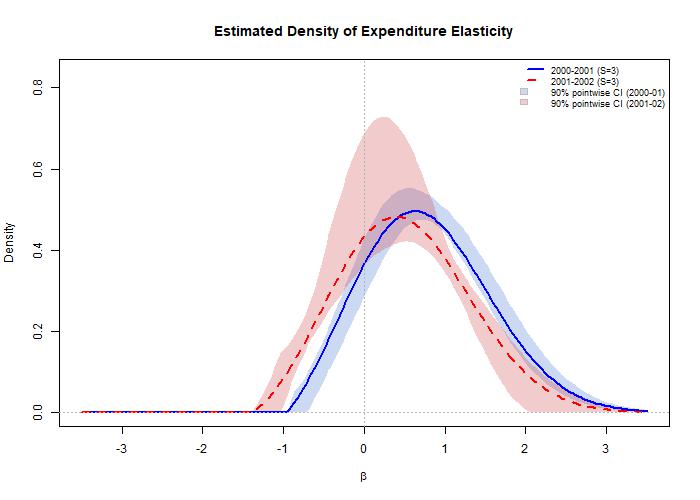}
  \caption{RPS $=1$}
\end{subfigure}

\vspace{0.5em}
\begin{subfigure}{0.5\textwidth}
  \centering
  \includegraphics[width=\linewidth]{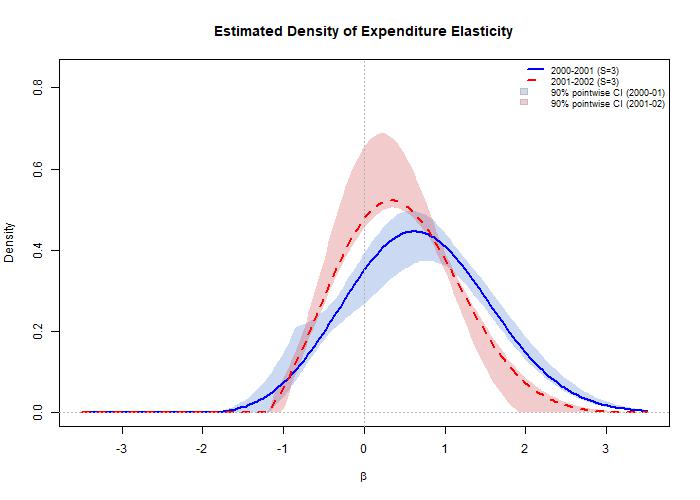}
  \caption{Full sample}
\end{subfigure}

  \caption{Scalar irregular density estimates of the calorie--expenditure
  elasticity $\beta_i$. Panels~(a) and~(b) use the RPS$\,{=}\,0$ ($N=653$) and RPS$\,{=}\,1$ 
  ($N=705$) subsamples, respectively; panel~(c) uses the full sample ($N=1{,}358$). Within each panel, solid blue and 
  dashed red lines correspond to 2000--2001 and 2001--2002; shaded regions are pointwise 90\% 
  centered (basic) bootstrap confidence bands based on $B=499$ pairs bootstrap replications. 
  Tuning parameters $(\tau_x, h_0, h_x, S)$ are held fixed at their original-sample values 
  across bootstrap draws. Densities are post-processed by truncating negative values at zero 
  and renormalizing to integrate to one.}
  \label{fig:irregular_densities}
\end{figure}

\begin{figure}[h!]
  \centering
  \begin{subfigure}{0.48\textwidth}
    \centering
    \includegraphics[width=\linewidth]{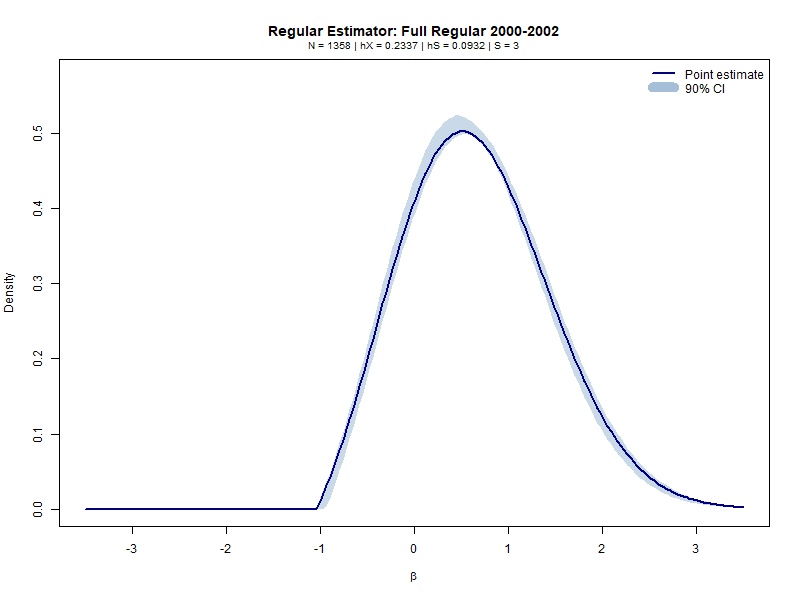}
    \caption{Full sample}
  \end{subfigure}
  \hfill
  \begin{subfigure}{0.48\textwidth}
    \centering
    \includegraphics[width=\linewidth]{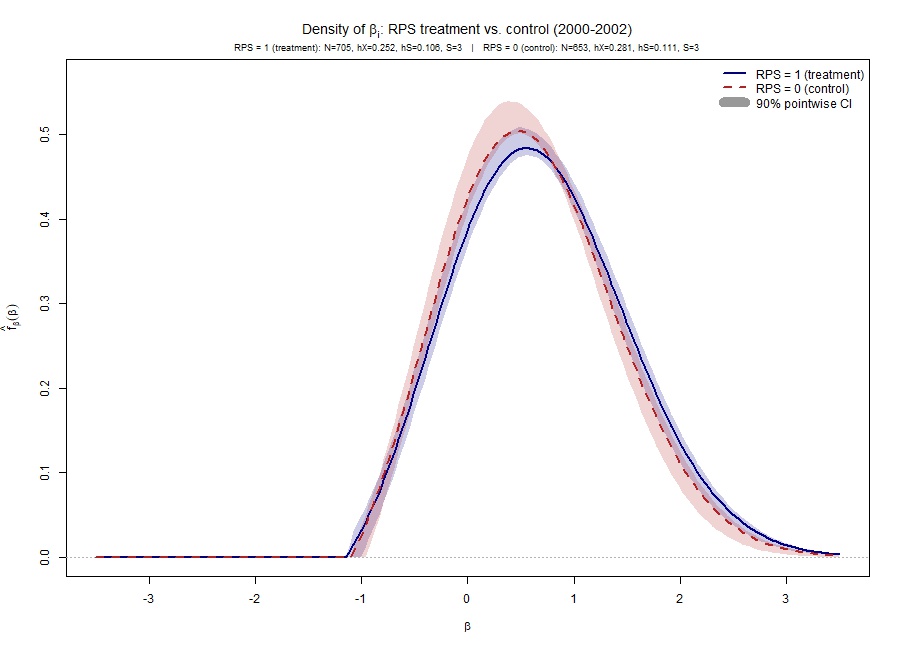}
    \caption{RPS $=1$ vs RPS $=0$}
  \end{subfigure}

  \vspace{0.5em}

  \caption{Regular density estimates of the calorie--expenditure elasticity
  $\beta_i$, pooling 2000--2001--2002. Panel~(a) shows the full sample 
  ($N=1{,}358$). Panel~(b) compares the RPS$\,{=}\,1$ subsample ($N=705$, 
  solid blue) with the RPS$\,{=}\,0$ subsample ($N=653$, dashed red). 
  Dotted lines are bootstrap medians; shaded regions are pointwise 
  $[5\%,95\%]$ bootstrap bands based on $B=499$ replications.}
  \label{fig:regular_densities}
\end{figure}

\begin{table}[ht]
\centering
\caption{Moments of the cross-sectional distribution of $\beta_i$}
\label{tab:moments_combined}
\renewcommand{\arraystretch}{1.15}
\begin{tabular}{lcccc}
\toprule
& Stayers (\%) & $\widehat{\mathbb{E}}[\beta_i]$ & $\widehat{\mathrm{Var}}[\beta_i]$ & $\widehat{\mathrm{SD}}[\beta_i]$ \\
\midrule
\multicolumn{5}{l}{\textit{Full sample}} \\[2pt]
Regular, 2000--2002
  & ---  & 0.677\,(0.035) & 0.591\,(0.017) & 0.769\,(0.011) \\
Irregular, 2000--2001
  & 8.1 & 0.684\,(0.081) & 0.765\,(0.188) & 0.875\,(0.107) \\
Irregular, 2001--2002
  & 7.1 & 0.467\,(0.074) & 0.535\,(0.224) & 0.732\,(0.153) \\
\addlinespace
\multicolumn{5}{l}{\textit{RPS $=1$}} \\[2pt]
Regular, 2000--2002
  & ---  & 0.698\,(0.042) & 0.635\,(0.027) & 0.797\,(0.017) \\
Irregular, 2000--2001
  & 9.8 & 0.787\,(0.086) & 0.604\,(0.128) & 0.777\,(0.082) \\
Irregular, 2001--2002
  & 9.1 & 0.518\,(0.099) & 0.642\,(0.474) & 0.801\,(0.296) \\
\addlinespace
\multicolumn{5}{l}{\textit{RPS $=0$}} \\[2pt]
Regular, 2000--2002
  & ---  & 0.625\,(0.050) & 0.586\,(0.028) & 0.766\,(0.018) \\
Irregular, 2000--2001
  & 9.6 & 0.726\,(0.161) & 0.610\,(0.177) & 0.781\,(0.113) \\
Irregular, 2001--2002
  & 9.5 & 0.606\,(0.082) & 0.729\,(0.236) & 0.854\,(0.138) \\
\bottomrule
\end{tabular}
\vspace{0.5em}

\footnotesize
\emph{Notes:} Bootstrap standard errors in parentheses, based on $B = 499$ nonparametric pairs bootstrap replications. Each bootstrap draw recomputes the estimator while holding tuning parameters fixed at their original-sample values, including cross-validation selections. Standard errors are computed as the sample standard deviation of the bootstrap estimates; standard errors for $\widehat{\mathrm{SD}}[\beta_i]$ use the delta method. ``Stayers'' denotes the percentage of households classified as stayers in the irregular design. The regular estimator pools all periods and does not rely on a stayer partition. Inference is conditional on the selected tuning parameters and does not account for additional uncertainty arising from their estimation.
\end{table}

\FloatBarrier

\section{Conclusion}

This paper develops identification and estimation results for the distribution of correlated random coefficients in short panel data models, without imposing restrictions on serial dependence in the error terms. We show that the structure of the panel model leads to two distinct identification strategies, depending on whether the design is regular or irregular, and we provide corresponding estimators based on characteristic-function methods. In irregular designs, identification relies on limit arguments near singular regressor realizations, while in regular designs it follows from orthogonal projections. In both cases, estimation proceeds via a regularized deconvolution step combined with a Hermite sieve minimum-distance approximation.

Our analysis highlights that the statistical difficulty of the problem is driven by the first-stage deconvolution, which involves estimating a characteristic function at random, regressor-dependent frequencies and controlling instability arising from near-zero denominators. The resulting estimator exhibits a nonstandard bias--variance tradeoff, governed by the interaction between the stayer approximation, Fourier-domain regularization, and sieve approximation.

Our empirical application to the elasticity of calorie expenditure contributes to an ongoing debate in the development economics literature on whether households are calorie constrained or instead engage in quality substitution as income rises. By recovering the full distribution of heterogeneous elasticities, rather than focusing on average effects alone, our framework allows for a richer characterization of household behavior. In particular, it permits distinguishing between populations for which calorie intake responds strongly to expenditure changes and those for which additional resources are allocated toward higher-quality, non-caloric food consumption. This distributional perspective complements existing approaches and provides new evidence on the extent and nature of heterogeneity underlying the aggregate elasticity estimates commonly reported in the literature.

An important direction for future work is the development of inference procedures that accommodate both generated moment conditions and random frequency dispersion. In our setting, moment conditions naturally depend on first-step estimates of the disturbance characteristic function, introducing additional sampling variation and bias. At the same time, the relevant characteristic functions are evaluated at effective frequencies of the form $u/X_i$, so that information is dispersed according to the distribution of $1/X_i$. This dispersion governs the stability and informativeness of the empirical moments and complicates inference relative to standard deconvolution settings. Designing inference procedures that remain robust to these features is a promising avenue for further research.

\appendix

\section{Proofs}

\subsection{Proof of Theorem \ref{main_result}}
\label{appendix_proof_thm1}

The proof proceeds in two steps. 

\medskip
\noindent\emph{Step 1.} We first identify $\varphi_D$. We proceed with the irregular design, and continue with the regular design.

\medskip
\noindent\textit{Irregular case ($T=p$).}
Take $\tau_1(X_i)=\operatorname{adj}(X_i)$, so that
\[
  \tau_1(X_i)X_i = \det(X_i)I_p,
\]
by the standard adjugate identity.
Fix $v\in\mathbb{R}$ and any $\lambda$ in the conditional support of
$\widetilde{W}_i'\mu_i$ given $\det(X_i)=0$.
Using $\widetilde{Y}_i = \tau_1(X_i)Y_i = \tau_1(X_i)X_i\beta_i
+ \tau_1(X_i)W_iD_i$, we obtain
\[
  \mu_i'\widetilde{Y}_i
  =
  \mu_i'\tau_1(X_i)X_i\beta_i
  +
  \mu_i'\widetilde{W}_iD_i.
\]
Since $T=p$ and $\tau_1(X_i)X_i=\det(X_i)I_p$,
\[
  \mu_i'\tau_1(X_i)X_i\beta_i = \det(X_i)\mu_i'\beta_i,
\]
and therefore
\[
  \exp\!\bigl(\ii v\,\mu_i'\widetilde{Y}_i\bigr)
  =
  \exp\!\bigl(\ii v\det(X_i)\mu_i'\beta_i\bigr)\,
  \exp\!\bigl(\ii v\,\mu_i'\widetilde{W}_iD_i\bigr).
\]
Conditioning on $\{\widetilde{W}_i'\mu_i=\lambda,\,\det(X_i)=0\}$,
the first factor equals $1$, while
$\mu_i'\widetilde{W}_iD_i = (\widetilde{W}_i'\mu_i)'D_i = \lambda'D_i$.
Hence
\[
  \mathbb{E}\!\left[
    e^{\ii v\mu_i'\widetilde{Y}_i}
    \,\middle|\,
    \widetilde{W}_i'\mu_i=\lambda,\;
    \det(X_i)=0
  \right]
  =
  \mathbb{E}\!\left[e^{\ii v\lambda'D_i}\right]
  =
  \varphi_D(v\lambda),
\]
where the equality
$\mathbb{E}[e^{\ii v\lambda'D_i}\mid \widetilde{W}_i'\mu_i=\lambda,
\det(X_i)=0] = \mathbb{E}[e^{\ii v\lambda'D_i}]$
uses the independence of $D_i$ from $(X_i,W_i)$ in
Assumption~\ref{assm:A1}(ii), and Assumption~\ref{support}(i) justifies
the interpretation of the conditioning event $\{\det(X_i)=0\}$ by
continuity, as $\det(X_i)\to 0$.

This identifies $\varphi_D(v\lambda)$ for every $v\in\mathbb{R}$ and
every $\lambda$ in the conditional support of $\widetilde{W}_i'\mu_i$
given $\det(X_i)=0$.
By Assumption~\ref{assm:A1_double},
$\operatorname{supp}(S_i\mid\det(X_i)=0)=\mathbb{S}^{d-1}$, so
$\varphi_D$ is identified on a dense subset of $\mathbb{R}^d$.
Since $\varphi_D$ is continuous as a characteristic function, it is
therefore identified on all of $\mathbb{R}^d$.

\medskip
\noindent\textit{Regular case ($T>p$).}
Under Assumption~\ref{support}(ii), $\operatorname{rank}(X_i)=p$ almost surely,
so $\tau_1(X_i)X_i=0$ almost surely by construction of $\tau_1$.
Hence
\[
  \mu_i'\widetilde{Y}_i
  =
  \mu_i'\widetilde{W}_iD_i
  =
  (\widetilde{W}_i'\mu_i)'D_i.
\]
Conditioning on $\{\widetilde{W}_i'\mu_i=\lambda\}$ and using the independence
of $D_i$ from $(X_i,W_i)$ in Assumption~\ref{assm:A1}(ii) gives
\[
  \mathbb{E}\!\left[
    e^{\ii v\mu_i'\widetilde{Y}_i}
    \,\middle|\,
    \widetilde{W}_i'\mu_i=\lambda
  \right]
  =
  \mathbb{E}\!\left[e^{\ii v\lambda'D_i}\right]
  =
  \varphi_D(v\lambda).
\]
As in the irregular case, by Assumption~\ref{assm:A1_double}, $\operatorname{supp}(S_i)=\mathbb{S}^{d-1}$, so $\varphi_D$ is identified
on $\mathbb{R}^d$.

\medskip
\noindent\emph{Step 2.} We identify $f_\beta$. 

Since $\widetilde{\widetilde{Y}}_i = \beta_i + \widetilde{\widetilde{W}}_iD_i$, Assumption~\ref{assm:A1}(ii) implies two things.
First, $D_i\perp(X_i,W_i)$, so
$\varphi_{\widetilde{\widetilde{W}}D\mid X,W}(u\mid X_i,W_i)
= \varphi_D(\widetilde{\widetilde{W}}_i'u)$ for almost every $(X_i,W_i)$.
Second, $\beta_i\perp D_i\mid(X_i,W_i)$, so the characteristic function of
$\widetilde{\widetilde{Y}}_i$ given $(X_i,W_i)$ factorizes as
\[
  \varphi_{\widetilde{\widetilde{Y}}\mid X,W}(u\mid X_i,W_i)
  =
  \varphi_{\beta\mid X,W}(u\mid X_i,W_i)\,
  \varphi_{\widetilde{\widetilde{W}}D\mid X,W}(u\mid X_i,W_i),
\]
for almost every $(X_i,W_i)$.
By Assumption~\ref{assm:A1_prime} ($\varphi_D$ nowhere zero), the denominator
is nonzero for all $u\in\mathbb{R}^p$, so
\[
  \varphi_{\beta\mid X,W}(u\mid X_i,W_i)
  =
  \frac{%
    \varphi_{\widetilde{\widetilde{Y}}\mid X,W}(u\mid X_i,W_i)
  }{%
    \varphi_{\widetilde{\widetilde{W}}D\mid X,W}(u\mid X_i,W_i)
  }.
\]
Assumption~\ref{Fourier_inversion} ensures that this ratio is absolutely
integrable in $u\in\mathbb{R}^p$ for almost every $(X_i,W_i)$, so Fourier
inversion yields the conditional density $f_{\beta\mid X,W}(\cdot\mid X_i,W_i)$.
Since $f_\beta$ is the marginal density of $\beta_i$ under
Assumption~\ref{assm:A1}(i), we then have that
\[
  f_\beta(b)
  =
  \int f_{\beta\mid X,W}(b\mid x,w)\,dF_{X,W}(x,w),
\]
which obtains ~\eqref{fbeta_main}. This concludes the proof.

\subsection{Proof of Corollary \ref{cor_fD}}
\label{appendix_proof_cor}

Define the observable quantity
\[
  Z_i^*
  \equiv
  \frac{\mu_i'\widetilde{Y}_i}{\|\lambda_i\|}.
\]
In the regular design, $\widetilde{Y}_i = \widetilde{W}_iD_i$ holds 
unconditionally, so $\mu_i'\widetilde{Y}_i = \lambda_i'D_i$ and hence 
$Z_i^* = S_i'D_i$ for all observations. In the irregular design, 
$\widetilde{Y}_i = \det(X_i)\beta_i + \widetilde{W}_iD_i$, so the 
equality $Z_i^* = S_i'D_i$ holds only on the event $\{\det(X_i) = 0\}$.

\medskip
\noindent\emph{Regular design.} Since $S_i$ is measurable with respect to $(X_i,W_i)$ and by  Assumption~\ref{assm:A1}(ii), we have
$D_i\perp S_i$. Hence the conditional distribution
of $Z_i^* = S_i'D_i$ given $S_i=s$ is the distribution of the scalar $s'D_i$, and since $f_D$ exists, its conditional density is
\[
  f_{Z^*\mid S}(u\mid s)
  =
  \int_{\{\delta\,:\,s'\delta=u\}}
    f_D(\delta)\,d\sigma_{s,u}(\delta)
  =
  \mathcal{R}f_D(s,\,u).
\]
Since $Z_i^*$ and $S_i$ are observable, and since $S_i$ has a density bounded away
from zero on $\mathbb{S}^{d-1}$, the conditional density
$f_{Z^*\mid S}(\cdot\mid s)$ is identified for every $s\in\mathbb{S}^{d-1}$.

\medskip
\noindent\emph{Irregular design.} Since $\widetilde{Y}_i =
\widetilde{W}_iD_i$ on the event $\det(X_i)=0$, similar arguments as in the regular case obtain
\[
  f_{Z^*\mid S,\,\det(X_i)=0}(u\mid s)
  =
  \int_{\{\delta\,:\,s'\delta=u\}}
    f_D(\delta)\,d\sigma_{s,u}(\delta)
  =
  \mathcal{R}f_D(s,\,u).
\]

\medskip
\noindent\emph{Inversion.} In both cases, $\mathcal{R}f_D(s,u)$ is identified
for every $s$ in the support of $S_i$. By
Assumption~\ref{assm:A1_double}, this support is all of
$\mathbb{S}^{d-1}$, unconditionally in the regular design and conditionally on
$\det(X_i)=0$ in the irregular design. The additional assumption that the density of
$S_i$ is bounded away from zero uniformly on $\mathbb{S}^{d-1}$ therefore implies
that $\mathcal{R}f_D(s,u)$ is identified for all
$(s,u)\in\mathbb{S}^{d-1}\times\mathbb{R}$.

For each fixed $s\in\mathbb{S}^{d-1}$, the Fourier transform of
$\mathcal{R}f_D(s,\cdot)$ with respect to $u$ equals the characteristic function of
$D_i$ along the ray $\{\omega s:\omega\in\mathbb{R}\}$ by the Fourier slice theorem:
\[
  \int_{\mathbb{R}} e^{\ii \omega u}\,\mathcal{R}f_D(s,u)\,du
  =
  \varphi_D(\omega s),
  \qquad \omega\in\mathbb{R}.
\]
Since $\mathcal{R}f_D(s,u)$ is identified for all $(s,u)$, it follows that
$\varphi_D(\xi)$ is identified for all $\xi\in\mathbb{R}^d$. Because $D_i$ has density
$f_D$, the characteristic function $\varphi_D$ uniquely determines $f_D$. Hence
$f_D$ is identified almost everywhere.

\section{Identification with Discrete Regressors}
\label{ID_with_discrete_regressors}

In this section, we extend the identification results extend discrete $(X_i, W_i)$. The analysis clarifies that identification hinges on a geometric support condition governing the set of directions along which the characteristic function of $D_i$ is observed.

Consider the outcome equation \eqref{eq:outcome_stack_t} and let Assumption \ref{assm:A1} hold. Let $\tau_1(X_i)$ denote the annihilator used in the first step of the identification argument, and consider the transformed random variables, as in the main text,
\[
\widetilde Y_i := \tau_1(X_i) Y_i, \qquad
\widetilde W_i := \tau_1(X_i) W_i.
\]
Additionally, as in Assumption \ref{assm:A1_double}, let $\mu_i = \mu(X_i, W_i)$ be a measurable function taking values in $\mathbb{R}^T$, and define
\[
\lambda_i := \widetilde W_i' \mu_i \in \mathbb{R}^d, \qquad
S_i := \frac{\lambda_i}{\|\lambda_i\|} \in \mathbb{S}^{d-1}.
\]

The first step identifies \(\varphi_D\) on
\[
  \mathcal C_1
  =
  \{v\lambda: v\in\mathbb R,\ \lambda\in\operatorname{supp}(\lambda_i)\}.
\]
The deconvolution step requires \(\varphi_D\) only on the second-step frequency set
\[
  \mathcal C_2
  =
  \{W_i'\tau_2(X_i)'u:
    u\in\mathbb R^p,\ (X_i,W_i)\in\operatorname{supp}(X_i,W_i)\}.
\]
A sufficient condition for compatibility is \(\mathcal C_2\subseteq\overline{\mathcal C_1}\).
The full-support condition \(\operatorname{supp}(S_i)=\mathbb S^{d-1}\) is a simple
sufficient condition because it implies \(\overline{\mathcal C_1}=\mathbb R^d\).

If \(\varphi_D\) is real analytic, values on any nonempty open subset of
\(\mathbb R^d\) determine \(\varphi_D\) globally. 

\subsection*{Discrete regressors} 

In the irregular design, suppose that exact stayers exist with positive probability:
\[
\Pr(\det(X_i)=0) > 0,
\]
and that $\Pr(\lambda_i \neq 0 \mid \det(X_i)=0)=1$. Then $\phi_D$ is identified from the conditional moments
\[
\mathbb{E}\!\left[
\exp\{iv \mu_i' \widetilde Y_i\}
\ \middle|\ 
\widetilde W_i' \mu_i = \lambda,\ \det(X_i)=0
\right]
= \phi_D(v\lambda).
\]
If, in addition,
\[
\overline{\operatorname{supp}\left(
\frac{\widetilde W_i' \mu_i}{\|\widetilde W_i' \mu_i\|}
\ \middle|\ \det(X_i)=0
\right)} = \mathbb{S}^{d-1},
\]
then $\phi_D$ is identified on $\mathbb{R}^d$, and hence $f_\beta$ is point identified.

In the regular design, the annihilator eliminates the $\beta_i$ term for all observations, yielding
\[
\widetilde Y_i = \widetilde W_i D_i.
\]
In this case, if $\Pr(\lambda_i \neq 0)=1$ and
\[
\overline{\operatorname{supp}\left(
\frac{\widetilde W_i' \mu_i}{\|\widetilde W_i' \mu_i\|}
\right)} = \mathbb{S}^{d-1},
\]
then $\phi_D$ and $f_\beta$ are point identified.

\subsection*{Finite support and compatibility}

If \((X_i,W_i)\) has finite support, then for any measurable
\(\mu_i=\mu(X_i,W_i)\), the vector
\[
  \lambda_i=\widetilde W_i'\mu_i
\]
also has finite support. Hence the first step identifies \(\varphi_D\)
only on the finite union of rays
\[
  \mathcal C_1
  =
  \{v\lambda:v\in\mathbb R,\lambda\in\operatorname{supp}(\lambda_i)\}.
\]
When \(d\ge2\), this is not enough to identify the full distribution of
\(D_i\), because finitely many one-dimensional projections do not determine
a multivariate distribution.

For identification of \(f_\beta\), however, full knowledge of
\(\varphi_D\) on \(\mathbb R^d\) is stronger than necessary. The deconvolution
step requires \(\varphi_D\) only on the second-step frequency set
\[
  \mathcal C_2
  =
  \{W_i'\tau_2(X_i)'u:
    u\in\mathbb R^p,\ (X_i,W_i)\in\operatorname{supp}(X_i,W_i)\}.
\]
Therefore, a sufficient compatibility condition is
\[
  \mathcal C_2\subseteq\overline{\mathcal C_1}.
\]
Under finite support, this condition requires the finitely many second-step
frequency rays to be contained in the finitely many rays recovered from the
first step. If this compatibility fails, then the denominator needed in the
deconvolution step is not identified at some required frequencies, and
\(f_\beta\) is not point identified without additional restrictions.

Thus, finite support does not mechanically rule out identification of
\(f_\beta\), but it makes identification depend on a restrictive geometric
alignment between the first-step rays and the second-step frequency set.
By contrast, the full-support condition
\[
  \operatorname{supp}(S_i)=\mathbb S^{d-1}
\]
is a simple sufficient condition because it implies
\[
  \overline{\mathcal C_1}=\mathbb R^d,
\]
and therefore guarantees the compatibility condition automatically.

The scalar case \(d=1\) is exceptional. Since
\(\mathbb S^0=\{-1,1\}\), any nonzero \(\lambda_i\) generates the entire
real line:
\[
  \{v\lambda_i:v\in\mathbb R\}=\mathbb R.
\]
Hence no directional richness condition is needed when \(d=1\).

\subsection*{Conditions restoring identification for \(d\ge2\)}

Identification can be restored in three ways.

First, even under finite support, \(f_\beta\) is identified if the first-step
rays cover all second-step frequencies $\mathcal C_2\subseteq\overline{\mathcal C_1}$. This is a compatibility condition between the directions generated by the
annihilated equation and the frequencies required by the deconvolution step.

Second, if \((X_i,W_i)\) contains continuous variation such that the induced
directions $S_i$ have full support on \(\mathbb S^{d-1}\), then
\(\overline{\mathcal C_1}=\mathbb R^d\), so \(\varphi_D\) is identified
globally and the compatibility condition holds automatically.

Third, one may impose additional structure on \(D_i\). For example, if the
components of \(D_i\) are mutually independent, then the joint characteristic
function factorizes, and limited directional information may be sufficient
provided the relevant marginal characteristic functions are identified. If
\(D_i\) belongs to a parametric or elliptical family, finitely many projections
may identify the finite-dimensional parameters. Finally, if \(\varphi_D\) is
real analytic, then values of \(\varphi_D\) on any nonempty open subset of
\(\mathbb R^d\) determine \(\varphi_D\) globally by analytic continuation.
Continuity alone is not sufficient for this conclusion.

\section{Feasibility of the Cross-Validation Environment}
\label{appendix:cv_feasibility}

Define the instability measure for fold $k$:
\begin{equation}
\Gamma_k
=
\max\!\left\{
\max_{j,\,i\in\mathcal M_{-k}(\tau_x)}
\frac{1}{|\widehat\varphi_D^{(-k)}(u_j/X_i)|},\;
\max_{j,\,i\in\mathcal M_k(\tau_x)}
\frac{1}{|\widehat\varphi_D^{(k)}(u_j/X_i)|}
\right\},
\label{eq:gammak_def}
\end{equation}
where the maxima are taken over pairs $(j,i)$ for which the denominator
exceeds $\tau_{\mathrm{den}}$. Let $\rho_k$ denote the fraction of
frequency--mover pairs $(u_j, X_i)$ in the validation fold that are
trimmed, i.e., 
\begin{equation}
\rho_k = \frac{1}{J \cdot N_{m,k}} 
\sum_{j=1}^{J} \sum_{i \in \mathcal{M}_k(\tau_x)} 
\mathbf{1}\!\left\{ 
|\widehat\varphi_D^{(k)}(u_j/X_i)| \le \tau_{\mathrm{den}} 
\right\}.
\label{eq:rhok_def}
\end{equation}

Define also
the retained count at frequency $u_j$ in fold $k$:
\begin{equation}
n_k(u_j)
=
\sum_{i\in\mathcal M_k(\tau_x)}
\mathbf 1\!\left\{
\bigl|\widehat\varphi_D^{(k)}(u_j/X_i)\bigr|>\tau_{\mathrm{den}}
\right\}.
\label{eq:nk_uj_def}
\end{equation}
 
The CV environment is infeasible if any of the following conditions
holds:
\begin{enumerate}[label=(\alph*), nosep]
\item Empty evaluation set:
      $N_{m,k}=0$ for some fold $k$;
\item Excessive trimming:
      $K^{-1}\sum_{k=1}^K \rho_k > \rho_{\max}$
      (we use $\rho_{\max}=0.5$);
\item Exploding instability:
      $\max_{k=1,\dots,K}\Gamma_k > \Gamma_{\max}$
      (typical values: $\Gamma_{\max}=100$ for simulations;
       relaxed to $\Gamma_{\max} \approx 1000$--$2000$ for empirical data
       with heavier-tailed errors);
\item Degenerate frequency:
      $\min_{k,j} n_k(u_j) = 0$, i.e., for some fold $k$ and frequency
      $u_j$, all validation movers are trimmed.
\end{enumerate}
Condition~(d) is essential: when $n_k(u_j)=0$, the validation target
$\widehat m_N^{(k)}(u_j)$ collapses mechanically to zero. This
is not a meaningful holdout target, and the CV comparison at frequency
$u_j$ becomes degenerate.
 
If any of conditions (a)--(d) fails, the procedure stops and the user must
revise $\tau_x$ or $\tau_{\mathrm{den}}$ before proceeding. Otherwise, the
CV search is conducted over the candidate grid $\Theta_{h_x}\times\Theta_S$.

\section{Algorithms}
\label{appendix:algorithms}

\begin{algorithm}[H]
\caption{First-stage estimator, scalar irregular design $(T,p,q)=(1,1,0)$}
\label{alg:irregular_estimator_short}
\begin{algorithmic}[1]
\small
\Require Data $\{(Y_i,X_i)\}_{i=1}^N$; tuning parameters
         $(h_0,h_x,\tau_x)$; common parameters
         $(\tau_{\rm den},U_N,S)$.
\Ensure First-stage estimates $\{\widehat{m}_N(u_\ell)\}_{\ell=1}^L$.

\State Partition the sample:
\(
  \mathcal{S}_N(\tau_x) = \{i : |X_i| < \tau_x\},
  \quad
  \mathcal{M}_N(\tau_x) = \{i : |X_i| \ge \tau_x\}.
\)

\State Precompute for each mover: \qquad
\(
  \widetilde{\widetilde{Y}}_i = \frac{Y_i}{X_i},
  \qquad i \in \mathcal{M}_N(\tau_x).
\)

\State Construct the stayer weights once:
\[
  \omega^{(0)}_j
  =
  \frac{K_{h_0}(X_j)}
       {\displaystyle\sum_{s \in \mathcal{S}_N(\tau_x)} K_{h_0}(X_s)},
  \qquad j \in \mathcal{S}_N(\tau_x).
\]

\For{each frequency $u_\ell \in \{u_1,\ldots,u_L\}$}

  \State Evaluate the denominator at each mover-specific argument:
  \[
    \widehat{\varphi}_D\!\left(\frac{u_\ell}{X_i}\right)
    =
    \sum_{j \in \mathcal{S}_N(\tau_x)}
    \omega^{(0)}_j\,
    \exp\!\left(\ii\,\frac{u_\ell}{X_i} Y_j\right),
    \qquad i \in \mathcal{M}_N(\tau_x).
  \]
  If $u_\ell=0$, set $\widehat{\varphi}_D(0)=1$.

  \For{each mover $i \in \mathcal{M}_N(\tau_x)$}

    \State Estimate the numerator by kernel regression:
    \begin{align*}
      \widehat\varphi_{\widetilde{\widetilde{Y}} \mid X}(u_\ell \mid X_i)
      &=
      \sum_{k \in \mathcal{M}_N(\tau_x)}
      \omega^{(x)}_{k,i}
      \exp\!\Bigl(\ii\, u_\ell \widetilde{\widetilde{Y}}_k\Bigr),
      \\
      \omega^{(x)}_{k,i}
      &=
      \frac{K_{h_x}(X_k - X_i)}
           {\displaystyle\sum_{s \in \mathcal{M}_N(\tau_x)}
            K_{h_x}(X_s - X_i)}.
    \end{align*}

    \State Form the trimmed conditional ratio:
    \[
      \widehat{R}_N(u_\ell \mid X_i)
      =
      \frac{
        \widehat{\varphi}_{\widetilde{\widetilde{Y}} \mid X}(u_\ell \mid X_i)
      }{
        \widehat{\varphi}_D(u_\ell / X_i)
      }
      \;\mathbf{1}\!\left\{
        \bigl|\widehat{\varphi}_D(u_\ell / X_i)\bigr| > \tau_{\mathrm{den}}
      \right\}.
    \]

  \EndFor

  \State Average over movers:
  \(
    \widehat{m}_N(u_\ell)
    =
    \dfrac{1}{N_m}
    \displaystyle\sum_{i \in \mathcal{M}_N(\tau_x)}
    \widehat{R}_N(u_\ell \mid X_i),\quad N_m = |\mathcal{M}_N(\tau_x)|.
  \)

\EndFor
\end{algorithmic}
\end{algorithm}

\begin{algorithm}[H]
\caption{First-stage estimator, regular design $(T,p,q)=(2,1,0)$}
\label{alg:regular_estimator_short}
\begin{algorithmic}[1]
\small
\Require Data $\{(Y_{i1}, Y_{i2}, X_{i1}, X_{i2})\}_{i=1}^N$; $(h_X, h_S)$; 
         $(\tau_{\rm den}, U_N, S)$.
\Ensure First-stage estimates $\{\widehat{m}_N(u_\ell)\}_{\ell=1}^L$.

\For{$i = 1, \dots, N$}
  \State Precompute:
  \[
    \widetilde{\widetilde{Y}}_i
    = \frac{X_i'Y_i}{X_i'X_i},
    \quad
    \lambda_i = (X_{i2}, -X_{i1})',
    \quad
    \|\lambda_i\| = \sqrt{X_{i1}^2 + X_{i2}^2},
    \quad
    S_i = \frac{\lambda_i}{\|\lambda_i\|},
    \quad
    Z_i^* = \frac{X_{i2}Y_{i1}-X_{i1}Y_{i2}}{\|\lambda_i\|}.
  \]
\EndFor

\For{each frequency $u_\ell \in \{u_1,\ldots,u_L\}$}

  \For{$i = 1, \dots, N$}

    \State Estimate the numerator by kernel regression:
    \[
      \widehat\varphi_{\widetilde{\widetilde{Y}} \mid X}(u_\ell \mid X_i)
      =
      \sum_{k=1}^N
      \omega^{(X)}_{k,i}
      \exp\!\Bigl(\ii\, u_\ell \widetilde{\widetilde{Y}}_k\Bigr),
      \qquad
      \omega^{(X)}_{k,i}
      =
      \frac{K_{h_X}(X_k - X_i)}
           {\displaystyle\sum_{s=1}^N K_{h_X}(X_s - X_i)}.
    \]

    \State Set $\xi_{i\ell} = u_\ell X_i / (X_i'X_i)$ and estimate
           the denominator:
    \If{$\xi_{i\ell} = 0$} \Comment{occurs whenever $u_\ell = 0$}
      \State Set $\widehat\varphi_D(0) = 1$.
    \Else
      \State Let $s(\xi_{i\ell}) = \xi_{i\ell}/\|\xi_{i\ell}\|$ and
             compute:
      \[
        \widehat\varphi_D(\xi_{i\ell})
        =
        \sum_{j=1}^N
        \omega^{(S)}_j\bigl(s(\xi_{i\ell})\bigr)
        \exp\!\Bigl(\ii\,\|\xi_{i\ell}\| Z_j^*\Bigr),
        \qquad
        \omega^{(S)}_j(s)
        =
        \frac{K\!\left(\|S_j - s\| / h_S\right)}
             {\displaystyle\sum_{k=1}^N K\!\left(\|S_k - s\| / h_S\right)}.
      \]
    \EndIf

    \State Compute:
    \(
      \widehat{R}_N(u_\ell \mid X_i)
      =
      \frac{
        \widehat\varphi_{\widetilde{\widetilde{Y}} \mid X}(u_\ell \mid X_i)
      }{
        \widehat\varphi_D(\xi_{i\ell})
      }
      \;\mathbf{1}\!\left\{
        \bigl|\widehat\varphi_D(\xi_{i\ell})\bigr| > \tau_{\rm den}
      \right\}.
    \)

  \EndFor

  \State Average over all observations:
  \(
    \widehat{m}_N(u_\ell)
    = \dfrac{1}{N}
    \displaystyle\sum_{i=1}^N
    \widehat{R}_N(u_\ell \mid X_i).
  \)

\EndFor
\end{algorithmic}
\end{algorithm}

\begin{algorithm}[H]
\caption{Second-stage sieve minimum distance estimator}
\label{alg:smd_second_stage}
\begin{algorithmic}[1]
\small
\Require First-stage estimates $\{\widehat{m}_N(u_\ell)\}_{\ell=1}^L$;
         frequency grid $\{u_\ell\}_{\ell=1}^L$;
         quadrature weights $\{w_\ell\}_{\ell=1}^L$;
         sieve dimension $S$.
\Ensure Estimated density $\widehat{f}_\beta$.

\State Evaluate the Hermite sieve basis on the frequency grid:
\[
  z_\ell
  =
  z^S(u_\ell)
  =
  \sqrt{2\pi}
  \bigl(
    q_0(u_\ell),\;
    \ii\,q_1(u_\ell),\;
    \ldots,\;
    \ii^{S-1} q_{S-1}(u_\ell)
  \bigr)',
  \qquad \ell = 1, \ldots, L.
\]

\State Form the quadrature matrix and vector:
\[
  \widehat{\Omega}
  =
  \sum_{\ell=1}^L
  w_\ell\,\mathrm{Re}\!\bigl(z_\ell \bar z_\ell'\bigr),
  \qquad
  \widehat{V}
  =
  \mathrm{Re}\!\left(
    \sum_{\ell=1}^L
    w_\ell\,\overline{\widehat{m}_N(u_\ell)}\, z_\ell
  \right).
\]
$\widehat{\Omega}$ is real symmetric and positive definite when the
$L\times S$ matrix with rows $z^S(u_\ell)'$ has full column rank $S$.

\State Set $A = z^S(0)\in\mathbb{R}^S$ and compute:
\[
  \widehat{\pi}
  =
  \widehat{\Omega}^{-1}
  \!\left(
    \widehat{V}
    +
    \frac{1 - A'\widehat{\Omega}^{-1}\widehat{V}}
         {A'\widehat{\Omega}^{-1}A}
    \,A
  \right).
\]
By construction $A'\widehat{\pi}=1$, so the unit-mass constraint
$\varphi_S(0;\widehat\pi)=1$ is satisfied exactly.

\State Return:
\[
  \widehat{f}_\beta(b)
  =
  q^S(b)'\widehat{\pi}
  =
  \sum_{s=0}^{S-1}\widehat{\pi}_s\,q_s(b),
  \qquad b\in\mathbb{R}.
\]
\end{algorithmic}
\end{algorithm}

For each repetition \(r\), fold \(k\), and candidate
\(\theta=(h_x,S)\), define the fold-specific validation loss
\begin{equation}
\label{eq:cv_fold_loss}
\mathrm{CV}_k^{(r)}(\theta)
=
\sum_{j=1}^J
\varpi_j
\left|
\widehat m_{\mathrm{val}}^{(k,r)}(u_j)
-
\widehat\varphi_{\beta}^{(-k,r)}(u_j;\theta)
\right|^2 .
\end{equation}
The repeated cross-validation criterion is
\begin{equation}
\label{eq:cv_mean}
\overline{\mathrm{CV}}(\theta)
=
\frac{1}{n_{\mathrm{rep}}K}
\sum_{r=1}^{n_{\mathrm{rep}}}\sum_{k=1}^K
\mathrm{CV}_k^{(r)}(\theta).
\end{equation}
Let
\begin{equation}
\label{eq:cv_rep_score}
\mathrm{CV}^{(r)}(\theta)
=
\frac{1}{K}\sum_{k=1}^K
\mathrm{CV}_k^{(r)}(\theta),
\qquad
\widehat{\mathrm{se}}(\theta)
=
\frac{
\mathrm{SD}\!\left(
\{\mathrm{CV}^{(r)}(\theta)\}_{r=1}^{n_{\mathrm{rep}}}
\right)
}{
\sqrt{n_{\mathrm{rep}}}
}.
\end{equation}
Let
\[
\widetilde\theta
=
(\widetilde h_x,\widetilde S)
\in
\arg\min_{\theta\in\Theta_{h_x}\times\Theta_S}
\overline{\mathrm{CV}}(\theta),
\qquad
\mathrm{CV}^{*}
=
\overline{\mathrm{CV}}(\widetilde\theta)
+
\widehat{\mathrm{se}}(\widetilde\theta),
\]
and define the one-standard-error set
\begin{equation}
\label{eq:one_se_set}
\Theta_{\mathrm{1SE}}
=
\left\{
\theta\in\Theta_{h_x}\times\Theta_S:
\overline{\mathrm{CV}}(\theta)\le \mathrm{CV}^{*}
\right\}.
\end{equation}
The selected pair is
\begin{equation}
\label{eq:one_se_selection}
h_x^*
=
\max\{h_x:\exists S\text{ such that }(h_x,S)\in\Theta_{\mathrm{1SE}}\},
\qquad
S^*
=
\min\{S:(h_x^*,S)\in\Theta_{\mathrm{1SE}}\}.
\end{equation}

\begin{algorithm}[H]
\caption{Repeated \(K\)-fold cross-validation for \((h_x,S)\) selection}
\label{alg:pcv_tuning}
\begin{algorithmic}[1]
\footnotesize
\Require Data \(\{(Y_i,X_i)\}_{i=1}^N\); fixed
\((\tau_x,h_0,\tau_{\mathrm{den}})\); pilot bandwidth
\(g_{\mathrm{pilot}}\); grids \(\Theta_{h_x},\Theta_S\);
folds \(K\); repetitions \(n_{\mathrm{rep}}\); frequencies
\(\{u_j,\varpi_j\}_{j=1}^J\); thresholds
\(\Gamma_{\max},\rho_{\max}\).
\Ensure Selected tuning parameters \((h_x^*,S^*)\).

\For{\(r=1,\ldots,n_{\mathrm{rep}}\)}
  \State Generate a random \(K\)-fold partition
  \(\{\mathcal I_k^{(r)}\}_{k=1}^K\).

  \For{\(k=1,\ldots,K\)}
    \State Define training/validation samples
    \(\mathcal I_{-k}^{(r)}\), \(\mathcal I_k^{(r)}\), and their
    stayer/mover subsets.
    \State Estimate
    \(\widehat\varphi_D^{(-k,r)}\) and
    \(\widehat\varphi_D^{(k,r)}\).
    \State Compute feasibility diagnostics
    \(\Gamma_k^{(r)}\), \(\rho_k^{(r)}\), and
    \(n_k^{(r)}(u_j)\), \(j=1,\ldots,J\).
    \State Construct the validation target
    \(\widehat m_{\mathrm{val}}^{(k,r)}(u_j)\) using the
    pilot bandwidth \(g_{\mathrm{pilot}}\).
  \EndFor

  \If{\(\exists k:N_{m,k}^{(r)}=0\), or
  \(K^{-1}\sum_k\rho_k^{(r)}>\rho_{\max}\), or
  \(\max_k\Gamma_k^{(r)}>\Gamma_{\max}\), or
  \(\min_{k,j}n_k^{(r)}(u_j)=0\)}
    \State \textbf{stop} and revise the fixed stability parameters.
  \EndIf

  \For{\(\theta=(h_x,S)\in\Theta_{h_x}\times\Theta_S\)}
    \For{\(k=1,\ldots,K\)}
      \State Estimate
      \(\widehat m_N^{(-k,r)}(u_j;h_x)\) on the training sample.
      \State Fit the sieve approximation
      \(\widehat\varphi_\beta^{(-k,r)}(u_j;\theta)\).
      \State Compute \(\mathrm{CV}_k^{(r)}(\theta)\) using
      \eqref{eq:cv_fold_loss}.
    \EndFor
  \EndFor
\EndFor

\State Compute \(\overline{\mathrm{CV}}(\theta)\) and
\(\widehat{\mathrm{se}}(\theta)\) using
\eqref{eq:cv_mean}--\eqref{eq:cv_rep_score}.
\State Form \(\Theta_{\mathrm{1SE}}\) using \eqref{eq:one_se_set}.
\State Select \((h_x^*,S^*)\) using \eqref{eq:one_se_selection}.
\State Re-estimate the density on the full sample using \((h_x^*,S^*)\).
\end{algorithmic}
\end{algorithm}

\section{Application}
\label{sec:applicationestimation}

\subsection{Details on Estimation and Tuning Parameters}

\paragraph{Scalar irregular estimator, $(T, p, q) = (1, 1, 0)$.}
When a single pair of periods is used, equation~\eqref{eq:app_fd} reduces to 
$Y_i = X_i \beta_i + D_i$ with scalar $D_i$. The characteristic function 
$\varphi_D$ is estimated from stayers, i.e., observations with 
$|X_i| < \tau_x$, via kernel smoothing at $X_i = 0$ 
(Algorithm~\ref{alg:irregular_estimator_short}). The ratio 
$\tTwo{Y}_i = Y_i / X_i = \beta_i + D_i / X_i$ is formed for movers 
($|X_i| \geq \tau_x$), and the density $f_\beta$ is recovered by 
deconvolution through the sieve minimum distance procedure 
(Algorithm~\ref{alg:smd_second_stage}).

The stayer threshold $\tau_x$ has a direct economic interpretation 
in this application. Because the regressor 
$X_i = \log(\text{Exp}_{is}) - \log(\text{Exp}_{i,s-1})$ is the change in 
log expenditure, and $\log(1 + r) \approx r$ for small $r$, the condition 
$|X_i| < \tau_x$ selects households whose expenditure changed by less than 
approximately $100 \times \tau_x$ percent between adjacent survey waves. 
Stayers are thus households that experienced relatively stable 
expenditure levels across periods, while movers are households that 
experienced more substantial expenditure fluctuations.

Following the discussion in Section~\ref{subsec:implementation}, the 
stayer threshold $\tau_x$ and the denominator bandwidth $h_0$ are 
fixed, while the mover bandwidth $h_x$ and the sieve dimension $S$ are 
selected jointly by the cross-validation procedure of 
Algorithm~\ref{alg:pcv_tuning}. Specifically, the stayer threshold is 
set as
\begin{equation}\label{eq:threshold}
  \tau_x = c_\tau \cdot \tau_x^{\mathrm{ref}}, \qquad
  \tau_x^{\mathrm{ref}} = N^{-1/3}\min\!\left\{
    \mathrm{SD}(X_i),\;\frac{\mathrm{IQR}(X_i)}{1.34}
  \right\},
\end{equation}
with $c_\tau = 4$ for the full sample and $c_\tau = 3$ for the RPS$=1$ 
and RPS$=0$ subsamples. The reference threshold 
$\tau_x^{\mathrm{ref}}$ is a robust measure of spread scaled by 
$N^{-1/3}$, which widens as the sample size decreases; reducing 
$c_\tau$ for the smaller subsamples offsets this widening, maintaining 
a stable stayer proportion and ensuring an adequate mover count for the 
second-stage estimator. In the application, the resulting thresholds 
$\tau_x$ correspond to expenditure changes of roughly 5--10 percent, 
yielding stayer proportions between 7 and 10 percent across 
specifications (Table~\ref{tab:tuning_parameters}).

The denominator bandwidth is set proportional to the threshold, 
$h_0 = \tau_x$, following the discussion after~\eqref{hx_ref}. This 
choice ensures that the kernel estimate of $\varphi_D$ at $X_i = 0$ 
draws primarily on observations within the stayer region. The 
denominator trimming threshold is $\tau_{\mathrm{den}} = 10^{-4}$.

With these parameters fixed, the cross-validation search is conducted 
over a grid of candidate values $\theta = (h_x, S)$. The candidate grid 
for $h_x$ is centered on the Silverman reference bandwidth~\eqref{hx_ref} 
computed from the mover subsample and multiplied by 
$\{0.5, 0.75, 1.0, 1.5, 2.0\}$, and $S$ ranges over odd integers from 
$3$ to $15$. The cross-validation criterion $\mathrm{PCV}(\theta)$ is 
evaluated using $K = 5$ folds with $n_{\mathrm{rep}} = 20$ repetitions 
for stability, and the selected tuning parameters are 
$\widehat\theta = \arg\min_{\theta \in \Theta} \mathrm{PCV}(\theta)$ 
subject to a one-standard-error rule that favors parsimony. We apply 
this estimator separately to the 2000--2001 and 2001--2002 differenced 
data.

As reported in Table~\ref{tab:tuning_parameters}, cross-validation 
selects $S = 3$ in all irregular specifications, reflecting the 
relatively smooth unimodal densities recovered from these data. The 
selected mover bandwidths $h_x$ vary across samples and periods, 
ranging from $0.120$ to $0.303$.

\paragraph{Regular estimator, $(T,p,q)=(2,1,0)$.}
When both differenced periods are stacked, the regressor matrix 
$X_i=(X_{i1},X_{i2})'\in\mathbb{R}^2$ has $T=2>p=1$, placing the problem in the 
regular design. The first-step transformation is the rank-one annihilator
\[
  \tau_1(X_i)
  =
  \begin{pmatrix} X_{i2} & -X_{i1} \\ 0 & 0 \end{pmatrix},
\]
which satisfies $\tau_1(X_i)X_i=0$ and yields the scalar first-step equation 
$Y_i^*=X_{i2}Y_{i1}-X_{i1}Y_{i2}=\lambda_i'D_i$ with $\lambda_i=(X_{i2},-X_{i1})'$ 
for every observation. The characteristic function 
$\varphi_D$ is then estimated from the full sample by smoothing over the 
normalized direction $S_i=\lambda_i/\|\lambda_i\|\in\mathbb{S}^1$ 
(Algorithm~\ref{alg:regular_estimator_short}). In the second step, the left inverse 
$\tau_2(X_i)=(X_i'X_i)^{-1}X_i'$ gives $\tTwo{Y}_i=\beta_i+\tTwo{W}_iD_i$, and 
deconvolution recovers $f_\beta$ via the sieve minimum distance procedure 
(Algorithm~\ref{alg:smd_second_stage}).

The tuning parameters for the regular design are selected by an analogous 
application of Algorithm~\ref{alg:pcv_tuning}. The directional bandwidth $h_S$, which governs the 
estimation of $\varphi_D$ by smoothing over $S_i\in\mathbb{S}^1$, plays the same 
role as $h_0$ in the irregular design: it is fixed by a rule of thumb. A natural default is Silverman's rule adapted to directional data on $\mathbb{S}^1$. 
The bandwidth $h_X$ for the numerator kernel regression and the sieve dimension $S$ 
are then selected jointly by cross-validation, with $h_X$ replacing $h_x$ in the 
candidate grid and the cross-validation loops of Algorithm~\ref{alg:pcv_tuning}.

The candidate grid for $h_X$ is centered on Silverman's bivariate reference 
bandwidth $h_X^{\mathrm{ref}}$ applied to $\{X_i\}_{i=1}^N$ and multiplied by 
$\{0.5, 0.75, 1.0, 1.5, 2.0\}$, while $S$ ranges over odd integers from $3$ to 
$15$. The weighting function $\nu$ is the standard normal density.

The tuning parameters selected by the CV algorithm can be found in Table \ref{tab:tuning_parameters}. These correspond to the tuning parameters used to compute the estimates in Table \ref{tab:moments_combined}.

\begin{table}[ht]
\centering
\caption{Selected tuning parameters and stayer shares}
\label{tab:tuning_parameters}
\renewcommand{\arraystretch}{1.15}
\begin{tabular}{llcc}
\toprule
Sample / Period & Bandwidth & $S$ & Stayers (\%) \\
\midrule
\multicolumn{4}{l}{\textit{Regular estimator (pooled 2000--2002)}} \\[2pt]
Full sample & $h_S = 0.093,\; h_X = 0.234$ & $3$ & --- \\
RPS $=1$    & $h_S = 0.106,\; h_X = 0.252$ & $3$ & --- \\
RPS $=0$    & $h_S = 0.111,\; h_X = 0.281$ & $3$ & --- \\
\addlinespace
\multicolumn{4}{l}{\textit{Irregular estimator: Full sample}} \\[2pt]
2000--2001 & $h_x = 0.276$ & $3$ & $8.1$ \\
2001--2002 & $h_x = 0.120$ & $3$ & $7.1$ \\
\addlinespace
\multicolumn{4}{l}{\textit{Irregular estimator: RPS $=1$}} \\[2pt]
2000--2001 & $h_x = 0.159$ & $3$ & $9.8$ \\
2001--2002 & $h_x = 0.262$ & $3$ & $9.1$ \\
\addlinespace
\multicolumn{4}{l}{\textit{Irregular estimator: RPS $=0$}} \\[2pt]
2000--2001 & $h_x = 0.160$ & $3$ & $9.6$ \\
2001--2002 & $h_x = 0.303$ & $3$ & $9.5$ \\
\bottomrule
\end{tabular}
\vspace{0.5em}

\footnotesize
\emph{Notes:} $S$ denotes the sieve dimension. Bandwidths $(h_S, h_X)$ correspond to the regular design, while $h_x$ denotes the scalar bandwidth in the irregular design. The regular estimator uses the full sample of $N = 1{,}358$ households in both stages. The irregular estimator is applied to three samples: the full sample ($N = 1{,}358$), RPS$\,{=}\,1$ ($N = 705$), and RPS$\,{=}\,0$ ($N = 653$). Common tuning settings across both designs: $\tau_{\mathrm{den}} = 10^{-4}$, $U_N = 4$, stabilized CV with $K = 5$ folds, $n_{\mathrm{rep}} = 20$ repetitions, pilot rate $N^{-1/7}$, $c_{\mathrm{pilot}} = 2$, and one-SE rule. The candidate grid for $S$ is $\{3,5,7,9,11,13,15\}$; in all six runs the one-SE rule selects the smallest sieve dimension within one standard error of the minimum cross-validation criterion.
\end{table}

\begin{remark}
The full-sample and subsample estimates in Table~\ref{tab:moments_combined} are not constrained to satisfy a mixture relationship: each is obtained by running the estimator independently on its own sample, with tuning parameters selected separately via cross-validation within that sample. The set of observations contributing to the directional smoother for $\widehat\varphi_D$ varies across runs, and the sample-specific bandwidths $(h_S, h_X)$ alter the shape of the denominator and numerator characteristic-function estimates. Moreover, the first-stage estimator $\widehat m_N(u)$ is an average of deconvolution ratios $\widehat\varphi_{\tTwo Y\mid X}(u\mid X_i)\big/\widehat\varphi_D\!\bigl(u\, X_i/(X_i'X_i)\bigr)$, each involving a sample-specific denominator; this nonlinearity in the data implies that $\widehat m_N(u)$ for the full sample is not a weighted average of its subsample counterparts. The nonlinearity of the second-stage sieve minimum-distance step then propagates these differences into the final density estimates. Consequently, neither $\widehat{\mathbb{E}}[\beta_i]$ nor $\widehat{\mathrm{Var}}[\beta_i]$ for the full sample is constrained to lie between the corresponding RPS$\,{=}\,1$ and RPS$\,{=}\,0$ values; a mixture identity is not enforced by the estimator. In the present realization, the full-sample point estimates do fall inside the subsample range ($\widehat{\mathbb{E}}[\beta_i]_{\mathrm{Full}} = 0.677$ between $0.625$ and $0.698$; $\widehat{\mathrm{Var}}[\beta_i]_{\mathrm{Full}} = 0.591$ between $0.586$ and $0.635$), but this is a feature of the particular sample rather than a property of the procedure, and need not be preserved across alternative samples or bootstrap draws.
\end{remark}

\subsection{Preliminary Assessment of the Support of $\beta_i$}
\label{subsec:beta_support}

Implementation of the sieve minimum distance estimator requires specification of a grid over which the density $f_\beta$ is evaluated. To inform this choice, we develop a simple diagnostic that provides indicative bounds on the support of $\beta_i$ using observable data.

Consider the scalar irregular design $Y_i = X_i \beta_i + D_i$. For stayers satisfying $|X_i| < \tau_x$, we have $Y_i \approx D_i$, so the distribution of $Y_i$ among stayers approximates the distribution of $D_i$. For movers with $|X_i| \geq \tau_x$, the individual effect satisfies $\beta_i = (Y_i - D_i)/X_i$, but $D_i$ is unobserved.

We proceed in two stages. First, we bound the disturbance using quantiles of the stayer outcome distribution, which provides robustness to outliers:
\begin{equation}
D^{\mathrm{lo}} = Q_{\alpha}(Y_i \mid |X_i| < \tau_x), \qquad D^{\mathrm{hi}} = Q_{1-\alpha}(Y_i \mid |X_i| < \tau_x),
\end{equation}
where $Q_\alpha(\cdot)$ denotes the $\alpha$-quantile. We set $\alpha = 0.05$ in the application.

Second, for each mover $i$, we compute bounds on $\beta_i$ by considering the extreme values of $D_i$ within the interval $[D^{\mathrm{lo}}, D^{\mathrm{hi}}]$:
\begin{equation}
\beta_i^L = \frac{Y_i - D^{\mathrm{hi}}}{X_i}, \qquad \beta_i^U = \frac{Y_i - D^{\mathrm{lo}}}{X_i} \qquad \text{for } X_i > 0,
\end{equation}
with the inequalities reversed when $X_i < 0$. This yields two vectors of length $N_m$: the collection of lower bounds $\{\beta_i^L\}_{i \in \mathcal{M}}$ and upper bounds $\{\beta_i^U\}_{i \in \mathcal{M}}$.

Even after trimming, individual $\beta$ bounds may exhibit extreme values due to outlier outcomes $Y_i$ or small values of $|X_i|$ near the threshold $\tau_x$. To address this, we summarize the bound distributions using their interquartile range:
\begin{equation}
\text{Typical range} = \left[Q_{0.25}(\beta^L), \; Q_{0.75}(\beta^U)\right].
\end{equation}
This doubly-robust summary trims both the stayer-derived $D$ bounds and the mover-derived $\beta$ bounds.

Using the $Q(5\%, 95\%)$ bounds on $D_i$, the range for the expenditure elasticity is approximately $[-1.8, 3]$. Table \ref{tab:beta_support} presents the results. These bounds assume $\mathrm{supp}(D_i \mid X_i \neq 0) \subseteq \mathrm{supp}(D_i \mid X_i = 0)$. The diagnostic should therefore be interpreted as indicative rather than as sharp identification bounds. The bounds in the table are used for display and for the internal normalization step described in Section~\ref{sec:estimation}.

\begin{table}[htbp]
\centering
\caption{Indicative Bounds on the Support of $\beta_i$}
\label{tab:beta_support}
\begin{tabular}{lcc}
\toprule
 & \multicolumn{2}{c}{Typical Range} \\
\cmidrule(lr){2-3}
Subsample / Year Pair & $Q_{.25}(\beta^L)$ & $Q_{.75}(\beta^U)$ \\
\midrule
RPS $= 1$, $\Delta$ 2000--2001 & $-1.70$ & $2.77$ \\
RPS $= 1$, $\Delta$ 2001--2002 & $-1.54$ & $2.12$ \\
RPS $= 0$, $\Delta$ 2000--2001 & $-1.81$ & $2.99$ \\
RPS $= 0$, $\Delta$ 2001--2002 & $-1.65$ & $2.43$ \\
\midrule
Combined & $-1.81$ & $2.99$ \\
Implemented & $-3$ & $3$ \\
\bottomrule
\end{tabular}

\smallskip
\parbox{0.85\textwidth}{\footnotesize\textit{Notes:} Bounds are computed using the two-stage quantile trimming procedure described in Section~\ref{subsec:beta_support}. The disturbance $D_i$ is bounded by the 5th and 95th percentiles of $Y_i$ among stayers ($|X_i| < \tau_x$). For each mover, individual bounds $[\beta_i^L, \beta_i^U]$ are computed by evaluating $\beta_i = (Y_i - D_i)/X_i$ at the extreme values of the $D$ interval. The ``Typical Range'' reports $Q_{.25}(\beta^L)$ and $Q_{.75}(\beta^U)$, the 25th percentile of lower bounds and 75th percentile of upper bounds across movers. The ``Combined'' 
row reports the envelope of the subgroup ranges, i.e., the minimum of the 
$Q_{.25}(\beta^L)$ values and the maximum of the $Q_{.75}(\beta^U)$ values 
across all subsamples, rather than the pooled quantiles. The "Combined" row reports the range used in the application.}
\end{table}

\FloatBarrier

\bibliographystyle{plainnat}
\bibliography{BotosaruPowell_refs}

\end{document}